%% file: paper.tex
\newif\iftechrep
\techreptrue 

\documentclass[a4paper,UKenglish]{article}
\input{prelude_paper}

\title{Strict Ideal Completions of the Lambda Calculus%
  \iftechrep %
  \newline{\Large(Technical Report)}
  \fi%
}
\date{May 2018}

\author{Patrick Bahr\\IT University of Copenhagen, Denmark\\\texttt{paba@itu.dk}}

\begin{document}

\maketitle

\begin{abstract}
  The infinitary lambda calculi pioneered by Kennaway et al.\ extend
  the basic lambda calculus by metric completion to infinite terms and
  reductions. Depending on the chosen metric, the resulting infinitary
  calculi exhibit different notions of \emph{strictness}. To obtain
  infinitary normalisation and infinitary confluence properties for
  these calculi, Kennaway et al.\ extend $\beta$-reduction with
  infinitely many `$\bot$-rules', which contract \emph{meaningless
    terms} directly to $\bot$. Three of the resulting \emph{B\"ohm
    reduction} calculi have unique infinitary normal forms
  corresponding to B\"ohm-like trees.
  
  In this paper we develop a corresponding theory of infinitary lambda
  calculi based on ideal completion instead of metric completion. We
  show that each of our calculi conservatively extends the
  corresponding metric-based calculus. Three of our calculi are
  infinitarily normalising and confluent; their unique infinitary normal
  forms are exactly the B\"ohm-like trees of the corresponding
  metric-based calculi. Our calculi dispense with the infinitely many
  $\bot$-rules of the metric-based calculi. The fully non-strict
  calculus (called $111$) consists of only $\beta$-reduction, while
  the other two calculi (called $001$ and $101$) require two
  additional rules that precisely state their strictness properties:
  $\lambda x.\bot \to \bot$ (for $001$) and $\bot\,M \to \bot$ (for
  $001$ and $101$).
\end{abstract}

\section{Introduction}
\label{sec:introduction}

In their seminal work on infinitary lambda calculus, Kennaway et
al.~\cite{kennaway97tcs} study different infinitary variants of the
lambda calculus, which are obtained by extending the ordinary lambda
calculus by means of metric completion. Different variants of the
calculus are obtained by choosing a different metric. The `standard'
metric on terms measures the distance between two terms depending on
how deep one has to go into the term structure to distinguish two
terms.  For example the term $x\,y$ is closer to the term $x\,z$ than
to the term $x$, because in the former case both terms are
applications whereas in the latter case one term is an application and
the other is a variable.

The different metric spaces arise by changing the way in which we
measure depth. Kennaway et al.~\cite{kennaway97tcs} indicate this
using a binary triple $abc$ with $a,b,c \in \set{0,1}$, where $a = 0$
indicates that we do not count lambda abstractions when calculating the
depth, and $b = 0$ or $c = 0$ indicates that we do not count the left
or the right side of applications, respectively. More intuitively
these three parameters can be interpreted as indicating
\emph{strictness}. For example, $a = 0$ indicates that lambda
abstraction is strict, i.e.\ if $M$ diverges, then so does
$\lambda x. M$.

Since the set of infinite terms is constructed from the set of finite
terms by means of metric completion, each calculus has a different
universe of terms, as well as a different mode of convergence, which
is based on the topology induced by the metric. For instance, from the
lambda term $N = (\lambda x.x\,x\,y) (\lambda x.x\,x\,y)$, we can
derive the infinite reduction $N \to N\,y \to N \, y\, y\to \dots$. In
the fully non-strict calculus, where $abc = 111$, this reduction
converges to the infinite term $M = \dots y\, y\, y$ (i.e.\ $M$
satisfies $M = M\,y$). By contrast, in the calculus $101$, which is
strict on the left-hand side of every application, this reduction does
not converge. In fact, $M$ is not even a valid term in the $101$
calculus.

In order to deal with divergence as exemplified for the $101$ calculus
above, Kennaway et al.~\cite{kennaway97tcs} extend standard
$\beta$-reduction to \emph{B\"ohm reduction} by adding rules of the
form $M \to \bot$, for each term $M$ that causes divergence
such as the term $N$ in the $101$ calculus. The resulting $001$,
$101$, and $111$ calculi based on B\"ohm reduction have unique normal
forms, which correspond to the well-known \emph{B\"ohm
  Trees}~\cite{levy78phd,barendregt84book}, \emph{Levy-Longo
  Trees}~\cite{levy76tcs,longo83apal} and \emph{Berarducci
  Trees}~\cite{berarducci96laa}, respectively.

In this paper, we introduce infinitary lambda calculi that are based on ideal
completion instead of metric completion with the goal of directly
dealing with diverging terms without the need for additional reduction
rules that contract diverging terms immediately to $\bot$. To this
end, we devise for each metric of the calculi of Kennaway et
al.~\cite{kennaway97tcs} a corresponding partial order with the
following property: Ideal completion of the set of finite lambda terms
yields the same set of infinite lambda terms as the corresponding
metric completion (Section~\ref{sec:ideal-completion}). We also find a
strong correspondence between the modes of convergence induced by
these structures: Each ideal completion yields a complete semilattice
structure, which means that the \emph{limit inferior} is always
defined. We show that this limit inferior is a conservative extension
of the limit in the corresponding metric completion in the sense that
both modes of convergence coincide on total lambda terms, i.e.\ terms
without $\bot$ (Section~\ref{sec:ideal-completion}).

Based on these partial order structures we define infinitary lambda
calculi by a straightforward instantiation of transfinite abstract
reduction systems~\cite{bahr10rta}. We find that the ideal completion
calculi form a conservative extension of the metric completion calculi
of Kennaway et al.~\cite{kennaway97tcs}
(Section~\ref{sec:transf-reduct}). Moreover, in analogy to
Blom~\cite{blom04rta} and Bahr~\cite{bahr10rta2}, we find that the
differences between the ideal completion approach and the metric
completion approach are compensated for by adding $\bot$-rules to the
metric calculi in the style of Kennaway et al.~\cite{kennaway99jflp}
(Section~\ref{sec:bohm-like-tree}). Finally, we also show infinitary
normalisation for our ideal completion calculi and infinitary
confluence for the $001$, $101$, and $111$ calculi
(Section~\ref{sec:bohm-like-tree}). However, in order to obtain
infinitary confluence for $001$ and $101$, we need to extend
$\beta$-reduction with two additional rules that precisely capture the
strictness properties of these calculi: $\lambda x.\bot \to \bot$ (for
$001$) and $\bot\,M \to \bot$ (for $001$ and $101$). In
Section~\ref{sec:related-work}, we give a brief overview of related
work.


\iftechrep %
We have abridged and in some cases omitted proofs in the main body of
the paper. The corresponding full proofs are collected in
Appendix~\ref{sec:full-proofs}.
\fi %

\section{The Metric Completion}
\label{sec:metric-completion}

In this section, we introduce infinite lambda terms as the result of
metric completion of the set of finite lambda terms. Before we get
started, we introduce some basic notions about transfinite sequences
and lambda terms. We presume basic familiarity with metric spaces and
ordinal numbers.

A \emph{sequence} over a set $A$ of length $\alpha$ is a mapping from
an ordinal $\alpha$ into $A$ and is written as
$(a_\iota)_{\iota<\alpha}$, which indicates the mapping $\iota \mapsto
a_\iota$; the notation $\len{(a_\iota)_{\iota<\alpha}}$ denotes the
length $\alpha$ of $(a_\iota)_{\iota<\alpha}$. If $\alpha$ is a limit
ordinal, then $(a_\iota)_{\iota<\alpha}$ is called \emph{open};
otherwise it is called \emph{closed}. If $(a_\iota)_{\iota<\alpha}$ is
finite, it is also written as $\seq{a_0,\dots,a_{\alpha-1}}$; in
particular, $\emptyseq$ denotes the empty sequence. We write $S\concat
T$ for the \emph{concatenation} of two sequences $S$ and $T$; $S$ is
called a (\emph{proper}) \emph{prefix} of $T$, denoted $S \le T$
(resp.\ $S < T$) if there is a (non-empty) sequence $S'$ such that
$S\concat S' = T$. The unique prefix of a sequence $S$ of length
$\beta \le \len S$ is denoted by $\prefix S \beta$.

We consider lambda terms with an additional symbol $\bot$; the
resulting set of \emph{lambda terms} $\plam$ is inductively defined by
the following grammar:%
\[
M,N ::= \bot \; |\; x \; | \; \lambda x. M \; | \; M N
\]
where $x$ is drawn from a countably infinite set $\calV$ of variable
symbols. The set of \emph{total lambda terms} $\lam$ is the subset of
lambda terms in $\plam$ that do not contain $\bot$.  Occurrences of a
variable $x$ in a subterm $\lambda x. M$ are called \emph{bound};
other occurrences are called \emph{free}. We use the notation $\rename
M x y$ to replace all free occurrences of the variable $x$ in $M$ with
the variable $y$. We use finite sequences over $\set{0,1,2}$, called
\emph{positions}, to point to subterms of a lambda term; we write
$\allpos$ for the set of all positions. For each $M \in \plam$,
$\pos{M}$ denotes the set of positions of $M$ (excluding `$\bot$'s)
recursively defined as follows: $\pos{\bot} = \emptyset$, $\pos{x} =
\set{\emptyseq}$, $\pos{M_1\,M_2} = \set{\emptyseq} \cup \setcom{\seq
  i \concat p}{i\in\set{1,2},p\in\pos{M_i}}$, and $\pos{\lambda x . M}
= \set{\emptyseq} \cup \setcom{\seq 0 \concat p}{p\in\pos{M}}$.

A \emph{conflict}~\cite{kennaway97tcs} between two lambda terms $M, N$
is a position $p \in \pos{M} \cup \pos{N}$ such that:
\begin{enumerate*}[(a)]
\item if $p = \emptyseq$, then $M$ and $N$ are not identical
  variables, not both $\bot$, not both applications, and not both
  abstractions;
\item if $p = \seq i \concat q$ and $i\in\set{1,2}$, then $M =
  M_1M_2$, $N = N_1 N_2$, and $q$ is a conflict of $M_i$ and $N_i$;
\item if $p = \seq 0 \concat q$, then $M= \lambda x.M'$, $N = \lambda
  y.N'$, and $q$ is a conflict of $\rename {M'} x z$ and $\rename {N'}
  y z$, where $z$ is a fresh variable occurring neither in $M$ nor
  $N$.
\end{enumerate*}
The terms $M$ and $N$ are said to be \emph{$\alpha$-equivalent} if
they have no conflicts. By convention we identify $\alpha$-equivalent
terms (i.e.\ $\plam$ and $\lam$ are assumed to be quotients by
$\alpha$-equivalence).

\begin{definition}
  Given a triple $\ol a = a_0a_1a_2 \in \set{0,1}^3$, called
  \emph{strictness signature}, a position is called
  \emph{$\ola$-strict} if it is of the form $q \concat\seq i$ with
  $a_i =0$; otherwise it is called \emph{$\ola$-non-strict}. If $\ola$
  is clear from the context, we only say \emph{strict} resp.\
  \emph{non-strict}.
\end{definition}
That is, a strictness signature indicates strictness by $0$ and
non-strictness by $1$. For example, if $\ola = 011$, lambda
abstraction is strict, and application is non-strict both from the
left and the right. We shall see what this means shortly: Following
Kennaway et al.~\cite{kennaway97tcs}, we derive, from a strictness
signature $\ola$, a depth measure $\adepth{\cdot}$, which counts the
number of non-strict, non-empty prefixes of a position. From this
depth measure we then derive a corresponding metric $\dda$ on lambda
terms.
\begin{definition}
  Given a strictness signature $\ola$, the \emph{$\ol a$-depth} of a
  position $p$, denoted $\adepth p$, is recursively defined as
  $\adepth\emptyseq = 0$ and
  $\adepth{q\concat\seq{i}} = \adepth q + a_i$. The $\ola$-distance
  $\dda(M,N)$ between two terms $M,N \in \plam$ is $0$ if $M$ and $N$
  are $\alpha$-equivalent and otherwise $2^{-d}$, where $d$ is the
  least number satisfying $d = \adepth p$ for some conflict $p$ of $M$
  and $N$.
\end{definition}
Kennaway et al.~\cite{kennaway97tcs} showed that the pair
$(\plam,\dda)$ forms an ultrametric space for any $\ola$. Intuitively,
the consequence of the definition of these metric spaces is that
sequences of terms, such as the sequence
$N , N\,y , N\, y\, y , \dots$, only converge if conflicts between
consecutive terms are guarded by an increasing number of non-strict
positions. In the example, conflicts between consecutive terms are
guarded by an increasing stack of applications to $y$. If $a_1 = 1$,
these applications correspond to non-strict positions, and thus the
sequence converges. However, if $a_1 = 0$, the sequence does not
converge.


We turn now to the metric completion. To facilitate later definitions
and to illustrate the resulting structures, we use a partial function
representation in the form of lambda trees taken from
Blom~\cite{blom04rta}, which will serve as mediator between metric
completion and ideal completion.\footnote{In %
  \iftechrep%
  Appendix~\ref{sec:direct-approach} %
  \else%
  the companion report~\cite{techrep} %
  \fi%
  we give a direct proof of the correspondence between metric and
  ideal completion based on the meta theory of Majster-Cederbaum and
  Baier~\cite{majster-cederbaum96tcs}.} A lambda tree is a (possibly
infinite) labelled tree where a label $\lambda$ indicates abstraction
and $@$ indicates application; labels in $\calV$ indicate free
variables and a label $p \in \allpos$ indicates a variable that is
bound by an abstraction at position $p$. There is no label
corresponding to $\bot$, which instead is represented as a `hole' in
the tree. We write $\dom{f}$ to denote the domain of a partial
function $f$, and $f(p) \simeq g(q)$ to indicate that the partial
functions $f$ and $g$ are either both undefined or have the same value
at $p$ and $q$, respectively.
\begin{definition}
  \label{def:lambda}
  A \emph{lambda tree} is a partial function $t\colon \allpos \pfunto
  \lamsig$ with $\lamsig =
  \set{\lambda,@}\uplus\allpos\uplus\calV$ so that\\[5pt]
  \begin{enumerate*}[(a),mode=unboxed]
    \begin{tabular}{rl@{\quad}l@{\quad}l}
    \item\label{item:lambda1} &$p \concat \seq 0 \in \dom{t}$ & $\implies$ & $t(p) = \lambda$,
      \\
    \item\label{item:lambda2} &$p \concat \seq 1 \in \dom{t}$ or $p \concat \seq 2 \in
      \dom{t}$ &$\implies$& $t(p) = @$, and\\
    \item\label{item:lambda3} &$t(p) = q$, where $q\in \allpos$& $\implies$& $q \le p$ and $t(q) =
      \lambda$.
    \end{tabular}
  \end{enumerate*}\\[7pt]
  As one would expect, the domain $\dom t$ of a lambda tree $t$ is
  prefix closed.
  
  The set of all lambda trees is denoted $\iptree$. The set of
  \emph{$\bot$-positions} in $t$, denoted $\domBot{t}$, is the
  smallest set that satisfies the following:
  \begin{enumerate*}[(a)]
  \item $\emptyseq \nin \dom t$ implies $\emptyseq \in \domBot t$;
  \item $t(p) = \lambda$, $p\concat \seq 0 \nin \dom t$ implies $p \concat
    \seq 0 \in \domBot t$; and
  \item $t(p) = @$, $p\concat\seq i \nin \dom t$, $i \in \set{1,2}$
    implies $p\concat\seq i \in \domBot t$.
  \end{enumerate*}
  A lambda tree $t$ is called \emph{total} if $\domBot t$ is empty.
  The set of all total lambda trees is denoted $\itree$.  A lambda
  tree $t$ is called \emph{finite} if $\dom{t}$ is a finite set. The
  set of all finite (total) lambda trees is denoted $\ptree$
  (respectively $\tree$). A \emph{renaming} of a lambda tree $t$ is a
  lambda tree $s$ such that there is a bijection
  $f \fcolon \calV \funto \calV$ with the following properties:
  $s(p) = t(p)$ if $t(p) \in \lamsig\setminus\calV$, $s(p) = f(t(p))$
  if $t(p) \in \calV$, and otherwise $s(p)$ is undefined.
\end{definition}
In order to avoid confusion, we use upper case letters $M, N$ for
lambda terms and lower case letters $s,t,u$ for lambda trees. Below,
we give a bijection from lambda terms to finite lambda trees that
should help illustrate the idea behind lambda trees. At the heart of
this bijection are the following constructions based on
Blom~\cite{blom04rta}:
\begin{definition}
  \label{def:treeConstr}
  Given lambda trees $t, t_1, t_2 \in \iptree$ and a variable
  $x\in\calV$, let $\bot$, $\vartree x$, $\abstree x t$ and $t_1 \,
  t_2$ be partial functions of type $\allpos \pfunto \lamsig$ defined
  by their graph as follows:\vspace{-5pt}
  \begin{align*}
    \bot &= \emptyset \qquad \vartree x = \set{(\emptyseq,x)}\\
    \abstree x t &= \set{(\emptyseq,\lambda)}
    \begin{aligned}[t]
      &\cup \setcom{(\seq 0 \concat p,l)}{l\in
        \set{\lambda,@}\uplus\calV\setminus\set{x}, (p,l)\in t}\\
      &\cup \setcom{(\seq 0 \concat p,\seq 0\concat q)}{q\in \allpos, (p,q)\in t}
      \cup \setcom{(\seq 0 \concat p,\emptyseq)}{(p,x)\in t}
    \end{aligned}\\
    t_1\, t_2 &= \set{(\emptyseq,@)}
    \begin{aligned}[t]
      &\cup \setcom{(\seq i \concat p,l)}{i\in\set{1,2}, l\in
        \set{\lambda,@}\uplus\calV, (p,l)\in t_i}\\
      &\cup \setcom{(\seq i \concat p,\seq i\concat q)}{i\in\set{1,2},
        q\in \allpos, (p,q)\in t_i}\
    \end{aligned}
  \end{align*}
\end{definition}
One can easily check that each of the above four constructions yields
a lambda tree, where $\bot$ is the empty lambda tree, $\vartree x$ the
lambda tree consisting of a single free variable $x$, $\abstree x t$
is a lambda abstraction over $x$ with body $t$, and $t_1\, t_2$ is an
application of $t_1$ to $t_2$.  The following translation of lambda
terms to finite lambda trees illustrates the use of these
constructions:
\begin{definition}
  \label{def:lamtree}
  Let $\lamtree\cdot \colon \plam \to \ptree$ be defined recursively
  as follows:\vspace{-5pt}
  \begin{gather*}
    \lamtree{\bot} = \bot \hspace{1.1cm}%
    \lamtree{\lambda x.M} = \abstree x \lamtree M\hspace{1.1cm}%
    \lamtree{x}    = \vartree x \hspace{1.1cm}%
    \lamtree{M\, N} = \lamtree M\, \lamtree N 
  \end{gather*}
\end{definition}


One can easily check that $\lamtree\cdot\colon \plam\to\ptree$ is
indeed a bijection, which, if restricted to $\lam$, is a bijection
from $\lam$ to $\tree$. Moreover, one can show that each
$t \in \iptree$ with some $\seq i\concat p \in \dom t$ is equal to
$\abstree x t'$ if $i =0$ and to $t_1\, t_2$ if $i \in \set{1,2}$, for
some $t', t_1, t_2 \in \iptree$. Following this observation, we
define, for each $t \in \iptree$ and $p \in \dom t$, the
\emph{subtree} of $t$ at $p$, denoted $\subtree t p$, by induction on
$p$ as follows: $\subtree t \emptyseq = t$,
$\subtree {\abstree x t} {\seq 0 \concat p} = \subtree t p$, and
$\subtree {t_1\, t_2}{\seq i \concat p} = \subtree {t_i} p$ for
$i \in \set{1,2}$. One can easily check that $\subtree t p$ is
uniquely defined modulo renaming of free variables.

\begin{definition}
  An \emph{infinite branch} in a lambda tree $t \in \iptree$ is an
  infinite sequence $S$ such that each proper prefix of $S$ is in
  $\dom{t}$. We call a proper prefix of $S$ a \emph{position along
    $S$}.
\end{definition}

Note that by instantiating K\"{o}nig's Lemma to lambda trees, we know
that a lambda tree is infinite iff it has an infinite branch.

The idea of the metric $\dda$ on lambda terms is to disallow (in the
ensuing metric completion) infinite branches that have only finitely
many non-strict positions along them. The following definition makes
this restriction explicit on lambda trees:
\begin{definition}
  \label{def:afinite}
  An infinite branch $S$ of a lambda tree $t$ is called
  \emph{$\ola$-bounded} if the $\ola$-depth of all positions along $S$
  is bounded by some $n <\omega$, i.e.\ $\adepth{p} < n$ for all
  $p < S$. The lambda tree $t$ is called \emph{$\ola$-unguarded} if it
  has an $\ola$-bounded infinite branch $S$. Otherwise, $t$ is called
  \emph{$\ola$-guarded}. The set of all $\ola$-guarded (total) lambda
  trees is denoted $\aptree$ (respectively $\atree$). In particular,
  $\aptree[000] = \ptree$ and $\aptree[111] = \iptree$.
\end{definition}
For example, the lambda tree $s$ with $s = s\, \vartree y$ is
$101$-unguarded while $t$ with $t = \abstree y t\, \vartree y$ is
$101$-guarded as each application is guarded by an abstraction (which
is non-strict).


For each strictness signature $\ola$, we give a metric $\ddta$ on
lambda trees that corresponds to the metric $\dda$ on lambda terms.
\begin{definition}
  For each two lambda trees $s,t \in \iptree$, define $\ddta(s,t) = 0$
  if $s = t$ and otherwise $\ddta(s,t) = 2^{-d}$, where $d$ is the
  least $\adepth{p}$ with $s(p) \not\simeq t(p)$.
\end{definition}

From the characterisation of the metric completion of $(\plam,\dda)$
from Kennaway et al.~\cite[Lemma~7]{kennaway97tcs} we know that the
metric space of $\ola$-guarded lambda trees $(\aptree,\ddta)$ is
indeed the metric completion of $(\plam,\dda)$ with the isometric
embedding $\lamtree\cdot\colon\plam\to\ptree$ \iftechrep%
(cf.\ Appendix~\ref{sec:direct-approach} for a more formal treatment).
\else%
(cf.\ the companion report~\cite{techrep}).
\fi%
Analogously, $(\atree,\ddta)$ is the metric completion of
$(\lam,\dda)$.

\section{The Ideal Completion}
\label{sec:ideal-completion}

In this section, we present an alternative to the metric completion
from Section~\ref{sec:metric-completion} that is based on a family of
partial orders on lambda terms indexed by strictness signatures. In
the following we assume basic familiarity with order theory.

\begin{definition}
  \label{def:alebot}
  Given a strictness signature $\ola$, the partial order $\alebot$ is
  the least transitive, reflexive order on $\plam$ satisfying the
  following for all
  $M,M',N,N' \in \plam$ and $x \in \calV$:\\[5pt]
  \begin{enumerate*}[(a),mode=unboxed]
    \begin{tabular}{rr@{\;}l@{\qquad}l@{\quad}l@{\quad}l@{\quad}r@{\,}l}
    \item \label{item:alebot1} & $\bot$ &$\alebot M$\\[1pt]
    \item \label{item:alebot2}& $\lambda x. M$ &$\alebot \lambda x . M'$ &if&  $M
      \alebot M'$ & and & $M \neq \bot$ &or 
      $a_0 = 1$\\[1pt]
    \item \label{item:alebot3}& $M N$ &$\alebot M'N$ &if& $M \alebot M'$ & and & 
      $M \neq \bot$ &or  $a_1 = 1$\\[1pt]
    \item \label{item:alebot4}& $M N$ &$\alebot MN'$ & if& $N \alebot N'$ & and & 
      $N \neq \bot$ &or  $a_2 = 1$
    \end{tabular}
  \end{enumerate*}
\end{definition}
For the case that $\ola = 111$, we obtain the partial order $\lebot$
that is typically used for ideal completions. This order is fully
monotone, i.e.\ $M \lebot M'$ implies $\lambda x. M$
$\lebot \lambda x . M'$, $M N \lebot M'N$ and $N M \lebot N M'$. By
contrast, $\alebot$ restricts monotonicity of abstraction in case
$a_0 =0$ and of application in case $a_1 = 0$ or $a_2 =
0$. Intuitively, we have $M \alebot N$ iff $N$ can be obtained from
$M$ by replacing occurrences of $\bot$ in $M$ at non-strict positions
with arbitrary terms. For example, if $\ola = 001$, then neither
$\lambda x. \bot \alebot \lambda x. x\, x$ nor
$\lambda x. \bot\, x \alebot \lambda x. x\, x$; but we do have that
$\lambda x.x\, \bot \alebot \lambda x. x\, x$.

With this intuition in mind, we translate $\alebot$ to a corresponding
order $\talebot$ on lambda trees as follows:
\begin{definition}
  \label{def:talebot}
  Given lambda trees $s,t\in\iptree$, we have $s
  \talebot t$ if\\[2pt]
  \begin{enumerate*}[(a),mode=unboxed]
    \begin{tabular}{rll}
    \item\label{item:talebot1} &$\dom{s} \subseteq \dom{t}$,& 
      \\
    \item\label{item:talebot2} &$s(p) = t(p)$  &for all $p \in
      \dom{s}$, and
      \\
    \item\label{item:talebot3} &$p \in \dom{s} \implies p\concat\seq i
      \in \dom{s}$ &for all $\ola$-strict positions $p\concat\seq i
      \in \dom t$.
    \end{tabular}
  \end{enumerate*}
\end{definition}
Conditions \ref{item:talebot1} and \ref{item:talebot2} alone would
give us the corresponding order for the standard partial order
$\lebot$. Condition \ref{item:talebot3} ensures that the partial
order $\talebot$ may not fill a hole in a strict position in the
left-hand side tree.

One can check that $(\iptree,\talebot)$ forms a partially ordered
set. Moreover, we have the following correspondence between the two
families of orders $\alebot$ and $\talebot$:
\begin{proposition}
  \label{prop:lamtreeIso}
  $\lamtree\cdot\colon(\plam,\alebot) \to (\ptree,\talebot)$ is an
  order isomorphism.
\end{proposition}

For the remainder of this section, we turn our focus to the partial
orders $\talebot$ on lambda trees. In particular, we show that
$(\aptree,\talebot)$ forms a \emph{complete semilattice} and that it
is (order isomorphic to) the ideal completion of $(\plam,\alebot)$.  A
complete semilattice is a partially ordered set $(A,\le)$ that is a
\emph{complete partial order} (\emph{cpo}) and that has a
\emph{greatest lower bound} (\emph{glb}) $\Glb B$ for every
\emph{non-empty} set $B \subseteq A$.\footnote{Equivalently, complete
  semilattices are bounded complete cpos. Hence, complete semilattices
  are a generalisation of \emph{Scott domains} (which in addition have
  to be \emph{algebraic}).} A partially ordered set $(A,\le)$ is a cpo if it
has a least element, and each directed set $D$ in $(A,\le)$ has a
\emph{least upper bound} (\emph{lub}) $\Lub D$; a set $D \subseteq A$
is called directed if for each two $a,b \in D$ there is some $c \in D$
with $a,b\le c$.

In particular, for any sequence $(a_\iota)_{\iota<\alpha}$ in a
complete semilattice, its \emph{limit inferior}, defined by
$\liminf_{\iota \limto \alpha}a_\iota = \Lub_{\beta<\alpha}
\left(\Glb_{\beta \le \iota < \alpha} a_\iota\right)$, exists. While
the metric completion lambda calculi are based on the limit of the
underlying metric space, our ideal completion lambda calculi are based
on the limit inferior.

To show that $(\aptree,\talebot)$ forms a complete semilattice
structure, we construct the appropriate lubs and glbs:
\begin{theorem}[cpo $(\aptree,\talebot)$]
  \label{thr:alebotCpo}
  The partially ordered set $(\aptree,\talebot)$ forms a complete
  partial order. In particular, the lub $t$ of a directed set $D$
  satisfies the following:\par\vspace{5pt} $\dom{t} = \bigcup_{s \in
    D} \dom{s}\hspace{2cm} s(p) = t(p) \quad \text{ for all } s\in D,
  p\in \dom{s}$
\end{theorem}
\begin{proof}[Proof sketch]
  The lambda tree $\bot$ is the least element in
  $(\aptree,\talebot)$. Construct the lub $t$ of $D$ as follows:
  $t(p) = s(p)$ iff there is some $s \in D$ with $p \in \dom{s}$. One
  can check that $t$ indeed is a well-defined lambda tree that is
  $\ola$-guarded and is the least upper bound of $D$.
\end{proof}

\begin{proposition}[glbs of $\talebot$]
  \label{prop:alebotGlb}
  Every non-empty subset $T$ of $\aptree$ has a glb $\Glb T$ in
  $(\aptree,\talebot)$ such that $\dom{\Glb T}$ is the largest set $P$
  satisfying the following properties:\vspace{-5pt}
  \begin{enumerate}[(1)]
  \item If $p\in P$, then there is some $l \in \lamsig$ such that
    $s(p) = l$ for all $s \in T$.
    \label{item:talebotGlb1}
  \item If $p \concat\seq i \in P$, then $p \in P$.
    \label{item:talebotGlb2}
  \item If $p \in P$, $a_i = 0$, and $p \concat \seq i \in \dom s$ for
    some $s \in T$, then $p\concat\seq i \in P$.
    \label{item:talebotGlb3}
  \end{enumerate}
\end{proposition}
\begin{proof}[Proof sketch]
  Let $P \subseteq \allpos$ be the largest set satisfying
  \ref{item:talebotGlb1} to \ref{item:talebotGlb3}. As these
  properties are closed under union, $P$ is well-defined. We define
  the partial function $t\colon \allpos \pfunto \lamsig$ as the
  restriction of an arbitrary lambda tree in $T$ to $P$.  Using
  \ref{item:talebotGlb1} and \ref{item:talebotGlb2}, one can show that
  $t$ is indeed a well-defined $\ola$-guarded lambda tree. One can then
  check that $t$ is the glb of $T$.
\end{proof}
For instance $\Glb\set{\abstree x\vartree x\,\vartree y,\abstree x
  \vartree y\,\vartree x}$ is $\abstree x \bot\,\bot$ for $011$,
$\abstree x \bot$ for $110$, and $\bot$ for $001$.

\begin{theorem}
  \label{thr:alebotCompSemi}
  $(\aptree,\talebot)$ is a complete semilattice for any $\ola$.
\end{theorem}
\begin{proof}
  Follows from Theorem~\ref{thr:alebotCpo} and
  Proposition~\ref{prop:alebotGlb}.
\end{proof}

We conclude this section by establishing the partially ordered set
$(\aptree,\talebot)$ as (order isomorphic to) the ideal completion of
$(\plam, \alebot)$.  Recall that, given a partially order set
$(A,\le)$, its ideal completion is an extension of the original
partially ordered set to a cpo. A set $B \subseteq A$ is called an
\emph{ideal} in $(A,\le)$ if it is directed and
\emph{downward-closed}, where the latter means that for all
$a \in A, b \in B$ with $a \le b$, we have that $a \in B$. The
\emph{ideal completion} of $(A,\le)$, is the partially ordered set
$(\calI,\subseteq)$, where $\calI$ is the set of all ideals in
$(A,\le)$ and $\subseteq$ is standard set inclusion.

\begin{theorem}[ideal completion]
  \label{thr:alebotCompl}
  The ideal completion of $(\plam,\alebot)$ is order isomorphic to
  $(\aptree, \talebot)$.
\end{theorem}
\begin{proof}[Proof sketch]
  By Proposition~\ref{prop:lamtreeIso}, it suffices to show that the
  ideal completion $(\calI,\subseteq)$ of $(\ptree,\talebot)$ is order
  isomorphic to $(\aptree,\talebot)$. To this end, we define two
  functions $\phi\colon \aptree \funto \calI$ and
  $\psi\colon \calI \funto \aptree$ as follows:
  $\phi(t) = \setcom{s\in\ptree}{s \talebot t}$; $\psi(T) = \Lub
  T$. Well-definedness of $\phi$ and $\psi$ follows from K\"onig's
  Lemma and Theorem~\ref{thr:alebotCpo}, respectively.
  Both $\phi$ and $\psi$ are obviously monotone and one can check that
  $\phi$ and $\psi$ are inverses of each other. Hence,
  $(\calI,\subseteq)$ is order isomorphic to $(\aptree,\talebot)$
\end{proof}
Now that we have established the connection between $\aptree$ and the
metric completion resp.\ the ideal completion of $\plam$, we turn our
focus to $\aptree$ for the rest of this paper.

The characterisation of lubs and glbs for the complete semilattice
$(\aptree,\talebot)$ allows us to relate the corresponding notion of
limit inferior with the limit in the complete metric space
$(\aptree,\ddta)$ as summarised in the following theorem:
\begin{theorem}
  \label{thr:limLiminf}
  Let $(t_\iota)_{\iota<\alpha}$ be a sequence in $\aptree$.\vspace{-5pt}
  \begin{enumerate}[(i)]
  \item If $\lim_{\iota\limto\alpha} t_\iota = t$ in
    $(\aptree,\ddta)$, then $\liminf_{\iota\limto \alpha} t_\iota = t$
    in $(\aptree,\talebot)$.
    \label{item:limLiminf1}
  \item If $\liminf_{\iota\limto\alpha} t_\iota = t$ in
    $(\aptree,\talebot)$ and $t$ is total, then
    $\lim_{\iota\limto\alpha} t_\iota = t$ in $(\aptree,\ddta)$.
    \label{item:limLiminf2}
  \end{enumerate}
\end{theorem}

The key to establish the correspondence above is the following
characterisation of the limit $t$ of a converging sequence
$(t_\iota)_{\iota<\alpha}$ in $(\aptree,\ddta)$:
\begin{center}
  $\dom{t} = \bigcup_{\beta<\alpha}\bigcap_{\beta\le\iota<\alpha}
  \dom{t_\iota},\text{ and} \quad t(p) = l \iff \exists \beta <
  \alpha \forall \beta \le
  \iota < \alpha\colon\; t_\iota(p) = l$
\end{center}
The proof of the correspondence result makes use of a notion of
truncation similar Arnold and Nivat's~\cite{arnold80fi} but
generalised to be compatible with the $\talebot$-orderings.

From the above findings we can conclude that the limit inferior in
$(\aptree,\talebot)$ restricted to total lambda trees coincides with
the limit in $(\atree,\ddta)$. In other words, the limit inferior is a
conservative extension of the limit. In the next section, we transfer
this result to (strong) convergence of reductions.

\section{Transfinite Reductions}
\label{sec:transf-reduct}

In this section, we study finite and transfinite reductions on lambda
trees. To this end, we assume for the remainder of this paper a fixed
strictness signature $\ola$ such that all subsequent definitions and
theorems work on the same universe of lambda trees $\aptree$ and its
associated structures $\ddta$ and $\talebot$ (unless stated
otherwise).  Moreover, we need a suitably general notion of reduction
steps beyond the familiar $\beta$- and $\eta$-rules in order to
accommodate B\"ohm reductions in Section~\ref{sec:bohm-like-tree}.
\begin{definition}
  
  A \emph{rewrite system} $R$ is a binary relation on $\aptree$ such
  that $(s,t) \in R$ implies that $s \neq \bot$. Given $s,t \in
  \aptree$ and $p \in \allpos$, an \emph{$R$-reduction step} from $s$
  to $t$ at $p$, denoted $s \to_{R,p} t$, is inductively defined as
  follows: if $(s,t) \in R$, then $s \to_{R,\emptyseq} t$; if $t
  \to_{R,p} t'$, then $\abstree x t \to_{R,\seq 0\concat p} \abstree
  x t'$, $t\, s \to_{R,\seq 1\concat p} t'\, s$, and $s\, t
  \to_{R,\seq 2\concat p} s\, t'$ for all $s \in \aptree$.  If $R$ or
  $p$ are irrelevant or clear from the context, we omit them in the
  notation $\to_{R,p}$. If $(t,t')\in R$, then $t$ is called an
  \emph{$R$-redex}. If $s \to_{R,p} t$, then $s$ is said to have an
  \emph{$R$-redex occurrence} at $p$. A lambda tree $t$ is called an
  \emph{$R$-normal form} if no $R$-reduction step starts from $t$. The
  prefix ``$R$-'' is dropped if $R$ is irrelevant or clear from the
  context.
\end{definition}

\begin{example}
  The familiar $\beta$- and $\eta$-rules form rewrite systems as
  follows:\vspace{-5pt}
  \begin{center}
    $\betared = \setcom{((\abstree x t)\, s,\subst t x
      s)}{s,t\in\aptree}%
    \qquad %
    \etared = \setcom{(\abstree x t\; \vartree x,t)}{t\in\aptree, x
      \nin\range{t}}$
  \end{center}\vspace{-5pt}
  where substitution $\subst t x s$ is defined as follows: for each $p
  \in \allpos$ we have that $\subst t x s (p) = t(p)$ if $t(p) \in
  \lamsig\setminus\set x$; $\subst t x s (p) = s(p_2)$ if $p = p_1\concat p_2, t(p_1) =
  x, s(p_2) \in \lamsig\setminus\allpos$; $\subst t x s (p) = p_1\concat s(p_2)$ if $p =
  p_1\concat p_2, t(p_1) = x, s(p_2) \in \allpos$; and $\subst t x s
  (p)$ is undefined
  otherwise.
\end{example}

The resulting $\betared$-reduction step relation $\to[\betared]$ on
lambda trees is isomorphic (via the isomorphism of
Theorem~\ref{thr:alebotCompl}) to the lifting of the ordinary finitary
$\beta$-reduction step relation on lambda terms to the ideal
completion via the lifting operator $[ \cdot \rangle$ of
Blom~\cite{blom01phd}. An analogous correspondence can be shown for
$\etared$ as well.

\begin{definition}
  A sequence $S =(t_\iota \to_{R,p_\iota} t_{\iota+1})_{\iota<\alpha}$ of
  $R$-reduction steps is called an \emph{$R$-reduction}; $S$ is called
  \emph{total} if each $t_\iota$ is total. If $S$ is finite, we also
  write $S\fcolon t_0 \fto*[R] t_\alpha$.
\end{definition}

The above notion of reductions is too general as it does not relate
lambda trees $t_\beta$ at a limit ordinal index $\beta$ to the lambda
trees $(t_\iota)_{\iota<\beta}$ that precede it. This shortcoming is
addressed with a suitable notion of convergence and continuity.  In
the literature on infinitary rewriting one finds two different
variants of convergence/continuity: a \emph{weak} variant, which
defines convergence/continuity only according to the underlying
structure (metric limit or limit inferior), and a \emph{strong}
variant, which also takes the position of contracted redexes into
consideration. While both the metric and the partial order lend
themselves to either variant, we only consider the strong variant here
and refer the reader to
  \iftechrep%
  Appendix~\ref{sec:weak-convergence}
  \else%
  the companion report~\cite{techrep}
  \fi%
for the weak variant.

We use the name $\mrs$-convergence and $\prs$-convergence to
distinguish between the metric- and the partial order-based notion of
convergence, respectively. Our notion of (strong) $\mrs$-convergence
is the same notion of convergence that Kennaway et
al.~\cite{kennaway97tcs} used for their infinitary lambda calculi. For
our notion of (strong) $\prs$-convergence we instantiate the abstract
notion of strong $\prs$-convergence from our previous
work~\cite{bahr10rta}. The key ingredient of $\prs$-convergence is the
notion of \emph{reduction context}, which assigns to each reduction
step $s \to t$ a lambda tree $c$ with $c \talebot s,t$. Intuitively,
this reduction context $c$ comprises the (maximal) fragment of $s$
that cannot be changed by the reduction step, regardless of the
reduction rule. For instance, the reduction context of
$\abstree x (\abstree y \vartree y)\, \vartree x \to \abstree x \vartree x$ is
$\abstree x \bot$ if $a_0 = 1$, and $\bot$ otherwise. The notion of
$\prs$-convergence is defined using the limit inferior of the sequence
of reduction contexts (instead of the original lambda trees
themselves). The canonical approach to derive such a reduction context
for any complete semilattice is to take the greatest lower bound of
the involved lambda trees $s$ and $t$ that does not contain any
position of the redex:
\begin{definition}
  The \emph{reduction context} of a reduction step $s \to_{p} t$ is
  the greatest lambda tree $c$ in $(\aptree,\talebot)$ with $c
  \talebot s,t$ and $p \nin \dom{c}$; we write $s \to[c] t$ to
  indicate the reduction context $c$.
\end{definition}

In order to simplify reasoning and provide an intuitive understanding
of the concept, we give a direct construction of reduction contexts as
well:
\begin{definition}
  Given $t \in \iptree$ and $p \in \dom t$, we write $\posminus t p$
  for the restriction of $t$ to the domain $\setcom{q\in\dom
    t}{p\not\le q}$, and $\acut p$ for the longest non-strict prefix
  of $p$.
\end{definition}
That is, $\posminus t p$ is obtained from $t$ by replacing the subtree
at $p$ with $\bot$.  Moreover, $\acut\,$ can be characterised as
follows: $\acut\emptyseq = \emptyseq$;
$\acut{(p\concat \seq i)} = p\concat \seq i$ if $a_i = 1$; and
$\acut{(p\concat \seq i)} = \acut p$ if $a_i = 0$.

\begin{lemma}
  \label{lem:reductionContextCut}
  The reduction context of $s\to_{p} t$ is equal to
  $\posminus{s}{\acut p}$ and $\posminus{t}{\acut p}$.
\end{lemma}
\begin{proof}[Proof sketch]
  By a straightforward induction on $p$.
\end{proof}
That is, the reduction context of $s\to_{p} t$ is obtained from $s$ by
removing the most deeply nested subtree that both contains the redex
and is in a non-strict position.  The ensuing notions of strong
convergence of reductions are spelled out as follows:
\begin{definition}
  An $R$-reduction $S=(t_\iota \to_{p_\iota, c_\iota}
  t_{\iota+1})_{\iota<\alpha}$ \emph{$\mrs$-converges} to $t_\alpha$,
  denoted $S\colon t_0 \mato[R] t_\alpha$, if
  $\lim_{\iota\limto\gamma} t_\iota = t_\gamma$ and
  $(\adepth{p_\iota})_{\iota < \gamma}$ tends to infinity for all
  limit ordinals $\gamma \le \alpha$. $S$ \emph{$\prs$-converges} to
  $t_\alpha$, denoted $S\colon t_0 \pato_R t_\alpha$, if
  $\liminf_{\iota\limto\gamma} c_\iota = t_\gamma$ for all limit
  ordinals $\gamma\le\alpha$. $S$ is called \emph{$\mrs$-continuous}
  resp.\ \emph{$\prs$-continuous} if the corresponding convergence
  conditions hold for limit ordinals $\gamma < \alpha$ (instead of
  $\gamma \le \alpha$).
\end{definition}
Intuitively, strong convergence under-approximates convergence in the
underlying structure (i.e.\ weak convergence) by assuming that every
contraction changes the root symbol of the redex. Thus, given a
reduction step $s \to[p] t$, strong convergence assumes that the
shortest position at which $s$ and $t$ differ is $p$.

The semilattice structure underlying $\prs$-convergence ensures that every
$\prs$-continuous reduction also $\prs$-converges, whereas
$\mrs$-convergence does not necessarily follow from $\mrs$-continuity:
\begin{example}
  \label{ex:strongConv}
  Given
  $\Omega = (\abstree x \vartree x\,\vartree x) (\abstree x \vartree
  x\,\vartree x)$ and
  $t = (\abstree x \vartree x\,\Omega)\,\vartree y$, we consider the
  $\betared$-reduction $S\fcolon t \to t \to \dots$ that repeatedly
  contracts the redex $\Omega$ in $t$.  $S$ is trivially $\mrs$- and
  $\prs$-continuous. However, it is not $\mrs$-convergent, since
  contraction takes place at a constant $\ola$-depth, namely
  $\adepth{\seq{1,0,2}}$. But it $\prs$-converges to
  $\posminus{t}{\acut{ \seq{1,0,2}}}$, which is also the reduction
  context of each reduction step in $S$ and is equal to
  $(\abstree x \vartree x\,\bot)\,\vartree y$ if $a_2 = 1$, to
  $(\abstree x \bot)\,\vartree y$ if $a_2 = 0$ but $a_0 = 1$, to
  $\bot\,\vartree y$ if $\ola=010$, and to $\bot$ if $\ola = 000$.
\end{example}

Similarly to the correspondence between the limit and the limit inferior
in Theorem~\ref{thr:limLiminf}, we find a correspondence between
$\prs$- and $\mrs$-convergence.
\begin{proposition}
  \label{prop:mconvPconv}
  For each reduction $S\colon s \mato t$, we also have that $S\colon s
  \pato t$.%
\end{proposition}
\begin{proof}[Proof sketch]
  Let $S=(t_\iota \to_{p_\iota,c_\iota}
  t_{\iota+1})_{\iota<\alpha}$. If $S$ $\mrs$-converges, then
  $(\adepth{p_\iota})_{\iota<\gamma}$ tends to infinity for all limit
  ordinals $\gamma < \alpha$, i.e.\ for each $d < \omega$ we have that
  $\adepth{p_\iota}\ge d$ after some $\delta < \gamma$. With the help
  of Lemma~\ref{lem:reductionContextCut}, one can show that the latter
  implies that $t_\iota$ and $c_\iota$ coincide up to $\ola$-depth $d$
  for all $\delta \le \iota < \gamma$. Consequently,
  $\lim_{\iota\limto\gamma} t_\iota = \lim_{\iota\limto\gamma}
  c_\iota$, which, by
  Theorem~\ref{thr:limLiminf}~\ref{item:limLiminf1}, implies
  $\lim_{\iota\limto\gamma} t_\iota = \liminf_{\iota\limto\gamma}
  c_\iota$. Since this holds for all limit ordinals $\gamma \le
  \alpha$, we know that $S$ also $\prs$-converges to $t$.
\end{proof}
With the proposition above, we derive the other direction of the
correspondence:
\begin{proposition}
  \label{prop:pconvMconv}
  $S\colon s \pato t$ implies $S\colon s \mato t$ whenever $S$ and $t$
  are total.
\end{proposition}
\begin{proof}[Proof sketch]
  One can show that the totality of $S$ and $t$ implies that the
  $\ola$-depth of contracted redexes in each open prefix of $S$ tends
  to infinity. Using Proposition~5.5 from \cite{bahr10rta}, we can
  show that the latter implies that $S$ also $\mrs$-converges. Then
  according to Proposition~\ref{prop:mconvPconv}, $S$ must
  $\mrs$-converge to the same lambda tree $t$.
\end{proof}
Note that it is not sufficient that the two trees $s$ and $t$ are
total. For example, the $\betared$-reduction
$S\colon(\abstree x \vartree y)\,\Omega \pato (\abstree x \vartree y)\,\bot
\to \vartree y$ $\prs$-converges to $\vartree y$ but does not $\mrs$-converge.

Putting Propositions~\ref{prop:mconvPconv} and \ref{prop:pconvMconv}
together we obtain that $\prs$-convergence is a conservative
extension of $\mrs$-convergence:
\begin{corollary}
\label{cor:mconvPconv}
  $S\colon s \mato t$ iff $S\colon s \pato t$ whenever $S$ and $t$ are
  total.
\end{corollary}

\section{Beta Reduction}
\label{sec:bohm-like-tree}

So far we have only studied the properties of $\prs$-convergence
independent of the rewrite system. In this section, we specifically
study $\betared$-reduction and show infinitary normalisation for all
of our calculi, and infinitary confluence for three of them. However,
considering pure $\betared$-reduction, infinitary confluence only
holds for the $111$ calculus. We can construct counterexamples for the
other calculi:
\begin{example}[\cite{kennaway97tcs}]
  \label{ex:counterCR}
  Given $a_2 = 0$ and $t = (\abstree x \vartree y)\, \Omega$, we find
  reductions $t \pato[\betared] \bot$ and
  $t \to[\betared] \vartree y$. Given $a_2 = 1$, $a_1 = 0$, and
  $t = (\abstree x \vartree x \,\vartree y)\,\Omega$, we have
  $t \pato[\betared] (\abstree x \vartree x \,\vartree y) \,
  \bot\to[\betared] \bot\,\vartree y$ and
  $t \to[\betared] \Omega\, \vartree y \pato[\betared]
  \bot$. Similarly, given $a_2 = 1$, $a_0 = 0$, and
  $t = (\abstree x \abstree y \vartree x)\,\Omega$, we have
  $t \pato[\betared] (\abstree x \abstree y \vartree x) \,
  \bot\to[\betared] \abstree y \bot$ and
  $t \to[\betared] \abstree y \Omega \pato[\betared] \bot$.
\end{example}
Infinitary confluence of pure $\betared$-reduction fails for all
$\mrs$-convergence calculi of Kennaway et al.\cite{kennaway97tcs} --
including the $111$ calculus. On the other hand, the B\"ohm reduction
calculi of Kennaway et al.~\cite{kennaway99jflp}, which extend pure
$\betared$-reduction with infinitely many rules of the form
$t \to \bot$, do satisfy infinitary confluence for the $001$, $101$,
and $111$ calculi.

We would like to obtain similar confluence results for the $001$,
$101$, and $111$ $\prs$-convergence calculi. However, the gap we have
to bridge to achieve infinitary confluence is much narrower in our
$\prs$-convergence calculi. Intuitively, confluence fails for $001$
and $101$ because $\prs$-convergence only captures partiality that is
due to infinite reductions, but not partiality that can propagate via
finite reductions: For example, in the $101$ calculus we have
$\Omega\, \vartree y \pato[\betared] \bot$ but
$\bot\,\vartree y \not\pato[\betared] \bot$. In order to obtain the
desired confluence properties, we have to add the rules
$\abstree x \bot \to \bot$ (for $001$) and $\bot\, t \to \bot$ (for
$001$ and $101$). More generally we define these $\strictred$-rules
formally as follows:
\begin{align*}
  \strictred =\ &\setcom{(t_1\,t_2,\bot)}{t_1,t_2 \in \aptree, t_i
                  =\bot, a_i = 0}
  \cup \setcom{(\abstree x \bot,\bot) )}{a_0 = 0}
\end{align*}
We use the notation $\betasred$ to denote $\betared \cup
\strictred$. Abusing notation, we also write $\betaSred$ to refer to
$\betared$ or $\betasred$, e.g.\ if a property holds for either
system. Note that for the $111$ calculus, $\betasred = \betared$.

In addition, we continue studying the relation between
$\mrs$-convergence and $\prs$-convergence: In general, they are subtly
different, but we show that a $\prs$-converging $\betaSred$-reduction
can be adequately simulated by an $\mrs$-converging
$\bohmred$-reduction and vice versa, where $\bohmred$ is an extension
of $\betared$, called B\"ohm rewrite system, which additionally
contains rules of the form $t \to \bot$. This result uses the same
construction used by Kennaway et al.~\cite{kennaway99jflp} to study
so-called \emph{meaningless terms}.

In the remainder of this section we first characterise the set of
lambda trees that $\prs$-converge to $\bot$
(Section~\ref{sec:partiality}); we then establish a correspondence
between pure $\prs$-convergence and $\mrs$-convergence extended with
rules $t \to \bot$ for lambda trees $t$ that $\prs$-converge to $\bot$
(Section~\ref{sec:correspondence}); and finally we prove infinitary
confluence and normalisation for $\prs$-convergent
$\betasred$-reductions in the $001$, $101$, and $111$ calculi
(Section~\ref{sec:infin-norm-confl}).  For the infinitary confluence
result, we make use of the correspondence between $\prs$-convergence
and $\mrs$-convergence.

\subsection{Partiality}
\label{sec:partiality}

We begin with the characterisation of lambda trees that
$\prs$-converge to $\bot$:
\begin{definition}
  Given an open reduction
  $S = (t_\iota \to[p_\iota] t_{\iota+1})_{\iota<\alpha}$, a position
  $p$ is called \emph{volatile} in $S$ if, for each $\beta < \alpha$,
  there is some $\beta \le \gamma < \alpha$ with
  $\acut{p_\gamma} \le p \le p_\gamma$. If $p$ is volatile in $S$ but
  no proper prefix of $p$ is, then $p$ is called \emph{outermost-volatile}
  in $S$.
\end{definition}

For instance, in the $\betared$-reduction in
Example~\ref{ex:strongConv}, $\seq{1,0,2}$ is volatile and
$\acut{\seq{1,0,2}}$ is outermost-volatile. Note that
outermost-volatile positions must be non-strict, because if $p$ is
volatile, then so is $\acut p$.

The presence of volatile positions characterises partiality in
$\prs$-convergent reductions, which by Corollary~\ref{cor:mconvPconv}
can be stated as follows:
\begin{proposition}
  \label{prop:pconvMconvVolatile}
  $S \fcolon s \mato t$ iff no prefix of $S$ has volatile
  positions and $S \fcolon s \pato t$.
\end{proposition}
\begin{proof}[Proof sketch]
  Let $S = (t_\iota \to[p_\iota] t_{\iota+1})_{\iota<\alpha}$.  The
  ``only if'' direction follows from Proposition~\ref{prop:mconvPconv}
  and the fact that if $(\adepth{p_\iota})_{\iota<\beta}$ tends to
  infinity, then $\prefix{S}{\beta}$ has no volatile positions.  For
  the ``if'' direction, the infinite pigeonhole principle yields that
  $(\adepth{p_\iota})_{\iota<\beta}$ tends to infinity. Using this
  fact, one can show that $S \fcolon s \mato t$.
\end{proof}

More specifically, outermost-volatile positions pinpoint the exact
location of partiality in the result of a $\prs$-converging reduction.
\begin{lemma}
  \label{lem:volatileBot}
  If $p$ is outermost-volatile in $S \fcolon s \pato t$, then $p \in
  \domBot t$.
\end{lemma}
\begin{proof}[Proof sketch]
  Let $S = (t_\iota \to[p_\iota,c_\iota]
  t_{\iota+1})_{\iota<\alpha}$. Since $p$ is volatile in $S$, we find
  for each $\beta < \alpha$ some $\beta \le \iota < \alpha$ with
  $\acut{p_\iota} \le p$. Hence, by
  Lemma~\ref{lem:reductionContextCut}, we know that
  $p \nin \dom{c_\iota}$. Consequently, by Theorem~\ref{thr:alebotCpo}
  and Proposition~\ref{prop:alebotGlb}, we have that $p \nin \dom{t}$.
  If $p = \emptyseq$, then $p \in \domBot{t}$ follows immediately. If
  $p = q\concat\seq 0$, then one can use the fact that no prefix of
  $q$ is volatile to show that $t(q) = \lambda$, which means that
  $p \in \domBot{t}$. The argument for the cases $p = q\concat\seq 1$
  and $p = q\concat\seq 2$ is analogous.
\end{proof}

This characterisation of partiality in terms of volatile positions
motivates the following notions of destructiveness and fragility:
\begin{definition}
  A reduction $S$ is called \emph{destructive} if it is
  $\prs$-continuous and $\emptyseq$ is volatile in $S$. A lambda tree
  $t\in \aptree$ is called \emph{fragile} if there is a destructive
  $\betared$-reduction starting from $t$. The set of all fragile
  \emph{total} lambda trees is denoted $\afrag$.
\end{definition}
Note that fragility is defined in terms of destructive
$\betared$-reductions. However, one can show that a destructive
$\betared$-reduction exists iff a destructive $\betasred$-reduction
exists.

The following proposition explains why destructive reductions have
deserved their name:
\begin{proposition}
  \label{prop:destructBot}
  An open reduction is destructive iff it $\prs$-converges to
  $\bot$.
\end{proposition}
\begin{proof}[Proof sketch]
  The ``only if'' direction follows from Lemma~\ref{lem:volatileBot};
  the converse direction can be shown using the characterisation of
  the limit inferior (Theorem~\ref{thr:alebotCpo},
  Proposition~\ref{prop:alebotGlb}).
\end{proof}
For example, the $\betared$-reduction $\Omega \to \Omega \to \dots$
(cf.\ Example~\ref{ex:strongConv}) $\prs$-converges to $\bot$ and is
thus destructive. As a corollary from the above proposition, we obtain
that every fragile lambda tree -- such as $\Omega$ -- can be
contracted to $\bot$ by an open $\prs$-convergent reduction.

\subsection{Correspondence}
\label{sec:correspondence}

To compare $\mrs$- and $\prs$-converging reductions, we employ B\"ohm
rewrite systems and the underlying notion of $\bot$-instantiation from
Kennaway et al.'s work on meaningless terms \cite{kennaway99jflp}.
\begin{definition}
  Let $\calU \subseteq \itree$ and $t \in \iptree$. A lambda tree
  $s \in \itree$ is called a $\bot$-instance of $t$ w.r.t.\ $\calU$ if
  $s$ is obtained from $t$ by inserting elements of $\calU$ into $t$
  at each position $p \in \domBot{t}$, i.e.\ $s(p) = t(p)$ for all
  $p \in \dom{t}$ and $\subtree s p \in \calU$ for all
  $p \in \domBot t$. The set of lambda trees that have a
  $\bot$-instance w.r.t.\ $\calU$ that is in $\calU$ itself is denoted
  $\calU_\bot$. In other words, $t \in \calU_\bot$ iff there is a
  lambda tree $s \in \calU$ such that $s$ is obtained from $t$ by
  replacing occurrences of $\bot$ in $t$ by lambda trees from $\calU$.
\end{definition}

In particular, we will use the above construction with the set of
fragile total lambda trees $\afrag$, which gives us the set
$\afrag_\bot$.

Finally, we give the construction of B\"ohm rewrite systems.
\begin{definition}
  For each set $\calU \subseteq \atree$, we define the following two
  rewrite systems:
  \begin{align*}
    \botred[\calU] = \setcom{(t,\bot)}{t
      \in\calU_\bot\setminus\set\bot},\hspace{2cm} \bohmred[\calU] = \betared\
    \cup\botred[\calU]
  \end{align*}
  If $\calU$ is clear from the context, we instead use the notation
  $\botred$ and $\bohmred$, respectively.
\end{definition}
In particular, we consider the B\"ohm rewrite system w.r.t.\ fragile
total lambda trees, denoted by $\bohmred[\afrag]$. We start with one
direction of the correspondence between $\prs$-converging
$\betaSred$-reductions and $\mrs$-converging
$\bohmred[\afrag]$-reductions:
\begin{theorem}
  \label{thr:pConvBohmred}
  If $s \pato[\betasred] t$, then $s \mato[\bohmred] t$, where
  $\bohmred = \bohmred[\afrag]$.
\end{theorem}
\begin{proof}[Proof sketch]
  Given $S\fcolon s \pato[\betasred] t$, we construct a
  $\bohmred$-reduction $T$ from $S$ that also $\prs$-converges to $t$
  but that has no volatile positions in any of its open
  prefixes. Thus, according to
  Proposition~\ref{prop:pconvMconvVolatile},
  $T\fcolon s \mato[\bohmred] t$. The construction of $T$ removes
  steps in $S$ that take place at or below any outermost-volatile
  position of some prefix of $S$ and replaces them by a single
  $\botred$-step. Such a $\botred$-step can be performed since a
  fragile lambda tree must be responsible for an outermost-volatile
  position. Moreover, $\strictred$-steps in $S$ are $\botred$-steps in
  $T$ since $\strictred \subseteq\
  \botred[\afrag]$. Lemma~\ref{lem:volatileBot} can then be used to
  show that the resulting $\bohmred$-reduction $T$ 
  $\prs$-converges to $t$.
\end{proof}

The converse direction of Theorem~\ref{thr:pConvBohmred} does not hold
in general. The problem is that $\botred$-steps can be more selective
in which fragile lambda subtree to contract to $\bot$ compared to
$\prs$-convergent reductions with volatile positions. If $p$ is a
volatile position, then so is $\acut p$. Consequently, volatile
positions and thus `$\bot$'s in the result of a $\prs$-converging
reduction are propagated upwards through strict positions. For
example, let $a_0 = 0$, and $t = \abstree y \Omega$. Since $\Omega$ is
fragile, we have the reduction $t \to[\botred] \abstree y \bot$. On
the other hand, via $\prs$-convergent $\betared$-reductions, $t$ only
reduces to itself and $\bot$. This phenomenon, however, does not occur
if we restrict ourselves to the strictness signature $111$ or if we
only consider $\botred$-normal forms. Indeed, in the above example,
$\abstree y \bot$ is not a $\botred$-normal form and can be contracted
to $\bot$ with a $\botred$-step.
\begin{theorem}
  \label{thr:bohmredPConv}
  Let $\bohmred = \bohmred[\afrag]$ and $s \mato[\bohmred] t$ such
  that $s$ is total. Then $s \pato[\betared] t$ if $\ola = 111$ or $t$
  is a $\botred$-normal form.
\end{theorem}
\begin{proof}[Proof sketch]
  The reduction $s \mato[\bohmred] t$ can be factored into
  $S\fcolon s \mato[\betared] s'$ and $T\fcolon s' \mato[\botred] t$
  (by the same proof as Lemma~27 of Kennaway et
  al.~\cite{kennaway99jflp}).  Moreover, we may assume w.l.o.g.\ that
  $T$ contracts disjoint $\botred$-redexes in $s'$ (using an argument
  similar to Lemma~7.2.4 of Ketema~\cite{ketema06phd}). By
  Proposition~\ref{prop:mconvPconv}, we have that
  $S\fcolon s \pato[\betared] s'$ and that
  $T\fcolon s' \pato[\botred] t$. For each step $u \to[\botred,p] v$
  in $T$ we find a reduction $T_p \fcolon u \pato[\betared] v'$ in
  which $p$ is volatile since $\subtree u p$ must be fragile. Given
  that $\ola = 111$ or that $t$ is a $\botred$-normal form, we can
  show that $p$ is in fact outermost-volatile in $T_p$. Hence, the
  equality $v = v'$ follows from
  Lemma~\ref{lem:volatileBot}. Therefore, we may replace each step
  $u \to[\botred,p] v$ in $T$ by $T_p$, which yields a reduction
  $s' \pato[\betared] t$.
\end{proof}

That is, in general we get one direction of the correspondence --
namely from metric to partial order reduction -- only for reductions
to normal forms. However, this does not matter that much as
$\prs$-converging $\betaSred$-reductions (an thus also
$\mrs$-converging $\bohmred[\afrag]$-reductions) are normalising as we
show below.

\subsection{Infinitary Normalisation and Confluence}
\label{sec:infin-norm-confl}

We begin by recalling the notion of active lambda
trees~\cite{kennaway99jflp}, which we use to establish infinitary
normalisation and as an alternative characterisation of fragile lambda
trees (in the $001$, $101$, and $111$ calculi).

\begin{definition}
  A lambda tree $t$ is called \emph{stable} if no lambda tree $t'$
  with $t \fto*[\betared] t'$ has a $\betared$-redex occurrence at
  $\ola$-depth 0; $t$ is called \emph{active} if no lambda tree $t'$
  with $t \fto*[\betared] t'$ is stable. The set of all active \emph{total}
  lambda trees is denoted by $\aact$.
\end{definition}
To construct normalising $\prs$-convergent reductions, we follow the
idea of Kennaway et al.~\cite{kennaway99jflp}: We contract all
subtrees of the initial lambda tree into stable form. The only way to
achieve this for active subtrees is to annihilate them by a
destructive reduction. The basis for that strategy is the following
observation:
\begin{lemma}
  \label{lem:aactAfrag}
  Every active lambda tree is fragile.
\end{lemma}
\begin{proof}
  If $t_0$ is active, we find a reduction $t_0 \fto*[\betared] t'_0$
  to a $\betared$-redex at $\ola$-depth $0$. By contracting this redex
  we get a lambda tree $t_1$ that is active, too. By repeating this
  argument we obtain a destructive reduction $t_0 \fto*[\betared] t_0'
  \to[\betared] t_1 \fto*[\betared] t_1' \to[\betared]
  \dots$.
\end{proof}
The following normalisation result then follows straightforwardly:
\begin{theorem}
  \label{thr:prsNormalising}
  For each $s \in \aptree$, there is a normalising reduction $s
  \pato[\betaSred] t$.
\end{theorem}
\begin{proof}[Proof sketch]
  Similar to Theorem~1 of Kennaway et al.~\cite{kennaway99jflp}: an
  active subtree at position $p$ is by Lemma~\ref{lem:aactAfrag} also
  fragile. Hence, there is a $\betared$-reduction in which a prefix of
  $p$ is outermost-volatile. By Lemma~\ref{lem:volatileBot}, such a
  reduction annihilates the active subtree at $p$. This yields a
  reduction $s \pato[\betared] t$ to $\betared$-normal form $t$, which
  can be extended by a reduction $t \pato[\strictred] u$ to a
  $\betasred$-normal form $u$.
\end{proof}

From the above we immediately obtain the corresponding result for
$\mrs$-convergence:
\begin{theorem}
  \label{thr:bohmNormalising}
  For each $s \in \aptree$ there is a normalising reduction $s
  \mato[{\bohmred[\afrag]}] t$.
\end{theorem}
\begin{proof}
  By Theorem~\ref{thr:prsNormalising} and \ref{thr:pConvBohmred},
  as $\betasred$-normal forms are also $\bohmred[\afrag]$-normal
  forms.
\end{proof}
Consequently, we can derive the following correspondence result.
\begin{corollary}
  For each $s \in \atree$ with $s \mato[{\bohmred[\afrag]}] t$, there
  is a reduction $t \mato[{\bohmred[\afrag]}] t'$ such that $s
  \pato[\betared] t'$.
\end{corollary}
\begin{proof}
  According to Theorem~\ref{thr:bohmNormalising}, there is a
  normalising reduction $t \mato[{\bohmred[\afrag]}] t'$. Then a
  reduction $s \pato[\betared] t'$ exists by
  Theorem~\ref{thr:bohmredPConv}.
\end{proof}

A shortcoming of this correspondence property and the correspondence
properties established in Section~\ref{sec:correspondence} is that
they consider $\mrs$-convergence in the system $\bohmred[\afrag]$,
which is unsatisfactory since $\afrag$ is defined using
$\prs$-convergence. A more appropriate choice would be the set $\aact$
of active terms, which is defined in terms of finitary reduction
only. To obtain a correspondence in terms of $\aact$, we will show
that $\afrag = \aact$ for strictness signatures $001$, $101$, and
$111$. To prove this equality of fragility and activeness, we need the
following key lemma, which can be proved using descendants and
complete developments%
\iftechrep\ (cf.\ Appendix~\ref{sec:infin-strip-lemma})\fi.

\begin{lemma}[Infinitary Strip Lemma]
  \label{lem:stripLem'}
  Given $S\fcolon s \pato[\betasred] t_1$ and
  $T \fcolon s \fto*[\betasred] t_2$, there are reductions
  $S'\fcolon t_1 \pato[\betasred] t$ and
  $T'\fcolon t_2 \pato[\betasred] t$, provided
  $\ola \in \set{001,101,111}$.
\end{lemma}
Recall that $\betasred = \betared$ for $\ola = 111$, i.e.\ the
infinitary strip lemma holds for pure $\betared$-reduction in the
$111$ calculus; but it does not hold for $001$ and $101$ as
Example~\ref{ex:counterCR} demonstrates. Hence, the need for
$\strictred$-rules. By contrast, in the metric calculi of Kennaway et
al.~\cite{kennaway97tcs} the infinitary strip lemma does not hold for
any $\ola$. In order to obtain the infinitary strip lemma and
confluence, Kennaway et al.\ extended $\beta$-reduction to B\"ohm
reduction.

We use the Infinitary Strip Lemma to show that $\prs$-convergent
reductions to $\bot$ can be compressed to length at most $\omega$.
\begin{lemma}
  \label{lem:botCompress}
  If $\ola \in \set{001,101,111}$ and
  $S\fcolon t \pato[\betasred] \bot$, then there is a reduction
  $T\fcolon t \pato[\betasred] \bot$ of length $\le\omega$. If $t$ is
  total, then $T$ is a $\betared$-reduction of length $\omega$.
\end{lemma}
\begin{proof}[Proof sketch]
  If $\len{S}\le\omega$, we are done. Otherwise, we can construct a
  finite reduction $t \fto*[\betasred] t'$ with at least one
  contraction at $\ola$-depth $0$ either using a finite approximation
  property of $\prs$-convergence (in case $S$ contracts
  $\betared$-redex at $\ola$-depth $0$) or by an induction argument
  (in case $S$ contracts $\strictred$-redex at root position). By
  Lemma~\ref{lem:stripLem'}, there is a reduction
  $S'\fcolon t' \pato[\betasred] \bot$. Thus, we can repeat the
  argument for $S'$. Iterating this argument yields either a reduction
  $t \fto*[\betasred] \bot$ or a reduction $t \pato[\betasred] s'$ of
  length $\omega$ with infinitely many contractions at $\ola$-depth
  $0$, and thus $s' = \bot$.
  If $s$ is total, then $T$ cannot be finite, as finite
  $\betasred$-reductions preserve totality. Hence, no step in $T$ can
  be an $\strictred$-step.
\end{proof}

\begin{lemma}
\label{lem:afragAact}
If $\ola \in \set{001,101,111}$, a total lambda tree is active iff
it is fragile.
\end{lemma}
\begin{proof}
  The ``only if'' direction follows from
  Lemma~\ref{lem:aactAfrag}. For the converse direction let $t$ be
  total and fragile, and let $t \fto*[\betared] t_1$. Since $t$ is
  fragile, there is a reduction $t \pato[\betasred] \bot$ according to
  Proposition~\ref{prop:destructBot}. Hence, by
  Lemma~\ref{lem:stripLem'}, there is a reduction
  $T\fcolon t_1 \pato[\betasred] \bot$, which we can assume, according
  to Lemma~\ref{lem:botCompress}, to be a $\betared$-reduction of
  length $\omega$. Since $T$ is, by
  Proposition~\ref{prop:destructBot}, destructive, there is a proper
  prefix $T' \fcolon t_1 \pato[\betared] t_2$ of $T$ such that $t_2$
  has a redex occurrence at $\ola$-depth $0$. Because $T$ is of length
  $\omega$, $T'$ is finite i.e.\ $T' \fcolon t_1 \fto*[\betared] t_2$.
\end{proof}

The above lemma allows us to derive confluence w.r.t.\
$\prs$-convergent reductions from the confluence results w.r.t.\
$\mrs$-convergence of Kennaway et al.~\cite{kennaway97tcs}:
\begin{theorem}[infinitary confluence]
  \label{thr:confluence}
  Given $\ola \in \set{001,101,111}$, we have that $s
  \pato[\betasred] t_1$ and $s \pato[\betasred] t_2$ implies that $t_1
  \pato[\betasred] t$ and $t_2 \pato[\betasred] t$.
\end{theorem}
\begin{proof}
  According to Theorem~\ref{thr:prsNormalising}, we can extend the
  existing reductions by normalising reductions
  $t_1 \pato[\betasred] t_1'$ and $t_2 \pato[\betasred]
  t_2'$. According to Theorem~\ref{thr:pConvBohmred} and
  Lemma~\ref{lem:afragAact}, the resulting normalising reductions
  $s \pato[\betasred] t_1'$ and $s \pato[\betasred] t_2'$ are also
  $\mrs$-convergent $\bohmred[\aact]$-reductions. Kennaway et
  al.~\cite{kennaway97tcs} have shown that such reductions are
  confluent. Hence, $t_1' = t_2'$ (as $\betasred$-normal forms are
  $\bohmred[\aact]$-normal forms too).
\end{proof}

Together with the earlier normalisation result, this means that the
$001$, $101$, and $111$ calculi have unique normal forms w.r.t.\
$\pato[\betasred]$. By the correspondence results between the metric
and the partial order calculi, these normal forms are the same as the
unique normal forms w.r.t.\
$\mato[{\bohmred[\aact]}]$~\cite{kennaway97tcs}, which in turn
correspond to B\"ohm Trees, Levy-Longo Trees, and Berarducci Trees,
respectively.

\section{Related Work}
\label{sec:related-work}

The use of ideal completion in lambda calculus to construct infinite
terms has a long history (see e.g.\ Ketema~\cite{ketema06phd} for an
overview), in particular in the form of constructing infinite normal
forms such as B\"ohm Trees. In that line of work, the ideal completion
is typically based on the fully monotone partial order $\lebot$
generated by $\bot \lebot M$ for any term $M$. Different kinds of
infinite normal forms are then obtained by modulating the set of rules
that are used to generate the normal forms. In this paper, we instead
modulated the partial order and we have constructed full infinitary
calculi in the style of Kennaway et al.~\cite{kennaway97tcs}. Blom's
abstract theory of infinite normal forms and infinitary rewriting
based on ideal completion~\cite{blom01phd} has been crucial for
developing our infinitary calculi.

In previous work, we have compared infinitary rewriting based on
partial orders vs.\ metric spaces in a first-order setting
\cite{bahr10rta2,bahr12rta}. However, in that work we have only
considered fully non-strict convergence, whereas we consider varying
modes of strictness in the present paper.

Blom's work \cite{blom04rta} on \emph{preservation calculi} is similar
to our ideal completion calculi. Blom also considers different calculi
indexed by strictness signatures and relates them to the corresponding
metric calculi. However, he uses the same partial order
$\talebot[111]$ for all calculi; the different calculi vary in the
notion of reduction contexts they use. Blom's reduction contexts are
the same as our reduction contexts, and his $\Omega$-rules are more
general variants of our $\strictred$-rules. However, his approach of
using a single partial order has some caveats:

Firstly, there is no corresponding weak notion of preservation
sequences that corresponds to weak $\mrs$-convergence. Secondly, the
partially ordered set $(\aptree,\talebot[111])$ is only a complete
semilattice for $\ola = 111$; otherwise it is not even a cpo and limit
inferiors do not always exist. For example, let $t$ be an
$\ola$-unguarded lambda tree (i.e.\ $t \nin \aptree$), and for each
$i < \omega$ let $t_i$ be the restriction of $t$ to positions of depth
$< i$, which means that $t_i \in \aptree$. Then
$\liminf_{i \limto \omega} t_i$ w.r.t.\ $\talebot[111]$ is $t$ itself
and thus not in $\aptree$ even though all $t_i$ are. This does not
cause a problem, if one only considers reduction contexts of
$\prs$-continuous reductions, though.

For the comparison of his preservation calculi with the metric
calculi, Blom uses a notion of \emph{$0$-active} terms, which is
different from the notion of active terms as used here and by Kennaway
et al.~\cite{kennaway97tcs,kennaway99jflp} (under the names
$0$-activeness resp.\ $abc$-activeness). Blom defines that a lambda
tree is $0$-active iff there is a destructive reduction of length
$\omega$ starting from it. $0$-activeness is demonstrably different
from activeness for any strictness signature with $a_2 = 0$ as
Example~\ref{ex:counterCR} shows. But $0$-activeness and activeness do
coincide for $001$, $101$, and $111$ as we have shown with the
combination of Lemma~\ref{lem:botCompress} and
Lemma~\ref{lem:afragAact}.

\bibliographystyle{plain}
\bibliography{../common/po-lam}

\iftechrep
\appendix
\tableofcontents
\clearpage

\input{appendix}
\fi
\end{document}

%% file: prelude_paper.tex
\usepackage{hyperref}
\hypersetup{final}

\input{prelude}
\usepackage{tikz}
\usepackage{appendix}
\usepackage{amsthm}
\usepackage[shortlabels,inline]{enumitem}
\usepackage{inputenc}
\theoremstyle{theorem}\newtheorem{theorem}{Theorem}[section]
\theoremstyle{definition}\newtheorem{definition}[theorem]{Definition}
\theoremstyle{theorem}\newtheorem{proposition}[theorem]{Proposition}
\theoremstyle{theorem}\newtheorem{corollary}[theorem]{Corollary}
\theoremstyle{theorem}\newtheorem{lemma}[theorem]{Lemma}
\theoremstyle{theorem}\newtheorem{example}[theorem]{Example}

\newenvironment{thmcpy}[1]{\noindent\textbf{#1}.\itshape}{}
\numberwithin{theorem}{section}
\newcommand\nin{\not\in}

\usepackage[obeyDraft,colorinlistoftodos]{todonotes}
\newcommand\continue[1][]{\todo{{\bf continue here}}}%

\usetikzlibrary{arrows,backgrounds,positioning,fit,shapes,calc}
\usetikzlibrary{%
  decorations.pathmorphing,%
  decorations.pathreplacing,%
  decorations.markings,%
  decorations.shapes,%
  decorations.footprints,%
  decorations.fractals,%
  decorations.text%
}%

\tikzset{%
  shorten/.style={shorten >=#1, shorten <=#1},%
  level distance=1cm,%
  edge from parent/.style={->,draw},%
  force edge/.style={edge from parent/.style={->,draw}},%
  no edge/.style={edge from parent/.style={}},%
  >=latex',%
  label distance=-3pt,%
  site/.style={%
    shape=circle,%
    minimum size=4mm%
  },%
  red site/.style={%
    site,%
    draw=redsite%
  },%
  term back/.style={%
    draw=termfringe,%
    fill=termback%
  },%
    term const/.style={%
    draw=termconstfringe,%
    fill=termconst%
  },%
  term stable/.style={%
    draw=termstablefringe,%
    fill=termstable%
  },%
  term fringe/.style={%
    draw=termfringe,%
  },%
  fun/.style={>=to,->},%
  single step/.style={>=to,->},%
  conv/.style={>=to,<->},%
  strongly/.style={>=to,->>},
  weakly/.style={>=to,right hook->},%
  weakly left/.style={>=to,left hook->}%
}%


%% file: prelude.tex
\usepackage{amsmath}
\usepackage{amssymb}
\usepackage{stmaryrd}
\usepackage[bbgreekl]{mathbbol}

\DeclareFontFamily{U} {MnSymbolA}{}
\DeclareFontShape{U}{MnSymbolA}{m}{n}{
  <-6> MnSymbolA5
  <6-7> MnSymbolA6
  <7-8> MnSymbolA7
  <8-9> MnSymbolA8
  <9-10> MnSymbolA9
  <10-12> MnSymbolA10
  <12-> MnSymbolA12}{}
\DeclareFontShape{U}{MnSymbolA}{b}{n}{
  <-6> MnSymbolA-Bold5
  <6-7> MnSymbolA-Bold6
  <7-8> MnSymbolA-Bold7
  <8-9> MnSymbolA-Bold8
  <9-10> MnSymbolA-Bold9
  <10-12> MnSymbolA-Bold10
  <12-> MnSymbolA-Bold12}{}

\DeclareSymbolFont{MnSyA} {U} {MnSymbolA}{m}{n}

\DeclareMathSymbol{\upmodels}{\mathrel}{MnSyA}{225}

\newcommand{\setcom}[2]{\set{#1\left\vert\vphantom{#1}\,#2\right.}}
\newcommand{\set}[1]{\left\{#1\right\}}
\newcommand\card[1]{\left|#1\right|}

\newcommand\fcolon{\colon}
\newcommand\funto{\rightarrow}
\newcommand\pfunto{\rightharpoonup}

\newcommand\limto{\rightarrow}
\newcommand\prefix[2]{#1|_{#2}}
\newcommand\segm[3]{#1|_{[#2,#3)}}

\newcommand\desc[2]{#1/\!/#2}
\newcommand\proj[2]{#1/#2}

\def\nothing{}

\newcommand\botbb{\upmodels}
\newcommand\botred[1][\nothing]{\botbb\ifthenelse{\equal{#1}{\nothing}}{}{\!\left({#1}\right)}}
\newcommand\bohmred[1][\nothing]{\mathbb{B}\ifthenelse{\equal{#1}{\nothing}}{}{\left({#1}\right)}}
\newcommand\betared{\bbbeta } \newcommand\etared{\bbeta }
\newcommand\strictred{\mathbb{S}}
\newcommand\betasred{\betared\strictred}
\newcommand\betaSred{\betared(\strictred)}

\newcommand\vartree[1]{\mathsf{#1}}
\newcommand\abstree[1]{\vartree{\lambda #1}.}
\newcommand\aact[1][\ola]{\calA^{#1}}
\newcommand\afrag[1][\ola]{\calF^{#1}}

\newcommand\lamtree[1]{\left\llbracket#1\right\rrbracket}

\newcommand\pos[1]{\calP(#1)}
\newcommand\dom[1]{\calD(#1)}
\newcommand\domBot[1]{\calD_\bot(#1)}
\newcommand\range[1]{\mathsf{Range}(#1)}
\newcommand\allpos{\calP}

\newcommand\pathpos[1]{\mathsf{pos}(#1)}
\newcommand\pathlab[1]{\mathsf{lab}(#1)}

\newcommand\prepaths[2]{\calP(#1,#2)}
\newcommand\paths[2]{\calP_\calT(#1,#2)}
\newcommand\dpaths[2]{\calP_\bot(#1,#2)}
\newcommand\pathLabs[2]{\calP\calL(#1,#2)}

\newcommand\strictnf[1]{#1{\downarrow_\strictred}}
\newcommand\doms[1]{\calD_\strictred(#1)}

\newcommand\posminus[2]{#1\setminus #2}
\newcommand\acut[2][\ola]{#2{\downarrow^{\!#1}}}

\newcommand\concat{\cdot}
\newcommand{\Concat}{\prod}
\newcommand\seq[1]{\langle #1 \rangle}
\newcommand\emptyseq{\seq{}}
\newcommand\len[1]{\left| #1 \right|}

\newcommand\oh{\widehat}
\newcommand\ol{\overline}

\newcommand\ola{\ol a}

\newcommand\iptree{\calT^\infty_\bot}
\newcommand\tree{\calT}
\newcommand\ptree{\calT_\bot}
\newcommand\itree{\calT^\infty}
\newcommand\aptree[1][\ola]{\calT^{#1}_\bot}
\newcommand\atree[1][\ola]{\calT^{#1}}

\newcommand\lam{\Lambda}
\newcommand\plam{\Lambda_\bot}
\newcommand\iplam[1][\ola]{\Lambda^{I,#1}_\bot}

\newcommand\mplam[1][\ola]{\Lambda^{M,#1}_\bot}

\newcommand\subtree[2]{#1|_{#2}}
\newcommand\subst[3]{#1\left[#2/#3\right]}
\newcommand\rename[3]{#1[#2\to#3]}

\newcommand\atrunc[3][\ola]{#2|^{#1}_{#3}}%

\newcommand\adepth[2][\ola]{\left| #2 \right|^{#1}}
\newcommand\lebot{\le_\bot}
\newcommand\alebot[1][\ola]{\le_\bot^{#1}}
\newcommand\nalebot[1][\ola]{\not\le_\bot^{#1}}

\newcommand\ahgt[2][\ola]{\mathsf{hgt}^{#1}\left(#2\right)}
\newcommand\hgt[1]{\ahgt[]{#1}}

\newcommand\andown[3][\ola]{{\downarrow}^{#2}_{#1}\left(#3\right)}
\newcommand\anidown[3][\ola]{{\Downarrow}^{#2}_{#1}\left(#3\right)}

\newcommand\talebot[1][\ola]{\trianglelefteq_\bot^{#1}}
\newcommand\talbot[1][\ola]{\vartriangleleft_\bot^{#1}}

\newcommand\dd{\mathbf{d}}
\newcommand\dda[1][\ola]{\dd^{#1}}
\newcommand\ddi[1][\ola]{\dd^{#1}_I}
\newcommand\ddac[1][\ola]{\dd^{#1}}
\newcommand\ddta[1][\ola]{\dd^{#1}_\tree}

\newcommand\findown[1]{#1{\downarrow_{fin}}}

\newcommand\Lub{\bigsqcup}
\newcommand\Glb{\bigsqcap}

\newcommand\calA{\mathcal{A}}

\newcommand\calD{\mathcal{D}}

\newcommand\calF{\mathcal{F}}

\newcommand\calI{\mathcal{I}}

\newcommand\calL{\mathcal{L}}

\newcommand\calP{\mathcal{P}}

\newcommand\calT{\mathcal{T}}
\newcommand\calU{\mathcal{U}}
\newcommand\calV{\mathcal{V}}

\newcommand\lamsig{\calL}

\newcommand\prs{\mathsf{p}}
\newcommand\mrs{\mathsf{m}}

\let\oldTo\to

\newcommand\preright{\mapsto}
\newcommand\preleft{\mapsfrom}
\newcommand\finleftright{\leftrightarrow}
\newcommand\finright{\oldTo}
\newcommand\finleft{\leftarrow}
\newcommand\weakright{\hookrightarrow}
\newcommand\weakleft{\hookleftarrow}
\newcommand\strongright{\twoheadrightarrow}
\newcommand\strongleft{\twoheadleftarrow}
\newcommand\mrsright{\mathrel{\twoheadrightarrow^{\hspace{-10pt}\mrs}\hspace{3pt}}}
\newcommand\mrsleft{\mathrel{\twoheadleftarrow^{\hspace{-7pt}\mrs}\hspace{0pt}}}
\newcommand\prsright{\mathrel{\twoheadrightarrow^{\hspace{-9pt}\prs}\hspace{4pt}}}
\newcommand\prsleft{\mathrel{\twoheadleftarrow^{\hspace{-5pt}prs}\hspace{0pt}}}

\newcommand\mrswright{\mathrel{\hookrightarrow^{\hspace{-9pt}\mrs}\hspace{1pt}}}
\newcommand\mrswleft{\mathrel{\hookleftarrow^{\hspace{-7pt}\mrs}\hspace{0pt}}}
\newcommand\prswright{\mathrel{\hookrightarrow^{\hspace{-7pt}\prs}\hspace{2pt}}}
\newcommand\prswleft{\mathrel{\hookleftarrow^{\hspace{-5pt}\prs}\hspace{0pt}}}

\newcommand{\RewArr}[2] {
  \RewStmt{#1}{\nothing}{#2}
}

\newcommand{\RewStmt}[3] {
  \def\RewArrArr{#1}
  \def\RewArrRhs{#2}
  \def\RewArrIter{#3}
  \RewArrI
}

\makeatletter \newcommand{\RewArrI}[1][\nothing] { \def\RewArrCxt{#1}
  \ifthenelse{\equal{\RewArrArr}{\oldTo} \OR
    \equal{\RewArrArr}{\preright} \OR
    \equal{\RewArrArr}{\finright} \OR
    \equal{\RewArrArr}{\finleftright} \OR
    \equal{\RewArrArr}{\prsright} \OR
    \equal{\RewArrArr}{\mrsright} \OR
    \equal{\RewArrArr}{\prswright} \OR
    \equal{\RewArrArr}{\mrswright} \OR
    \equal{\RewArrArr}{\strongright} \OR
    \equal{\RewArrArr}{\weakright}} {
    \RewArrArr\ifthenelse{\equal{\RewArrIter}{\nothing}}{}{^{\RewArrIter}}\ifthenelse{\equal{\RewArrCxt}{\nothing}}{}{_{\RewArrCxt}}
  } { \ifthenelse{\equal{\RewArrArr}{\finleft} \OR
      \equal{\RewArrArr}{\preleft} \OR
      \equal{\RewArrArr}{\prsleft} \OR
      \equal{\RewArrArr}{\mrsleft} \OR
      \equal{\RewArrArr}{\prswleft} \OR
      \equal{\RewArrArr}{\mrswleft} \OR
      \equal{\RewArrArr}{\strongleft} \OR
      \equal{\RewArrArr}{\weakleft}}{
      \RewArrArr\ifthenelse{\equal{\RewArrIter}{\nothing}}{}{^{\RewArrIter}}\ifthenelse{\equal{\RewArrCxt}{\nothing}}{}{_{\RewArrCxt}}
    } { \latex@error{Rewrite arrow not defined}\@ehc } } \RewArrRhs }
\makeatother

\newcommand{\twoheadleftrightarrow}{\twoheadleftarrow\hspace{-7pt}\twoheadrightarrow}
\renewcommand{\to}{\RewArr{\finright}{\nothing}}

\newcommand{\fto}{\RewArr{\finright}}

\newcommand{\pato}{\RewArr{\prsright}{\nothing}}

\newcommand{\mato}{\RewArr{\mrsright}{\nothing}}

\newcommand{\mwato}{\RewArr{\mrswright}{\nothing}}

\newcommand{\pwato}{\RewArr{\prswright}{\nothing}}

\newcommand{\join}[1][\nothing]
{
  \ifthenelse{\equal{#1}{\nothing}}{\downarrow}{\downarrow_{#1}}
}

\newcommand{\wconv}[1][\nothing]
{
  \ifthenelse{\equal{#1}{\nothing}}{\leftrightarrow^{w}}{\leftrightarrow^{w}_{#1}}
}

\newcommand{\sconv}[1][\nothing]
{
  \ifthenelse{\equal{#1}{\nothing}}{\twoheadleftrightarrow}{\twoheadleftrightarrow_{#1}}
}



%% file: appendix.tex
\section{Full Proofs}
\label{sec:full-proofs}

\subsection{Ideal Completion}
\label{sec:ideal-completion-1}

\begin{proposition}
  \label{prop:lamtreeBij}
  The function $\lamtree\cdot\colon \plam\to\ptree$ is a bijection.
\end{proposition}
\begin{proof}[Proof Proof of Proposition~\ref{prop:lamtreeBij}]
  For injectivity assume some $M,N \in \plam$ such that $\lamtree M =
  \lamtree N$. We proceed by induction on $M$.  If $M = \bot$, then
  also $N = \bot$. Likewise, if $M=x \in \calV$, then $N = x$, too.

  If $M = \lambda x . M'$, then $N = \lambda y. N'$. W.l.o.g.\ we may
  assume that $x = y$ (otherwise we can just rename both to a fresh
  variable $z$). Hence, $\lamtree{M'} = \lamtree{N'}$ and thus, by
  induction hypothesis, $M' = N'$. Consequently, we have that $\lambda
  x . M' = \lambda y. N'$.

  If $M = M_1M_2$, then $N = N_1N_2$ and $\lamtree{M_i} =
  \lamtree{N_i}$. By applying the induction hypothesis to the latter,
  we obtain that $M_i = N_i$ and thus $M_1M_2 = N_1N_2$.

  For surjectivity we assume some $t \in \ptree$ and construct by
  induction on the cardinality of $\dom{t}$ a term $M \in \plam$ with
  $\lamtree M = t$. If $\dom{t} = \emptyset$, then $\lamtree{\bot} =
  t$. Otherwise, we know that $\emptyseq \in \dom{t}$. If
  $t(\emptyseq) = x$, then $\lamtree{x} = t$. The case $t(\emptyseq)
  \in \allpos$ is not possible.

  If $t(\emptyseq) = \lambda$, then construct the lambda tree
  $s$ as follows:
  \[
  s(p) =
  \begin{cases}
    x &\text{if } t(\seq 0 \concat p) = \emptyseq\\
    q &\text{if } t(\seq 0 \concat p) = \seq 0 \concat q\\
    t(\seq 0 \concat p) &\text{otherwise}\\
  \end{cases}
  \]
  where $x$ is a fresh variable not occurring in the image of $t$. One
  can easily check that $s$ is indeed a lambda tree. Since
  $\card{\dom{s}} = \card{\dom{t}} - 1$, we can apply the induction
  hypothesis to obtain some $M \in \plam$ with $\lamtree M = s$. We
  then have that $\lamtree{\lambda x . M} = t$.

  If $t(\emptyseq) = @$, then construct for each $i\in\set{1,2}$ a
  lambda tree $s_i$ as follows:
  \[
  s_i(p) =
  \begin{cases}
    q &\text{if } t(\seq i \concat p) = \seq i \concat q\\
    t(\seq i \concat p) &\text{otherwise}\\
  \end{cases}
  \]
  One can easily check that both $s_1$ and $s_2$ are lambda
  terms. Since $\card{\dom{s_1}} + \card{\dom{s_2}} = \card t - 1$, we
  may use the induction hypothesis for $s_i$. Hence, we find $M_i \in
  \plam$ with $\lamtree{M_i} = s_i$ and we can conclude that
  $\lamtree{M_1 M_2} = t$.
\end{proof}

Before we proceed, we give the explicit proof that $\talebot$ is
indeed a partial order as this fact is needed in the proof of
Proposition~\ref{prop:lamtreeIso}.
\begin{proposition}
  \label{prop:talebotPO}
  For each strictness signature $\ola$, the relation $\talebot$ is
  a partial order on $\iptree$. 
\end{proposition}
\begin{proof}
  Reflexivity and antisymmetry of $\talebot$ follow immediately from
  the definition. For transitivity, let $t_1 \talebot t_2$ and
  $t_2\talebot t_3$. Then conditions \ref{item:talebot1} and
  \ref{item:talebot2} for $t_1 \talebot t_3$ follow immediately. For
  \ref{item:talebot3}, let $a_i = 0$, $p \in \dom{t_1}$ and
  $p\concat\seq i \in \dom{t_3}$. Then $p \in \dom{t_2}$ due to $t_1
  \talebot t_2$, which in turn implies $p\concat\seq i \in \dom{t_2}$
  due to $t_2\talebot t_3$. Hence, $p\concat\seq i \in \dom{t_1}$ due
  to $t_1 \talebot t_2$.
\end{proof}

\begin{thmcpy}{Proposition~\ref{prop:lamtreeIso}}
  The function $\lamtree\cdot\colon\plam \to \ptree$ is an order isomorphism
  from $(\plam,\alebot)$ to $(\ptree,\talebot)$.
\end{thmcpy}
\begin{proof}[Proof Proof of Proposition~\ref{prop:lamtreeIso}]
  By Proposition~\ref{prop:lamtreeBij} it remains to be shown that $M
  \alebot N$ iff $\lamtree M \talebot \lamtree N$ for all $M,N \in
  \plam$.

  We show that the relation $R$ defined by $(M, N) \in R \iff \lamtree
  M \talebot \lamtree N$ has the properties given in
  Definition~\ref{def:alebot}. Since $\alebot$ is the least such
  relation, we obtain the ``only if'' direction of the
  equivalence. The relation $R$ is a preorder since $\talebot$ is a
  preorder according to Proposition~\ref{prop:talebotPO}.

  We trivially have $\lamtree \bot \talebot \lamtree M$ since
  $\dom{\lamtree{\bot}} = \emptyset$. If $a_0 = 1$ or $M \neq \bot$,
  we can easily check that $\lamtree M \talebot \lamtree N$ implies
  $\lamtree {\lambda x . M} \talebot \lamtree {\lambda x . N}$, using
  Definition~\ref{def:lamtree}. The same goes for the remaining two
  closure properties in Definition~\ref{def:alebot}.

  For the converse direction, we assume that $\lamtree{M} \talebot
  \lamtree{N}$ and show that then $M \alebot N$ by induction on
  $M$. If $M = \bot$, we immediately have $M \alebot N$ by
  \ref{item:alebot1} of Definition~\ref{def:alebot}. If $M = x$, then
  $\lamtree{M} \talebot \lamtree{N}$ implies, by \ref{item:talebot2}
  of Definition~\ref{def:talebot}, that $N = x$. By reflexivity of
  $\alebot$, we thus have $M \alebot N$.

  If $M = M_1 M_2$, then $\lamtree{M} \talebot \lamtree{N}$ implies,
  by \ref{item:talebot2} of Definition~\ref{def:talebot}, that $N$ is
  of the form $N_1 N_2$. Using the definition of $\lamtree{M}$ (resp.\
  $\lamtree{N}$) in terms of $\lamtree{M_1}$ and $\lamtree{M_2}$
  (resp.\ $\lamtree{N_1}$ and $\lamtree{N_2}$), we can derive that
  $\lamtree{M_i} \talebot \lamtree{N_i}$ for all $i \in \set{1,2}$. By
  induction hypothesis, we thus have that $M_i \alebot N_i$ for all $i
  \in \set{1,2}$. Moreover, if $a_i = 0$ for $i\in{1,2}$, then we can
  derive from $\lamtree{M} \talebot \lamtree{N}$, using
  \ref{item:talebot3} of Definition~\ref{def:talebot}, that $\emptyseq
  \in \lamtree{N_i}$ implies $\emptyseq \in \lamtree{M_i}$. That is,
  $M_i = \bot$ implies $N_i = \bot$. Consequently, we have that
  $M_1M_2 \alebot N_1 M_2$ due to \ref{item:alebot3} of
  Definition~\ref{def:alebot} in case $a_1 = 1$, or $a_1 = 0$ and $M_1
  \neq \bot$. In case, $a_1 = 0$ and $M_1 = \bot$, we know that also
  $N_1 = \bot$. Thus, $M_1M_2 \alebot N_1 M_2$ follows by
  reflexivity. By the same argument, also $N_1M_2 \alebot N_1 N_2$
  holds, which means that by transitivity, we obtain that $M_1M_2
  \alebot N_1 N_2$.

  If $M = \lambda x. M'$, then $\lamtree{M} \talebot \lamtree{N}$
  implies, by \ref{item:talebot2} of Definition~\ref{def:talebot},
  that $N$ is of the form $\lambda y. N'$ and w.l.o.g.\ we may assume
  that $y = x$. By the same argument as above, we derive from
  $\lamtree{M} \talebot \lamtree{N}$, that $\lamtree{M'} \talebot
  \lamtree{N'}$, which, by induction hypothesis, yields $M' \alebot
  N'$. Likewise we obtain, in case that $a_0 = 0$ that $\emptyseq \in
  \lamtree{N'}$ implies $\emptyseq \in \lamtree{M'}$, which means that
  $M' = \bot$ implies $N' = \bot$. Hence, if $a_0 = 0$ and $M' =
  \bot$, we obtain that $N' = \bot$, which means that $\lambda x. M'
  \alebot \lambda y. N'$ follows by reflexivity. Otherwise, we may
  apply \ref{item:alebot2} of Definition~\ref{def:alebot} to obtain
  $\lambda x. M' \alebot \lambda y. N'$.
\end{proof}

\begin{thmcpy}{Theorem~\ref{thr:alebotCpo}}
  For each strictness signature $\ola$, the partially ordered set
  $(\aptree,\talebot)$ forms a complete partial order. In particular,
  the lub $t$ of a directed set $D$ satisfies the following:
  \begin{align*}
    \dom{t} = \bigcup_{s \in D} \dom{s}\hspace{2cm}
    s(p) = t(p) \quad \text{ for all } s\in D, p\in \dom{s}
  \end{align*}
\end{thmcpy}
\begin{proof}[Proof Proof of Theorem~\ref{thr:alebotCpo}]
  The lambda tree $\bot$ is obviously the least element in
  $(\aptree,\talebot)$. To show that $(\aptree,\talebot)$ is directed
  complete, we assume a directed set $D$ in $(\aptree,\talebot)$ and
  construct a lambda tree $t \in \aptree$ that is the lub of
  $D$. Define $t$ as follows: $t(p) = s(p)$ iff there is some $s \in
  D$ with $p \in \dom{s}$.
  \begin{itemize}
  \item At first we show that this indeed defines a partial function
    $t\colon \allpos\pfunto \lamsig$. To this end assume $s_1,s_2\in D$ and
    $p\in \dom{s_1} \cap \dom{s_2}$.  Since $D$ is directed, there is
    some $s \in D$ with $s_1,s_2 \talebot s$, which implies $s_1(p)
    = s(p) = s_2(p)$.
  \item Next, we show that $t$ is a well-defined lambda tree
    according to Definition~\ref{def:lambda}:
    \begin{enumerate}
    \item[\ref{item:talebot1}] If $p\concat\seq 0 \in \dom{t}$, then
      $p\concat\seq 0 \in \dom{s}$ for some $s\in D$. Hence, $s(p) =
      \lambda$ and thus $t(p) = \lambda$.
    \item[\ref{item:talebot2}] If $p\concat\seq 1 \in \dom{t}$ or
      $p\concat\seq 2 \in \dom{t}$, then $p\concat\seq 1 \in \dom{s}$
      or $p\concat\seq 2 \in \dom{s}$ for some $s \in D$. Hence, $s(p)
      = @$ and thus $t(p) = @$.
    \item [\ref{item:talebot3}] If $t(p) = q \in \allpos$, then
      $s(p) = q$ for some $s \in D$. Hence, $q \le p$ and $s(q) =
      \lambda$, which implies $t(q) = \lambda$.
    \end{enumerate}
  \item Next, we show that $t$ is $\ola$-guarded and thus member of
    $\aptree$. Assume that $t$ is $\ola$-unguarded, i.e.\ $t$ has an
    $\ola$-bounded infinite branch $(m_i)_{i<\omega}$. That means $p_j
    \in \dom t$ for all $j<\omega$ where $p_j = (m_i)_{i<j}$ and there
    is some $n < \omega$ such that $a_{m_i} = 0$ for all $n \le i <
    \omega$. Consequently, we find for each $j < \omega$ some $s_j \in
    D$ such that $p_j \in \dom{s_j}$. We will show by induction on $j$
    that $p_j \in \dom{s_n}$ for all $j < \omega$. From this we can
    then conclude that $(m_i)_{i<\omega}$ is an $\ola$-bounded
    infinite branch in $s_n$, which means that $s_n \nin
    \aptree$. This contradicts the assumption that $D \subseteq
    \aptree$, and we can thus conclude that $t$ is $\ola$-guarded.

    The case $j \le n$ is trivial since $p_n \in \dom{s_n}$ and thus
    $p_j \in \dom{s_n}$. For the case $n < j + 1 < \omega$, we have
    that $p_{j + 1} = p_j \concat \seq{m_j}$ with $a_{m_j} = 0$. By
    induction hypothesis, we have that $p_j \in \dom{s_n}$ and since
    $D$ is directed, we find some $s\in D$ with $s_{j+1} \talebot s$
    and $s_n \talebot s$. The former yields that $p_{j+1} \in
    \dom{s}$. According to \ref{item:talebot3} of
    Definition~\ref{def:talebot}, the latter yields that $p_{j+1} \in
    \dom{s_n}$ due to $a_{m_j}=0$, $p_{j+1} \in \dom{s}$ and $p_j \in
    \dom{s_n}$.
  \item Next, we show that $t$ is an upper bound of $D$. To this end
    we assume some $s \in D$ and show that $s \talebot t$:
    \begin{enumerate}
    \item[\ref{item:talebot1}] \& \ref{item:talebot2} Immediate.
    \item[\ref{item:talebot3}] Let $a_i = 0$, $p \in \dom{s}$ and
      $p\concat\seq i \in \dom{t}$. Then there is some $s_1\in D$ with
      $p\concat\seq i \in \dom{s_1}$. As $D$ is directed, we find some
      $s_2\in D$ with $s_1 \talebot s_2$ and $s \talebot
      s_2$. The former yields that $p \concat\seq i \in \dom{s_2}$,
      which together with the latter implies that $p \concat\seq i \in
      \dom{s}$.
    \end{enumerate}
  \item Finally, we show that $t$ is the least upper bound of $D$. To
    this end, we assume some $t'$ with $s \talebot t'$ for all $s
    \in D$ and show that then $t \talebot t'$.
    \begin{enumerate}
    \item[\ref{item:talebot1}] \& \ref{item:talebot2} If $p \in \dom{t}$
      with $t(p) = l$ then there is some $s \in D$ with $s(p) =
      l$. Since $s \talebot t'$, we then obtain that $t'(p) = l$,
      too.
    \item[\ref{item:talebot3}] Let $a_i = 0$, $p \in \dom{t}$ and
      $p\concat\seq i \in \dom{t'}$. Then there is some $s \in D$ with
      $p \in \dom{s}$, which implies $p\concat\seq i \in \dom{s}$ due
      to $s \talebot t'$. Hence, $p\concat\seq i \in \dom{t}$.
    \end{enumerate}
  \end{itemize}
\end{proof}

\begin{thmcpy}{Proposition~\ref{prop:alebotGlb}}
  Every non-empty subset $T$ of $\aptree$ has a glb $\Glb T$ in
  $(\aptree,\talebot)$ such that $\dom{\Glb T}$ is the largest set $P$
  satisfying the following properties:
  \begin{enumerate}[(1)]
  \item If $p\in P$, then there is some $l \in \lamsig$ such that
    $s(p) = l$ for all $s \in T$.
    \label{item:talebotGlb1}
  \item If $p \concat\seq i \in P$, then $p \in P$.
    \label{item:talebotGlb2}
  \item If $p \in P$, $a_i = 0$, and $p \concat \seq i \in \dom s$ for
    some $s \in T$, then $p\concat\seq i \in P$.
    \label{item:talebotGlb3}
  \end{enumerate}
\end{thmcpy}
\begin{proof}[Proof Proof of Proposition~\ref{prop:alebotGlb}]
  Given a non-empty subset $T$ of $\aptree$, we construct a lambda
  tree $t \in \aptree$ and show that it is the glb of $T$.

  Let $P$ be the largest subset of $\allpos$ satisfying properties
  \ref{item:talebotGlb1} through \ref{item:talebotGlb3}.  Since these
  properties are closed under union, $P$ is well-defined.  Let $\oh s$
  be and arbitrary lambda tree in $T$. We define the partial function
  $t\colon \allpos \pfunto \lamsig$ as the restriction of $\oh s$ to
  $P$. This construction is justified since $P \subseteq \dom{\oh s}$
  by \ref{item:talebotGlb1}.
  
  \begin{itemize}
  \item At first, we show that $t$ is a well-defined lambda
    tree. For all three parts below, we make use of the fact that $P$
    is closed under taking prefixes according to
    \ref{item:talebotGlb2}.
    \begin{enumerate}
    \item[\ref{item:lambda1}] If $p \concat\seq 0 \in P$, then $p
      \concat\seq 0 \in \dom{\oh s}$ by \ref{item:talebotGlb1}. Hence,
      $t(p) = \oh s(p) = \lambda$.
    \item[\ref{item:lambda2}] If $p \concat\seq 1 \in P$ or $p
      \concat\seq 2 \in P$, then $p \concat\seq 1 \in \dom{\oh s}$ or
      $p \concat\seq 2 \in P$ by \ref{item:talebotGlb1}. Hence, $t(p) =
      \oh s(p) = @$.
    \item[\ref{item:lambda3}] If $t(p) = q \in \allpos$, then
      $\oh s(p) = q$ by \ref{item:talebotGlb1}. Hence, $q \le p$ and
      $t(q) = \oh s(q) = \lambda$.
    \end{enumerate}
  \item Next, we show that $t$ is $\ola$-guarded and thus member of
    $\aptree$. To this end, we assume that $t$ is $\ola$-unguarded. That
    is, $t$ has an $\ola$-bounded infinite branch $S$. By
    \ref{item:talebotGlb1}, each position along $S$ is also in $\oh s$,
    which means that $S$ is an infinite branch of $\oh s$ as
    well. Hence, $\oh s$ is $\ola$-unguarded, too. Since this
    contradicts the assumption that $T \subseteq \aptree$, we can
    conclude that $t$ is $\ola$-guarded.
  \item Next, we show that $t$ is a lower bound of $T$. To this end,
    we assume some $s \in T$ and show that then $t \talebot s$:
    \begin{enumerate}
    \item[\ref{item:talebot1}] Immediate consequence of the
      construction of $t$.
    \item[\ref{item:talebot2}] If $p \in P$, then $t(p) = \oh s(p)
      \stackrel{\ref{item:talebotGlb1}}= s(p)$.
    \item[\ref{item:talebot3}] Immediate consequence of
      \ref{item:talebotGlb3}.
    \end{enumerate}
  \item Finally, we show that $t$ is the greatest lower bound of
    $T$. To this end, we assume some $t' \in \aptree$ with
    $t'\talebot s$ for all $s \in T$ and show that then $t'
    \talebot t$:
    \begin{enumerate}
    \item[\ref{item:talebot1}] In order to prove the inclusion
      $\dom{t'} \subseteq P$, we show that $\dom{t'}$ satisfies
      \ref{item:talebotGlb1} through \ref{item:talebotGlb3} of the
      coinductive definition of $P$: \ref{item:talebotGlb1} and
      \ref{item:talebotGlb3} follow from the fact that $t' \talebot s$
      for all $s \in T$, whereas \ref{item:talebotGlb2} follows from
      the fact that $t'$ is a lambda tree.
    \item[\ref{item:talebot2}] If $p \in \dom{t'}$, then $t'(p) = \oh
      s(p)$ since $t' \talebot \oh s$. Because $\dom{t'} \subseteq
      P$ as shown above, we know that $p \in P$. Hence, $t(p) = \oh
      s(p) = t'(p)$.
    \item[\ref{item:talebot3}] Let $a_i = 0$, $p \in \dom{t'}$, and $p
      \concat\seq i \in P$. From the latter, we obtain that $p
      \concat\seq i \in \dom{\oh s}$, which implies $p \concat\seq i
      \in \dom{t'}$ due to $t' \talebot \oh s$.
    \end{enumerate}
  \end{itemize}
\end{proof}

In the proof of Theorem~\ref{thr:alebotCompl}, we use the following
lemma, which allows us to construct ideals:

\begin{lemma}
  \label{lem:findownIdeal}
  For each $t \in \aptree$, the set $\findown{t} = \setcom{t\in\ptree}{s
    \talebot t}$ forms an ideal in $(\ptree,\talebot)$.
\end{lemma}
\begin{proof}
  By construction, $\findown t$ is downwards-closed. To argue that
  $\findown t$ is directed, we first observe that $\findown t$ is
  non-empty since $\bot \in \findown t$. Furthermore, let $s_1,s_2 \in
  \findown t$, i.e.\ $s_1, s_2 \talebot t$. Since $(\aptree,\talebot)$
  is a complete semilattice according to
  Theorem~\ref{thr:alebotCompSemi}, every set with an upper bound also
  has a lub. In particular, $\set{s_1, s_2}$ has a lub $s$ in
  $(\aptree,\talebot)$. Since $t$ is an upper bound of
  $\set{s_1,s_2}$, we have that $s \talebot t$. It only remains to be
  shown that $s$ is finite.

  Assume that $s$ is not finite. By K\"{o}nig's Lemma, there is an
  infinite branch $S$ in $s$. Moreover, since $s_1,s_2$ are finite,
  $S$ cannot be an infinite branch in $s_1$ and $s_2$. That is, there
  is some $p_1 < S$ such that $p_1 \nin
  \dom{s_1}\cup\dom{s_2}$. Moreover, since $s$ is $\ola$-guarded, $S$
  cannot be $\ola$-bounded, which means that we find some
  $p_2\concat\seq k$ with $p_1 < p_2 \concat\seq k < S$ and $a_k = 1$.

  Let $s'$ be the restriction of $s$ to
  $\setcom{p\in\dom{s}}{\text{not }p_2 < p < S}$. Clearly, $s'$ is an
  $\ola$-guarded lambda tree, too.  We show that $s' \talebot s$:
  \begin{enumerate}
  \item[\ref{item:talebot1}] \& \ref{item:talebot2} follow from the
    construction of $s'$.
  \item[\ref{item:talebot3}] If $a_i = 0$, $p \in \dom{s'}$, and $p
    \concat\seq i \in \dom s$, then we know that $p \not< S$ or $p
    \not> p_2$. In the former case, also $p \concat\seq i \not< S$. In
    the latter case, if $p \concat \seq i > p_2$, then $p =
    p_2$. Consequently, $p \concat\seq i \not< S$ since otherwise $a_i
    = a_k = 1$. For either case, we conclude that $p \concat\seq i \in
    \dom{s'}$.
  \end{enumerate}
  Since $p_2 \concat\seq{i}$ is in $\dom{s}$ but not in $\dom{s'}$, we
  know that $s \neq s'$ and thus $s' \talbot s$.

  Next, we show that $s_j \talebot s'$ for all $j \in \set{1,2}$:
  \begin{enumerate}
  \item[\ref{item:talebot1}] If $p\in \dom{s_j}$, then $p \in \dom{s}$
    since $s_j \talebot s$. Since $p \in \dom{s_j}$ and thus $p \not\ge
    p_1$ and a fortiori $p \not\ge p_2$, we have that $p \in
    \dom{s'}$.
  \item[\ref{item:talebot2}] If $p\in \dom{s_j}$, then $s_j(p) = s(p)$
    due to $s_j \talebot s$. Moreover, by construction of $s'$, we
    obtain $s'(p) = s(p) = s_j(p)$.
  \item[\ref{item:talebot3}] If $a_i = 0$, $p \in \dom{s_j}$, and $p
    \concat\seq i \in \dom{s'}$, then $p \concat\seq i \in \dom{s}$,
    too. Since, $s_j \talebot s$, we can then conclude that
    $p\concat\seq i \in \dom{s_j}$.
  \end{enumerate}
  This contradicts the fact that $s$ is the lub of $s_1, s_2$. Hence,
  $s$ must be finite.
\end{proof}

\begin{thmcpy}{Theorem~\ref{thr:alebotCompl}}
  The ideal completion of $(\plam,\alebot)$ is order isomorphic to
  $(\aptree,\talebot)$.
\end{thmcpy}
\begin{proof}[Proof Proof of Theorem~\ref{thr:alebotCompl}]
  By Proposition~\ref{prop:lamtreeIso}, it suffices to show that the
  ideal completion $(I,\subseteq)$ of $(\ptree,\talebot)$ is order
  isomorphic to $(\aptree,\talebot)$. To this end, we define two
  functions $\phi\colon \aptree \funto I$ and $\psi\colon I \funto
  \aptree$:
  \[
  \phi(t) = \setcom{s\in\ptree}{s \talebot t} \qquad\psi(T) = \Lub T
  \]
  By Lemma~\ref{lem:findownIdeal}, $\phi$ is well-defined. Moreover,
  as each ideal $T$ of $(\ptree,\talebot)$ is directed and thus has a
  lub in $(\aptree,\talebot)$ according to
  Theorem~\ref{thr:alebotCpo}, $\psi$ is well-defined, too. Both
  $\phi$ and $\psi$ are obviously monotonic. Hence, it remains to be
  shown that $\phi$ and $\psi$ are inverses of each other:

  \begin{itemize}
  \item For each $T \in I$, we show $\phi(\psi(T)) \subseteq T$. If $t
    \in \phi(\psi(T))$, then $t \in \ptree$ and $t \talebot \oh t$ for
    $\oh t = \Lub T$. According to Theorem~\ref{thr:alebotCpo}, there
    is, for each $p \in \dom{\oh t}$, a $t_p \in T$ such that $t_p(p)
    = \oh t(p)$. Since $t \talebot \oh t$, we thus have that 
    \begin{equation}
      t(p) = \oh t(p) = t_p(p) \text{ for each } p \in \dom{t}.
      \tag{1}
      \label{eq:alebotCompl1}
    \end{equation}
    Moreover, as $\dom{t}$ is finite and $T$ is directed, we find some
    $s \in T$ with
    \begin{equation}
      t_p \talebot s \text{ for all } p \in \dom{t}.
      \tag{2}
      \label{eq:alebotCompl2}
    \end{equation}
    We will show that $t \talebot s$:
    \begin{enumerate}
    \item[\ref{item:talebot1}] \& \ref{item:talebot2} If $p \in \dom{t}$,
      then $t(p) \stackrel{\eqref{eq:alebotCompl1}}= t_p(p)
      \stackrel{\eqref{eq:alebotCompl2}}= s(p)$.
    \item[\ref{item:talebot3}] If $a_i = 0$, $p \in \dom{t}$, and
      $p\concat\seq{i}\in \dom{s}$, then $p\concat\seq{i}\in \dom{\oh
        t}$ since $s \talebot \oh t$. Because $t \talebot \oh t$, we can
      then conclude that $p \concat \seq i \in \dom{t}$.
    \end{enumerate}
    As $s \in T$ and $T$ is downwards-closed, $t \talebot s$ implies
    that $t \in T$.
  \item For each $T \in I$, we show $\phi(\psi(T)) \supseteq T$. If
    $t\in T$, then $t \in \ptree$ and $t \talebot \Lub T$. That is,
    $t\in \phi(\psi(T))$.
  \item For each $t \in \aptree$, we show that $\psi(\phi(t)) = t$. Let
    $\oh t = \Lub \phi(t)$. We will show that $t = \oh t$.

    If $p \in \dom{\oh t}$, then there is some $s \in \phi(t)$ with $p
    \in \dom{s}$ and $\oh t(p) = s(p)$, according to
    Theorem~\ref{thr:alebotCpo}. Since $s \talebot t$, this means that
    $\oh t(p) = s(p) = t(p)$.

    If $p \in \dom{t}$, then consider $\atrunc t d$ for $d = \adepth p
    + 1$. By construction of $\atrunc t d$, we have that $p \in
    \dom{\atrunc t d}$. According to Lemma~\ref{lem:atruncFin} and
    Lemma~\ref{lem:atruncAlebot}, $\atrunc t d$ is finite and $\atrunc
    t d \talebot t$. That is, $\atrunc t d \in \phi(t)$. Hence, we can
    employ Theorem~\ref{thr:alebotCpo} to conclude that $t(p) =
    \atrunc t d (p) = \oh t(p)$.
  \end{itemize}
\end{proof}


Next we show correspondences between the limit inferior in
$(\aptree,\talebot)$ and the limit in $(\aptree, \ddta)$, akin to the
corresponding result on first-order terms~\cite{bahr10rta2}, but with
the addition of selective strictness according to $\ola$. As the first
step, we give a direct characterisation of the limit of a converging
sequence of lambda trees:
\begin{lemma}
  \label{lem:alebotLim}
  If a sequence $(t_\iota)_{\iota<\alpha}$ converges to $t$ in
  $(\aptree,\ddta)$, then\vspace{-5pt}
  \begin{center}
    $\dom{t} = \bigcup_{\beta<\alpha}\bigcap_{\beta\le\iota<\alpha}
    \dom{t_\iota},\text{ and} \qquad t(p) = l \iff \exists \beta <
    \alpha \forall \beta \le
    \iota < \alpha\colon\; t_\iota(p) = l$
  \end{center}
\end{lemma}
\begin{proof}
  We only show one direction for each of the two equalities above. The
  other direction follows analogously. Let $t(p) = l$ and $d =
  \adepth{p} + 1$. Since $(t_\iota)_{\iota<\alpha}$ converges to $t$,
  there is some $\beta < \alpha$ such that $\ddta(t_\iota,t) < 2^{-d}$
  for all $\beta \le \iota < \alpha$. That is, $t_\iota(q) \simeq
  t(q)$ for all $q \in \allpos$ with $\adepth{q} < d$. In particular,
  $t_\iota(p) \simeq t(p)$. Since $p \in \dom{t}$, this means that $p
  \in \bigcap_{\beta\le\iota<\alpha} \dom{t_\iota}$, and since $t(p) =
  l$, we have that $t_\iota(p) = t(p) = l$ for all $\beta \le \iota <
  \alpha$.
\end{proof}

The following definition of truncations will help us to compare the
limit inferior in $(\aptree,\talebot)$ and the limit in the
corresponding metric space $(\aptree,\ddta)$:
\begin{definition}
  \label{def:atrunc}
  Given a strictness signature $\ola$, a depth $d \le \omega$, and a
  lambda tree $t\in \iptree$, the \emph{$\ola$-truncation}
  $\atrunc{t}{d}$ of $t$ at $d$ is defined as the restriction of $t$
  to the domain $\setcom{p\in \dom{t}}{\adepth{p}<d}$.
\end{definition}
The above definition of truncation is a straightforward translation of
the notion of truncation used by Arnold and Nivat~\cite{arnold80fi} to
the $\ola$-depth measures that we use here.  In the following, we make
use of the fact that the metric $\ddta$ can be characterised by
$\ddta(s,t) = 2^{-d}$ with $d = \max\setcom{d\le \omega}{\atrunc{s}{d}
  = \atrunc{t}{d}}$. This observation follows immediately from
Definition~\ref{def:atrunc}.

\begin{lemma}
  \label{lem:atruncFin}
  If $t\in \aptree$ and $d < \omega$, then $\atrunc{t}{d} \in \ptree$.
\end{lemma}
\begin{proof}
  We show the contraposition: Assume that $\atrunc{t}{d}$ is
  infinite. Then, by K\"{o}nig's Lemma, $\atrunc t d$ has an infinite
  branch $S$. By construction of $\atrunc t d$, $S$ is $\ola$-bounded,
  viz.\ by $d$. Since $\dom{\atrunc t d} \subseteq \dom{t}$, $S$ is
  also an infinite branch of $t$, which means that $t$ is
  $\ola$-unguarded.
\end{proof}

We can then derive the following proposition that characterises
$\ola$-guarded lambda trees:
\begin{proposition}
  \label{prop:adepth}
  A lambda tree is $\ola$-guarded iff it does not have infinitely many
  positions that have the same $\ola$-depth.
\end{proposition}
\begin{proof}
  The ``if'' direction follows from the fact that an $\ola$-bounded
  infinite branch has infinitely many positions of the same
  $\ola$-depth; the converse direction follows from
  Lemma~\ref{lem:atruncFin}.
\end{proof}

The $\ola$-truncation construction is monotonic w.r.t.\ $\talebot$:
\begin{lemma}
  \label{lem:atruncAlebot}
  For each $t\in \iptree$ and $d \le e \le \omega$, we have
  $\atrunc{t}{d} \talebot \atrunc{t}{e}$. In particular,
  $\atrunc{t}{d} \talebot t$.
\end{lemma}
\begin{proof}
  We show the properties \ref{item:talebot1} through \ref{item:talebot3}
  from Definition~\ref{def:talebot}:
  \begin{enumerate}
  \item[\ref{item:talebot1}] If $p \in \dom{\atrunc{t}{d}}$, then $p\in
    \dom{t}$ with $\adepth{p} < d \le e$. Hence, $p \in
    \dom{\atrunc{t}{e}}$.
  \item[\ref{item:talebot2}] If $p \in \dom{\atrunc{t}{d}}$, then
    $\atrunc{t}{d}(p) = t(p) = \atrunc{t}{e}(p)$.
  \item[\ref{item:talebot3}] If $a_i = 0$ and $p \in
    \dom{\atrunc{t}{d}}$, then $\adepth{p\concat\seq i} = \adepth p +
    a_i = \adepth p < d$, i.e.\ $p\concat \seq i \in \dom{\atrunc t
      d}$.
  \end{enumerate}
\end{proof}

The theorem below detail the two directions of the correspondence
between the limit inferior and the limit.
\vspace{5pt}

\begin{thmcpy}{Theorem~\ref{thr:limLiminf}}
  Let $(t_\iota)_{\iota<\alpha}$ be a sequence in $\aptree$.
  \begin{enumerate}[(i)]
  \item If $\lim_{\iota\limto\alpha} t_\iota = t$ in
    $(\aptree,\ddta)$, then $\liminf_{\iota\limto \alpha} t_\iota = t$
    in $(\aptree,\talebot)$.
  \item If $\liminf_{\iota\limto\alpha} t_\iota = t$ in
    $(\aptree,\talebot)$ and $t$ is total, then
    $\lim_{\iota\limto\alpha} t_\iota = t$ in $(\aptree,\ddta)$.
  \end{enumerate}
\end{thmcpy}
\begin{proof}[Proof of Theorem~\ref{thr:limLiminf}]
  \begin{enumerate}[(i)]
  \item 
    If $(t_\iota)_{\iota<\alpha}$ convergences to $t$, we find for
    each $d < \omega$ some $\beta < \alpha$ such that $\atrunc{t}{d} =
    \atrunc{t_\iota}{d}$ for all $\beta \le \iota < \alpha$. By
    Lemma~\ref{lem:atruncAlebot}, we thus have that $\atrunc{t}{d}
    \talebot t_\iota$ for all $\beta \le \iota < \alpha$, which means
    that $\atrunc{t}{d} \talebot s_{\beta}$, for $s_\beta =
    \Glb_{\beta\le\iota<\alpha} t_\iota$. This inequality implies that
    $\Lub_{d < \omega}\atrunc{t}{d} \talebot \Lub_{\beta<\alpha}
    s_{\beta}$.  The left-hand side is well-defined as the set
    $\setcom{d<\omega}{\atrunc{t}{d}}$ is directed according to
    Lemma~\ref{lem:atruncAlebot}. Moreover, by
    Theorem~\ref{thr:alebotCpo}, the left-hand side is equal to
    $t$. Since the right-hand side is by definition equal to $t' =
    \liminf_{\iota\limto\alpha} t_\iota$, we obtain the inequality $t
    \talebot t'$. Consequently, $\dom t \subseteq \dom{t'}$ and $t(p)
    = t'(p)$ for all $p \in \dom t$. It thus remains to be shown that
    $\dom{t'} \subseteq \dom t$. To this end, assume some $p \in
    \dom{t'}$. By Theorem~\ref{thr:alebotCpo}, we have that there is
    some $\beta < \alpha$ such that $p \in
    \dom{\Glb_{\beta\le\iota<\alpha} t_\iota}$ and thus $p \in
    \dom{t_\iota}$ for all $\beta\le\iota<\alpha$. Therefore, by
    Lemma~\ref{lem:alebotLim}, we have that $p \in \dom{t}$.

  \item 
    In order to prove that $(t_\iota)_{\iota<\alpha}$ converges to
    $t$, we need to show that for each $d < \omega$ there is some
    $\beta < \alpha$ such that $\atrunc{t_\iota}{d} = \atrunc{t}{d}$
    for all $\beta \le \iota< \alpha$. Let $d < \omega$. According to
    the definition of the limit inferior, $t = \Lub_{\beta<\alpha}
    s_\beta$ with $s_\beta = \Glb_{\beta\le\iota<\alpha} t_\iota$. By
    Theorem~\ref{thr:alebotCpo}, we thus know that for each $p \in
    \dom t$ there is some $\beta_p < \alpha$ such that $s_{\beta_p}(p)
    = t(p)$. Let $B = \setcom{\beta_p}{p\in\dom{\atrunc t d}}$. Since,
    according to Lemma~\ref{lem:atruncFin}, $\dom{\atrunc t d}$ is
    finite, so is $B$. Hence, $B$ has a maximal element, say
    $\beta$. Since $(s_\iota)_{\iota< \alpha}$ is monotonic w.r.t.\
    $\talebot$, we thus have that $s_{\beta_p} \talebot s_{\beta}$ for
    each $p\in \dom{\atrunc t d}$, which means that $t(p) =
    s_{\beta_p}(p) = s_{\beta}(p)$ for all $p\in \dom{\atrunc t
      d}$. By construction, $s_\beta \talebot t_\iota$ for all $\beta
    \le \iota < \alpha$, and thus $t(p) = s_\beta(p) = t_\iota(p)$ for
    all $p\in \dom{\atrunc t d}$ and $\beta \le \iota < \alpha$. We
    can therefore conclude that $\atrunc t d(p) =
    \atrunc{t_\iota}d(p)$ for all $p\in \dom{\atrunc t d}$ and $\beta
    \le \iota < \alpha$.

    Now it only remains to be shown that whenever $p\nin \dom{\atrunc
      t d}$ then $p\nin \dom{\atrunc {t_\iota} d}$ for all $\beta \le
    \iota < \alpha$. If $\adepth{p} \ge d$, then $p\nin \dom{\atrunc
      {t_\iota} d}$ trivially holds. Otherwise, if $\adepth{p} < d$,
    then $p\nin \dom{\atrunc t d}$ implies $p\nin \dom t$. Since $t$
    is total, $\emptyseq \in \dom t$, i.e.\ there is some prefix of
    $p$ in $\dom{t}$. Let $q$ be the longest such prefix. Then $q \in
    \dom{\atrunc t d}$ and, by the previous paragraph, we know that
    $t_\iota(q) = t(q)$ for all $\beta \le \iota < \alpha$. Since $q$
    is maximal in $\dom t$ and $t$ is total, we know that $t(q) \nin
    \set{@,\lambda}$ and thus $t_\iota(q) \nin \set{@,\lambda}$ for
    all $\beta \le \iota < \alpha$. Hence, $p$ is not in
    $\dom{t_\iota}$ and therefore not in $\dom{\atrunc{t_\iota}d}$.
  \end{enumerate}
\end{proof}

\subsection{Transfinite Reductions}
\label{sec:transf-reduct-1}

For the proof of Lemma~\ref{lem:reductionContextCut}, we need the
following property:
\begin{lemma}
  \label{lem:reductionStepPos}
  For each reduction step $s \to_{R,p} t$, we have that $p \in
  \dom{s}$.
\end{lemma}
\begin{proof}
  We proceed by induction on $p$. If $p = \emptyseq$, then $p \in
  \dom{s}$ follows from the restriction of rewrite systems such that
  $(l,r)\in R$ implies that $l \neq \bot$. If $p = \seq i\concat q$,
  then $p \in \dom{s}$ follows immediately by the induction
  hypothesis.
\end{proof}

\begin{lemma}
  \label{lem:talebotAcut}
  For all $s,t \in \iptree$ with $s \talebot t$ and for all $p \in
  \dom t$, we have that $\acut p \in \dom s$ implies $p \in \dom s$.
\end{lemma}
\begin{proof}
  We show that if $\acut p \in \dom s$, then all $q$ with $\acut p \le
  q \le p$ are in $\dom{s}$. We proceed by induction on $q$. The case
  where $q \le \acut p$ is trivial. Otherwise, $q = q' \concat \seq i$
  for some $q'$ with $\acut p \le q' \le p$. Hence, by induction
  hypothesis, $q' \in \dom{s}$. Moreover, $\acut p \le q' \le p$
  implies that $a_i = 0$. Additionally, since $p$ is in $\dom t$, so
  is its prefix $q$. Since $s \talebot t$, we thus can conclude that
  $q \in \dom{s}$.
\end{proof}

\begin{thmcpy}{Lemma~\ref{lem:reductionContextCut}}
  The reduction context of a step $s\to_{p} t$ is equal to
  $\posminus{s}{\acut p}$ and $\posminus{t}{\acut p}$.
\end{thmcpy}
\begin{proof}[Proof Proof of Lemma~\ref{lem:reductionContextCut}]
  Let $\oh c$ be the reduction context of $s\to_{p} t$. We proceed by
  induction on $p$. The case $p = \emptyseq$ is trivial since then
  $\oh c = \posminus{s}{(\acut p)} = \posminus{t}{(\acut p)} = \bot$.

  Let $p = \seq j \concat p'$. We only look at the case where $j = 0$,
  the cases $j = 1$ and $j = 2$ follow by a similar argument. If $p =
  \seq 0 \concat p'$, then $s = \abstree x s'$, $t = \abstree x t'$
  and $s' \to[p'] t'$. By induction hypothesis, the reduction context
  $c'$ of $s' \to[p'] t'$ is equal to $\posminus{s'}{\acut{p'}}$ and
  $\posminus{t'}{\acut{p'}}$.

  We consider two cases. At first, suppose $\acut p =
  \emptyseq$. Consequently, $\posminus{s}{\acut p} =
  \posminus{t}{\acut p} = \bot$. If $\oh c\neq \bot$, then $\acut p
  \in \dom{\oh c}$. Since, by Lemma~\ref{lem:reductionStepPos}, $p\in
  \dom s$, and since $\oh c \talebot s$, we can apply
  Lemma~\ref{lem:talebotAcut} to conclude that $p \in \dom{\oh
    c}$. This contradicts the definition of reduction contexts. Hence,
  $\oh c = \bot$.

  Finally, suppose that $\acut p \neq \emptyseq$. Consequently, $\acut
  p = \seq 0 \concat \acut{p'}$, which means that $\posminus s {\acut
    p} = \abstree x (\posminus{s'}{\acut{p'}})$ and $\posminus t
  {\acut p} = \abstree x (\posminus{t'}{\acut{p'}})$. Hence,
  $\posminus s {\acut p} = \abstree x c' =\posminus t {\acut p}$. We
  claim that $\oh c = \abstree x c'$. To prove this we show the
  following two statements, where $c = \abstree x c'$:
  \begin{enumerate*}[(i)]
  \item $c \talebot s, t$, and \item if $d \talebot s, t$ with $p\nin
    \dom d$, then $d \talebot c$.
  \end{enumerate*}
  
  \begin{enumerate}[(i)]
  \item We show \ref{item:talebot1}-\ref{item:talebot3} of
    Definition~\ref{def:talebot}:
    \begin{itemize}
    \item[\ref{item:talebot1}\&\ref{item:talebot2}] Let $c(q) = l$. If
      $q = \emptyseq$ then $l = \lambda$ and $s(q) = t(q) = \lambda$.

      If $q = \seq 0 \concat q'$, then we have three cases to
      distinguish:
      \begin{itemize}
      \item $l \in \lamsig\setminus\allpos$: Then $c'(q') = l$, which means
        that $s'(q') = t'(q') = l$ since $c' \talebot s',
        t'$.
      \item $l = \emptyseq$: Then $c'(q') = x$, which
        means that $s'(q') = t'(q') = x$ since $c' \talebot s',
        t'$.
      \item $l = \seq 0 \concat \ol q$: Then $c'(q') = \ol q$, which
        means that $s'(q') = t'(q') = \ol q$ since $c' \talebot s',
        t'$.
      \end{itemize}
      In all cases we can conclude that $s(q) = t(q) = l$.
    \item[\ref{item:talebot3}] Assume that $q \in \dom c$ and that $q
      \concat \seq i$ is in $\dom s \cup \dom t$ and strict. We have
      to show that $q \concat \seq i \in \dom c$.
      \begin{itemize}
      \item If $q = \emptyseq$, then $i = 0$, $a_i = 0$ and $\emptyseq
        \in \dom{s'}\cup\dom{t'}$. Since $\acut p \neq \emptyseq$ and
        $a_0 = 0$, we know that $\acut{p'} \neq \emptyseq$. Since $c'
        = \posminus{s'}{\acut{p'}} = \posminus{t'}{\acut{p'}}$, we can
        thus conclude that $\emptyseq \in \dom {c'}$. That means that
        $q \concat \seq i \in \dom c$.
      \item If $q = \seq 0 \concat q'$, then $q' \in \dom{c'}$,
        $q'\concat \seq i \in \dom{s'}\cup \dom{t'}$ and $q' \concat
        \seq i$ is strict. Since $c' \talebot s',t'$, we can
        thus conclude that $q'\concat \seq i \in \dom{c'}$, which
        means that $q\concat\seq i \in \dom c$.
      \end{itemize}
    \end{itemize}
    \item Assuming $d \talebot s, t$ and $p\nin \dom d$, we show
      \ref{item:talebot1}-\ref{item:talebot3} of
      Definition~\ref{def:talebot} to prove that $d \talebot c$:
      \begin{itemize}
      \item[\ref{item:talebot1}\&\ref{item:talebot2}] Let $d(q) = l$.
        If $\acut p \le q$, then also $\acut p \in \dom d$. Since $p
        \in \dom s$, by Lemma~\ref{lem:reductionStepPos}, and $d
        \talebot s$, we can apply Lemma~\ref{lem:talebotAcut} to
        obtain that $p \in \dom d$, which contradicts the
        assumption. Thus, we can assume that $\acut p \not\le
        q$. Consequently, $(\posminus s {\acut p})(q) = l$, which
        means that $c(q) = l$.
      \item[\ref{item:talebot1}] Assume that $q \in \dom d$ and that
        $q\concat\seq i$ is in $\dom c$ and strict. Consequently,
        since $c \talebot s$, we have that $q \concat \seq i \in \dom
        s$. The latter implies that $q \concat\seq i \in \dom d$ since
        $d \talebot s$.
    \end{itemize}
  \end{enumerate}
\end{proof}

The following property, which relates $\mrs$-convergence and
-continuity, follows from the fact that our definition of
$\mrs$-convergence on $\aptree$ instantiates the abstract notion of
(strong) $\mrs$-convergence from our previous work~\cite{bahr10rta}:
\begin{lemma}
  \label{lem:mContConv}
  Let $S = (t_\iota \to[p_\iota] t_{\iota+1})_{\iota<\alpha}$ be an
  open $\mrs$-continuous reduction. If
  $(\adepth{p_\iota})_{\iota<\alpha}$ tends to infinity, then $S$ is
  $\mrs$-convergent.
\end{lemma}
\begin{proof}
  Special case of Proposition~5.5 from \cite{bahr10rta}; also cf.\
  \cite[Thm.~B.2.5]{bongaerts11master}.
\end{proof}

\begin{thmcpy}{Proposition~\ref{prop:mconvPconv}}
  For each reduction $S\colon s \mato t$, we also have that $S\colon s
  \pato t$.%
\end{thmcpy}
\begin{proof}[Proof of Proposition~\ref{prop:mconvPconv}]
  Let $S=(t_\iota \to_{p_\iota,c_\iota}
  t_{\iota+1})_{\iota<\alpha}$. Given a limit ordinal $\gamma \le
  \alpha$, we have to show that $\liminf_{\iota<\gamma} c_\iota =
  t_\gamma$, assuming $t_\alpha = t$. By $\mrs$-convergence of $S$, we
  know that $(\adepth{p_\iota})_{\iota<\gamma}$ tends to infinity and
  thus so does $(\adepth{\acut{p_\iota}})_{\iota<\gamma}$. In other
  words, for each $d < \omega$ there is some $\delta < \gamma$ with
  $\adepth{\acut{p_\iota}} \ge d$ for all
  $\delta\le\iota<\gamma$. Since, by
  Lemma~\ref{lem:reductionContextCut}, $c_\iota =
  \posminus{t_\iota}{\acut{p_\iota}}$, we thus have that
  $\atrunc{c_\iota}d = \atrunc{t_\iota}d$ for all
  $\delta\le\iota<\gamma$. Consequently, $(c_\iota)_{\iota<\gamma}$
  converges to the same lambda tree as $(t_\iota)_{\iota<\gamma}$ if
  any. Since $\mrs$-convergence of $S$ implies that
  $\lim_{\iota\limto\gamma} t_\iota = t_\gamma$, we can therefore
  conclude that $\lim_{\iota\limto\gamma} c_\iota =
  t_\gamma$. According to Theorem~\ref{thr:limLiminf}~\ref{item:limLiminf1}, we thus have
  that $\liminf_{\iota\limto\gamma} c_\iota = t_\gamma$.
\end{proof}

The above proposition allows us to prove the following lemma that
provides a characterisation for when $\prs$-convergence implies
$\mrs$-convergence:
\begin{lemma}
  \label{lem:pConvMconv}
  Let $S \fcolon s \pato t$. Then $S \fcolon s \mato t$ iff the
  $\ola$-depth of contracted redex occurrences tends to infinity for
  each open prefix of $S$.
\end{lemma}
\begin{proof}
  The ``only if'' direction follows from the definition of
  $\mrs$-convergence. For the converse direction, we assume that
  $S=(t_\iota \to_{p_\iota} t_{\iota+1})_{\iota<\alpha}$, and show
  that $\prefix S \beta\fcolon s \mato t_\beta$ for each prefix
  $\prefix S \beta$ of $S$ by induction on its length $\beta$.  The
  case $\beta = 0$ is trivial; if $\beta$ is a successor ordinal the
  statement follows immediately from the induction hypothesis.  Let
  $\beta$ be a limit ordinal. By the induction hypothesis, each proper
  prefix $\prefix S \gamma$ of $\prefix S \beta$ $\mrs$-converges to
  $t_\gamma$, which means that $\prefix S \beta$ is
  $\mrs$-continuous. Since $(\adepth{p_\iota})_{\iota<\alpha}$ tends
  to infinity, there is, by Lemma~\ref{lem:mContConv}, some $t'_\beta$
  such that $\prefix S \beta \fcolon s \mato t'_\beta$, which implies
  $\prefix S \beta \fcolon s \pato t'_\beta$ by
  Proposition~\ref{prop:mconvPconv}. As we know that $\prefix S \beta
  \fcolon s \pato t_\beta$, we can then conclude that $t_\beta' =
  t_\beta$, and thus $\prefix S \beta \fcolon s \mato t_\beta$.
\end{proof}

\begin{thmcpy}{Proposition~\ref{prop:pconvMconv}}
  $S\colon s \pato t$ implies $S\colon s \mato t$ whenever $S$ and $t$
  are total.
\end{thmcpy}
\begin{proof}[Proof Proof of Proposition~\ref{prop:pconvMconv}]
  Let $S=(t_\iota \to_{p_\iota,c_\iota}
  t_{\iota+1})_{\iota<\alpha}$. To complete the proof it remains to be
  show that $(\adepth{p_\iota})_{\iota<\gamma}$ tends to infinity for
  each limit ordinal $\gamma\le\alpha$. To this end, we assume some
  limit ordinal $\gamma \le \alpha$ and $d < \omega$ and construct
  some $\delta < \gamma$ such that $\adepth{p_\iota} \ge d$ for all
  $\delta \le \iota < \gamma$.

  By $\prs$-convergence, we know that
  $t_\gamma=\liminf_{\iota\limto\gamma} c_\iota$, i.e.\ $t_\gamma =
  \Lub_{\beta<\gamma} s_\beta$ where $s_\beta =
  \Glb_{\beta\le\iota<\gamma} c_\iota$. By
  Theorem~\ref{thr:alebotCpo}, we find for each $p \in
  \dom{\atrunc{t_\gamma}{d}}$ some $\delta(p) < \gamma$ with $p\in
  \dom{s_{\delta(p)}}$. Since $\dom{s_{\delta'}} \subseteq
  \dom{s_{\delta''}}$ whenever $\delta' \le \delta''$ and since
  $\atrunc{t}{a}$ is finite according to Lemma~\ref{lem:atruncFin}, we
  find some $\delta < \gamma$ with $\dom{\atrunc{t_\gamma}{d}}
  \subseteq \dom{s_\delta}$, namely $\delta =
  \max\setcom{\delta(p)}{p\in\dom{\atrunc{t_\gamma} d}}$. Since, by
  definition, $s_\delta \talebot c_\iota$ for all
  $\delta\le\iota<\gamma$, we then have that
  $\dom{\atrunc{t_\gamma}{d}} \subseteq \dom{c_\iota}$ for all $\delta
  \le \iota < \gamma$. From this we derive the following for all
  $\delta \le \iota < \gamma$:\par\vspace{5pt}
  \begin{enumerate*}[(1)]
  \item $p_\iota \nin \dom{\atrunc{t_\gamma}{d}}$, and \hspace{5mm}
    \label{it:strongConvPM1}
  \item $t_\iota(p) = t_\gamma(p)$ for all
    $p\in\dom{\atrunc{t_\gamma}{d}}$.
    \label{it:strongConvPM2}
  \end{enumerate*}\\[5pt]
  \ref{it:strongConvPM1} follows from the fact that $p_\iota \nin
  \dom{c_\iota}$. For \ref{it:strongConvPM2}, assume that $p \in
  \dom{\atrunc{t_\gamma}{d}}$. Then $p \in \dom{s_\delta}$, which
  implies that $s_\delta(p) = t_\gamma(p)$ as $s_\delta \talebot
  t_\gamma$. Since $s_\delta \talebot c_\iota$, we also have that
  $s_\delta(p) = c_\iota(p)$, and since $c_\iota \talebot t_\iota$, we
  have that $c_\iota(p) = t_\iota(p)$. Altogether, we thus have that
  $t_\gamma(p) = t_\iota(p)$.

  Finally, we prove the claim that $\adepth{p_\iota} \ge d$ for all
  $\delta \le \iota < \gamma$. If $p_\iota \in \dom{t_\gamma}$, then
  $\adepth{p_\iota} \ge d$ follows immediately from
  \ref{it:strongConvPM1}. Otherwise, if $p_\iota \nin \dom{t_\gamma}$,
  we find a maximal prefix $q < p_\iota$ with $q\in
  \dom{t_\gamma}$. Because $t_\gamma$ is total, we know that
  $t_\gamma(q)\in\calV\uplus\allpos$. Assume that $\adepth{p_\iota}
  \ge d$ does not hold. Consequently, $\adepth{q} \le \adepth{p_\iota}
  < d$, which means that $q \in \dom{\atrunc{t_\gamma}{d}}$. According
  to \ref{it:strongConvPM2}, we thus obtain that $t_\iota(q) =
  t_\gamma(q)$. Hence, $t_\iota(q)\in \calV\uplus\allpos$ which means
  that $p_\iota \nin \dom{t_\iota}$, which, according to
  Lemma~\ref{lem:reductionStepPos}, contradicts the fact that there is
  a reduction step from $t_\iota$ at $p_\iota$ since $\domBot{t_\iota}
  = \emptyset$. Consequently, $\adepth{p_\iota} \ge d$.
\end{proof}

\subsection{Beta Reduction}
\label{sec:bohm-reduction}

\begin{thmcpy}{Proposition~\ref{prop:pconvMconvVolatile}}
  $S \fcolon s \mato t$ iff no prefix of $S$ has volatile positions
  and $S \fcolon s \pato t$.
\end{thmcpy}
\begin{proof}[Proof of Proposition~\ref{prop:pconvMconvVolatile}]
  Let $S = (t_\iota \to[p_\iota] t_{\iota+1})_{\iota<\alpha}$.  The
  ``only if'' direction follows from Proposition~\ref{prop:mconvPconv}
  and the fact that if $(\adepth{p_\iota})_{\iota<\beta}$ tends to
  infinity, then $\prefix{S}{\beta}$ has no volatile positions.  The
  ``if'' direction follows from Lemma~\ref{lem:pConvMconv} as the
  absence of volatile positions implies that the $\ola$-depth of
  contracted redex occurrences tends to infinity (by the infinite
  pigeonhole principle).
\end{proof}

\begin{lemma}
  \label{lem:pconvOpenPos}
  Let $S = (t_\iota \to[c_\iota,p_\iota] t_{\iota+1})_{\iota<\alpha}$
  be an open reduction $\prs$-converging to $t$. Then we have the
  following:
  \begin{enumerate}[(i)]
  \item $p \in \dom t \implies \exists \beta < \alpha\forall \beta\le\iota <
    \alpha: \acut{p_\iota}\not\le p \text{ and } t_\iota(p)=t(p)$.
    \label{item:pconvOpenPos1}
  \item $p \in \dom{t_\beta} \text{ and } \forall
    \beta\le\iota < \alpha: \acut{p_\iota}\not\le p \implies \forall
    \beta \le \iota < \alpha: t_\iota(p)=t(p)$.
    \label{item:pconvOpenPos2}
  \end{enumerate}
\end{lemma}
\begin{proof}
  \begin{enumerate}[(i)]
  \item If $p \in \dom t$, then, by Theorem~\ref{thr:alebotCpo}, there
    is some $\beta <\alpha$ such that $s(p) = t(p)$, where $s =
    \Glb_{\beta\le\iota<\alpha} c_\iota$. Since $s \talebot c_\iota$,
    we thus have that $c_\iota(p) = s(p) = t(p)$ for all
    $\beta\le\iota<\alpha$. According to
    Lemma~\ref{lem:reductionContextCut}, $c_\iota =
    \posminus{t_\iota}{\acut p_\iota}$. Consequently, $\acut p_\iota
    \not\le p$ and $t_\iota(p)=t(p)$ for all $\beta \le \iota <
    \alpha$.
  \item Let $\beta < \alpha$ and $P = \setcom{p \in
      \dom{t_\beta}}{\forall \beta \le\iota<\alpha:\acut{p_\iota}
      \not\le p}$. We show the
    following statements for all $\beta\le\gamma\le\alpha$:
    \begin{align}
      \label{eq:pconvOpenPos1}
      c_\iota(p) &= t_\beta(p) & \quad&\text{for all } p \in P, \beta
      \le \iota < \gamma\tag{A}\\
      t_\gamma(p) &= t_\beta(p) & \quad&\text{for all } p \in P
      \tag{B}
      \label{eq:pconvOpenPos2}
    \end{align}
    Then \ref{item:pconvOpenPos2} follows from
    \eqref{eq:pconvOpenPos2}. We proceed by induction on $\gamma$.

    For $\gamma = \beta$, \eqref{eq:pconvOpenPos1} is vacuously true
    and \eqref{eq:pconvOpenPos2} is trivial. 

    Let $\gamma = \gamma' + 1 > \beta$. For \eqref{eq:pconvOpenPos1},
    it remains to be shown that $c_{\gamma'}(p) = t_\beta(p)$ for all
    $p \in P$. Since, according to
    Lemma~\ref{lem:reductionContextCut}, $c_{\gamma'} =
    \posminus{t_{\gamma'}}{\acut{p_{\gamma'}}}$, we know that
    $c_{\gamma'}(p) = t_{\gamma'}(p)$ for all $p \in P$. By the
    induction hypothesis for \eqref{eq:pconvOpenPos2}, we then have
    that $c_{\gamma'}(p) = t_{\beta}(p)$ for all $p \in P$. Moreover,
    since $c_{\gamma'} \talebot t_\gamma$, we have that $t_{\gamma} (p) =
    c_{\gamma'}(p) = t_{\beta}(p)$ for all $p \in P$.

    Let $\gamma$ be a limit ordinal. Then \eqref{eq:pconvOpenPos1}
    follows immediately from the induction hypothesis. For
    \eqref{eq:pconvOpenPos2}, we observer that properties
    \ref{item:talebotGlb1} to \ref{item:talebotGlb3} from
    Proposition~\ref{prop:alebotGlb} are satisfied for the glb $s =
    \Glb_{\beta\le\iota<\gamma}c_\iota$. Property
    \ref{item:talebotGlb1} follows from the induction hypothesis for
    \eqref{eq:pconvOpenPos1} and \ref{item:talebotGlb2} is immediate
    from the construction of $P$. To see that \ref{item:talebotGlb3}
    is true, assume some $\beta\le \iota<\gamma$ and $p \in P$ with
    $p\concat\seq i \in \dom{c_\iota}$ but $p \concat\seq i \nin
    P$. Then there is some $\beta\le\iota'<\gamma$ such that
    $p\concat\seq i = \acut{p_{\iota'}}$. Consequently, $a_i=1$.
    Since \ref{item:talebotGlb1}-\ref{item:talebotGlb3} are fulfilled,
    we may apply Proposition~\ref{prop:alebotGlb}, to conclude that
    $s(p) = c_\beta(p)$ for all $p \in P$. Hence, because $s \talebot
    t_\gamma$, we have that $t_\gamma(p) = c_\beta(p)$ for all $p \in
    P$. By applying the induction hypothesis for
    \eqref{eq:pconvOpenPos1}, we can then conclude that $t_\gamma(p) =
    t_\beta(p)$ for all $p \in P$.
  \end{enumerate}
\end{proof}

\begin{corollary}
  \label{cor:pconvOpenPosBot}
  Let $S = (t_\iota \to[c_\iota,p_\iota] t_{\iota+1})_{\iota<\alpha}$
  be an open reduction $\prs$-converging to $t$. Then we have the
  following:
  \begin{enumerate}[(i)]
  \item $p \concat\seq i \in \domBot{t} \text{ and } a_i = 0 \implies
    \exists \beta < \alpha\forall \beta \le \iota < \alpha: p \concat \seq i
    \in \domBot{t_\iota}$
    \label{item:pconvOpenPosBot1}
  \item $p\concat\seq i \in \domBot{t_\beta} \text{ and } \forall
    \beta\le\iota < \alpha: \acut{p_\iota}\not\le p \implies \forall
    \beta \le \iota < \alpha: p\concat\seq i \in \domBot{t}$.
    \label{item:pconvOpenPosBot2}
  \end{enumerate}
\end{corollary}
\begin{proof}
  \begin{enumerate}[(i)]
  \item Since $p \concat\seq i \in \domBot{t}$, we know that
    $p \in \dom{t}$. Consequently, by
    Lemma~\ref{lem:pconvOpenPos}~\ref{item:pconvOpenPos1}, there is
    some $\beta < \alpha$ such that $\acut{p_\iota}\not\le p$ and
    $t_\iota(p)=t(p)$ for all $\beta \le \iota < \alpha$. Since
    $a_i = 0$, we have that $\acut{p_\iota}\not\le p\concat\seq{i}$
    for all $\beta \le \iota < \alpha$. Hence,
    $p\concat \seq i \nin \dom{t_\iota}$ for all
    $\beta \le \iota < \alpha$, because otherwise
    $p\concat \seq i \nin \dom{t}$ according to
    Lemma~\ref{lem:pconvOpenPos}~\ref{item:pconvOpenPos2}. Consequently,
    $p\concat \seq i \in \domBot{t_\iota}$ for all
    $\beta \le \iota < \alpha$.

  \item Let $i = 0$. The other cases follow by a similar argument.

    Then $t_\beta(p) = \lambda$, and by
    Lemma~\ref{lem:pconvOpenPos}~\ref{item:pconvOpenPos2}, also
    $t(p) = \lambda$. Hence,
    $p\concat \seq i \in \dom t \cup \domBot t$. By induction, we show
    below that $p\concat \seq i \in \dom t$ implies
    $p\concat \seq i \in \dom{t_\beta}$ for any reduction $S$
    (including closed ones). Since
    $p\concat \seq i \nin \dom{t_\beta}$, that must mean that
    $p\concat \seq i \in \domBot t$.

    The base case $\beta = \alpha$ is trivial. If
    $\alpha = \gamma +1$, then also
    $p\concat \seq i \in \dom {t_\gamma}$ and by induction hypothesis
    $p\concat \seq i \in \dom {t_\beta}$. If $\alpha$ is a limit
    ordinal, then $p\concat \seq i \in \dom {t_\gamma}$ for some
    $\beta \le \gamma < \alpha$ according
    Lemma~\ref{lem:pconvOpenPos}~\ref{item:pconvOpenPos1}. Then
    $p\concat \seq i \in \dom{t_\beta}$ follows by the induction
    hypothesis.
  \end{enumerate}

\end{proof}

\begin{thmcpy}{Lemma~\ref{lem:volatileBot}}
  If $p$ is outermost-volatile in $S \fcolon s \pato t$,
  then $p \in \domBot t$.
\end{thmcpy}
\begin{proof}[Proof of Lemma~\ref{lem:volatileBot}]
  
  Let $S = (t_\iota \to[p_\iota,c_\iota]
  t_{\iota+1})_{\iota<\alpha}$. Since $p$ is volatile in $S$, we find
  for each $\beta < \alpha$ some $\beta \le \iota < \alpha$ with
  $\acut{p_\iota} \le p$. Hence, by
  Lemma~\ref{lem:reductionContextCut}, we know that $p \nin
  \dom{c_\iota}$. Consequently, by Theorem~\ref{thr:alebotCpo} and
  Proposition~\ref{prop:alebotGlb}, we have that $p \nin \dom{t}$.

  If $p = \emptyseq$, then $p \in \domBot{t}$ follows immediately. If
  $p = q\concat\seq 0$, then we have to show that $t(q) = \lambda$ to
  conclude that $p \in \domBot{t}$. Since $p$ is outermost-volatile in
  $S$, we find some $\beta < \alpha$ such that $p \in \dom{t_\beta}$
  and $\acut{p_\iota} \not\le q$ for all $\beta\le \iota < \alpha$
  (the latter because otherwise some prefix of $q$ would be volatile
  in $S$). Hence, by Lemma~\ref{lem:pconvOpenPos} $t(q) =
  t_\beta(q)$. Moreover, since $q\concat\seq 0 \in \dom{t_\beta}$, we
  know that $t_\beta(q) = \lambda$. Consequently, $t(q) =
  \lambda$. The argument for the cases $p = q\concat\seq 1$ and $p =
  q\concat\seq 2$ is analogous.
\end{proof}

\begin{corollary}
  \label{cor:volatileBot}
  If $p$ is volatile in $S \fcolon s \pato t$, then $p \nin \dom t$.
\end{corollary}
\begin{proof}
  Follows from Lemma~\ref{lem:volatileBot}.
\end{proof}

\begin{thmcpy}{Proposition~\ref{prop:destructBot}}
  An open reduction is destructive iff it $\prs$-converges to $\bot$.
\end{thmcpy}
\begin{proof}
  The ``only if'' direction follows immediately from
  Lemma~\ref{lem:volatileBot}.

  For the ``if'' direction, let $S = (t_\iota \to[p_\iota]
  t_{\iota+1})_{\iota<\alpha}$. Because $\bot$ is not a redex, we know
  that $\dom{t_\iota}$ is non-empty for all $\iota < \alpha$. Since
  $S$ $\prs$-converges to $\bot$, we thus know, according to (the
  contrapositive of)
  Lemma~\ref{lem:pconvOpenPos}~\ref{item:pconvOpenPos2}, that for each
  $\beta < \alpha$ there is some $\beta\le\iota<\alpha$ such that
  $\acut{p_\iota}\le \emptyseq$. That is, $S$ is destructive.
\end{proof}

\begin{lemma}
  \label{lem:botAfrag}
  Given $t \in \aptree$ and $p \in \domBot t$ with $\adepth{p} = 0$,
  we have that $t \in \afrag_\bot$.
\end{lemma}
\begin{proof}
  Let $s$ be a $\bot$-instance of $t$ w.r.t.\ $\afrag$. Then $\subtree
  s p \in \afrag$, i.e.\ there is a destructive reduction $\subtree s
  p \pato[\betared] \bot$. By embedding this reduction into $s$ at
  position $p$, we obtain a reduction $S\fcolon s \to s'$ that has a
  volatile position $p$. Since $\adepth p = 0$, also $\emptyseq$ is
  volatile in $S$. Consequently, $s \in \afrag$, which means that $t
  \in \afrag_\bot$.
\end{proof}

\begin{lemma}
  \label{lem:openOneBetaRed}
  For any open $\betasred$-reduction $t \pato[\betasred] \bot$, we
  find an open $\betared$-reduction $t \pato[\betared] \bot$.
\end{lemma}
\begin{proof}
  We first show that any reduction $S \fcolon t \pato[\betasred] \bot$
  contracts infinitely many $\betared$-redexes at $\ola$-depth $0$.

  To this end, suppose this was not the case. Then, by
  Proposition~\ref{prop:destructBot}, $S$ contracts infinitely many
  $\strictred$-redexes at $\ola$-depth $0$. However,
  $\betasred$-reduction at $\ola$-depth $>0$ creates no new
  $\strictred$-redexes at $\ola$-depth $0$, and $\strictred$-reduction
  at $\ola$-depth $0$ creates at most one $\strictred$-redex at the
  same $\ola$-depth (but at a strictly smaller $111$-depth). Hence,
  contraction of infinitely many $\strictred$-redexes at $\ola$-depth
  $0$ requires $t$ to contain infinitely many $\strictred$-redexes at
  $\ola$-depth $0$, which is impossible by
  Proposition~\ref{prop:adepth}

  Finally we construct a $\betared$-reduction
  $T \fcolon t \pato[\betared] \bot$ from $S$ by removing all
  $\strictred$-steps. Clearly $\betared$-redexes contracted in $S$ are
  still $\betared$-redexes in $T$. Moreover, since infinitely many
  redexes at $\ola$-depth $0$ are contracted $T$ also $\prs$-converges
  to $\bot$ by Proposition~\ref{prop:destructBot}.
\end{proof}

\begin{thmcpy}{Theorem~\ref{thr:pConvBohmred}}
  If $s \pato[\betasred] t$, then $s \mato[\bohmred] t$, where
  $\bohmred = \bohmred[\afrag]$.
\end{thmcpy}
\begin{proof}[Proof of Theorem~\ref{thr:pConvBohmred}]
  Let
  $S = (\phi_\iota\fcolon t_\iota \to[p_\iota]
  t_{\iota+1})_{\iota<\alpha}$ be a $\betasred$-reduction that
  $\prs$-converges to $t_\alpha$. We construct a $\bohmred$-reduction
  $T$ from $S$ that also $\prs$-converges to $t_\alpha$ but that has
  no volatile positions in any of its open prefixes. Thus, according
  to Proposition~\ref{prop:pconvMconvVolatile}, $T$ also
  $\mrs$-converges to $t_\alpha$. The construction removes steps in
  $S$ that take place at or below outermost-volatile positions of some
  prefix of $S$ and replaces them by $\botred$-steps.  Let $p$ be an
  outermost-volatile position of some prefix $\prefix{S}{\beta}$. Then
  there is some ordinal $\gamma < \beta$ such that no reduction step
  between $\gamma$ and $\beta$ in $S$ takes place strictly above $p$,
  i.e. $p_\iota \not< p$ for all $\gamma \le \iota < \beta$. Hence, we
  can construct a destructive reduction
  $S_1\fcolon \subtree{t_\gamma}{p} \pato[\betasred] \bot$ by taking
  the subsequence of the segment $\segm S \gamma {\beta}$ that
  contains the reduction steps $\phi_\iota$ with $p \le
  p_\iota$. Moreover, by Lemma~\ref{lem:openOneBetaRed}, we find a
  $\betared$-reduction
  $S_2\fcolon \subtree{t_\gamma}{p} \pato[\betared] \bot$.

  Note that $\subtree{t_\gamma}{p}$ may not be total. However, the
  applicability of $\betared$-steps is preserved by forming
  $\bot$-instances. In particular, we can form $\bot$-instances
  w.r.t.\ $\afrag$. Let $r$ be such a $\bot$-instance of
  $\subtree{t_\gamma}{p}$ w.r.t.\ $\afrag$. Then there is a
  destructive reduction $S_3\fcolon r \pato[\betared] \bot$ that
  contracts redexes at the same positions as $S_2$. Hence,
  $r \in \afrag$, which means that
  $\subtree{t_\gamma}{p} \in \afrag_\bot$. Additionally,
  $\subtree{t_\gamma}{p} \neq \bot$ since $\subtree{t_\gamma}{p}$
  contains a $\betared$-redex. Consequently, there is a pair
  $(\subtree{t_\gamma}{p}, \bot) \in\ \botred$. Let $T'$ be the
  reduction that is obtained from $\prefix S {\beta}$ by replacing the
  $\gamma$-th step, which we can assume w.l.o.g.\ to take place at
  $p$, by a $\botred$-step at the same position $p$ and removing all
  reduction steps $\phi_\iota$ with $p \le p_\iota$ and
  $\gamma < \iota < \beta$. Let $t'$ be the lambda tree that the
  reduction $T'$ $\prs$-converges to. $t_{\beta}$ and $t'$ can only
  differ at position $p$ or below. However, by construction, we have
  $p \in \domBot{t'}$ and, by Lemma~\ref{lem:volatileBot},
  $p \in \domBot{t_{\beta}}$, too. Consequently, $t' = t_\beta$.
\end{proof}

\begin{lemma}
  \label{lem:botredEventual}
  If $\bohmred = \bohmred[\afrag]$ and $s \mato[\bohmred] t$, then $s
  \mato[\betared] s'$ and $s' \mato[\botred] t$.
\end{lemma}
\begin{proof}
  According to Lemma~27 of Kennaway et al.~\cite{kennaway99jflp}, this
  property holds for the metric $\ddta[111]$ for all $\bohmred[\calU]$
  given $\calU$ is closed under substitution. The proof works for all
  other metrics of the form $\ddta$ as well, and $\afrag$ is clearly
  closed under substitution: given a total, fragile lambda tree $t$
  witnessed by the destructive $\betared$-reduction $S$ and a
  substitution $\sigma$ (of total lambda trees), then $\sigma(t)$ is
  also fragile witnessed by the reduction obtained from $S$ by
  applying $\sigma$ to each of its lambda trees.
\end{proof}

\begin{lemma}
  \label{lem:botredOmega}
  Let $\botred = \botred[\calU]$ for some $\calU\subseteq \atree$ and
  $S\fcolon s \mato[\botred] t$. Then there is a reduction $T\fcolon s
  \mato[\botred] t$ of length at most $\omega$ contracting disjoint
  $\botred$-redexes of $s$.
\end{lemma}
\begin{proof}
  Let $S = (t_\iota\to[p_\iota] t_{\iota+1})_{\iota<\alpha}$, and let
  $P$ be the set of outermost positions of redexes contracted in $S$,
  i.e.\ $P = \setcom{p_\beta}{\beta < \alpha, \forall \iota <
    \alpha\fcolon p_\iota \not< p_\beta}$. Then, for each $p \in P$,
  also $\subtree s p$ is a $\botred$-redex by the definition of
  $\calU_\bot$. Let $T$ be the reduction that contracts all redexes in
  $s$ at positions $P$. By Proposition~\ref{prop:adepth}, this yields
  a reduction $\mrs$-converging to some lambda tree $t'$. Since all
  redexes that are contracted in $S$ are subtrees of some redex
  contracted in $T$, we have that $t' = t$.
\end{proof}

\begin{thmcpy}{Theorem~\ref{thr:bohmredPConv}}
  Let $\bohmred = \bohmred[\afrag]$ and $s \mato[\bohmred] t$ such
  that $s$ is total. Then $s \pato[\betared] t$ if $\ola = 111$ or $t$
  is a $\botred$-normal form.
\end{thmcpy}
\begin{proof}[Proof of Theorem~\ref{thr:bohmredPConv}]
  By Lemmas~\ref{lem:botredEventual} and \ref{lem:botredOmega}, we
  find reductions $S\fcolon s \mato[\betared] s'$ and $T\fcolon s'
  \mato[\botred] t$, where $T$ contracts disjoint $\botred$-redexes in
  $s'$. By Proposition~\ref{prop:mconvPconv}, we have that $S\fcolon s
  \pato[\betared] s'$ and that $T\fcolon s' \pato[\botred] t$. Let $u
  \to[\botred,p] v$ be a step in $T$. Then $\subtree u p \in
  \afrag_\bot$. Because $s$ is total, so is $s'$. Together with the
  fact that all steps in $T$ occur at disjoint positions, this implies
  that $\subtree u p$ is total and thus an element of
  $\afrag$. Consequently, we have a destructive reduction starting in
  $\subtree u p$. By embedding this reduction in $u$ at position $p$,
  we obtain a reduction $U\fcolon u \pato[\betared] u'$ that has $p$
  as a volatile position. Following Lemma~\ref{lem:volatileBot}, we
  only have to show that $p$ is outermost-volatile in $U$ in order to
  obtain that $u' = v$. Since all steps in $U$ take place at $p$ or
  below, $p$ can only fail to be outermost-volatile if it is
  strict. We show that $p$ is non-strict. If $p = \emptyseq$, this is
  trivial. Otherwise, $p = q \concat \seq i$. If $\ola = 111$, then
  $p$ is non-strict. Otherwise, $t$ must be a $\botred$-normal form
  according to the assumption. Moreover, we know that $\subtree t p =
  \bot$. Hence, $\subtree t q \neq \bot$, and, according to
  Lemma~\ref{lem:botAfrag}, $\subtree t q \in \afrag_\bot$ whenever
  $a_i = 0$. That means, $\subtree t q$ is a $\botred$-redex whenever
  $a_i = 0$, which contradicts the assumption that $t$ is a
  $\botred$-normal form. Hence, $a_i = 1$ and, thus, $p$ is non-strict.

  Let $T'$ be the reduction that is obtained from $T$ by replacing
  each step $u \to[\botred,p] v$ with a reduction $u \pato[\betared]
  v$ as constructed above. Clearly, we then have that $T'\fcolon s'
  \pato[\betared] t$.
\end{proof}
Note that the restriction to total lambda trees $s$ is crucial: if
$a_0 = 0$, then we have a single step reduction $\abstree x \bot
\to[\bohmred] \bot$ to a $\botred$-normal form, but there is no
$\prs$-converging reduction $\abstree x \bot \pato[\betared] \bot$ as
$\abstree x \bot$ is a $\betared$-normal form.\\[.7em]
\begin{thmcpy}{Theorem~\ref{thr:prsNormalising}}
  For each $s \in \aptree$, there is a normalising reduction $s
  \pato[\betaSred] t$.
\end{thmcpy}
\begin{proof}[Proof Proof of Theorem~\ref{thr:bohmNormalising}]
  For each lambda tree $u$ and non-strict position $p\in \dom{u}$, we
  have the following: If $\subtree u p$ is not active, then there is a
  finite reduction $u \fto*[\betared] v$, where $\subtree v p$ is
  stable. If, on the other hand, $\subtree u p$ is active, then it is,
  according to Lemma~\ref{lem:aactAfrag}, also fragile. Consequently,
  we find a reduction $S\fcolon u \pato[\betared] v$, in which $p$ is
  volatile. Hence, according to Lemma~\ref{lem:volatileBot}, we have
  that $p\nin \dom v$, i.e.\ subtree at $p$ has been annihilated.

  By performing the above reductions starting with $s$ at root
  position and proceeding at positions of increasing depth, we obtain
  a $\prs$-converging reduction $s \pato[\betared] t$ such that each
  subtree of $t$ is stable. That is, $t$ is a $\betared$-normal form.

  We also find a reduction $s \pato[\betasred] u$ to
  $\betasred$-normal form $u$ by extending the $\betared$-reduction
  $s \pato[\betared] t$ with a $\strictred$-reduction
  $t \pato[\strictred] u$ that consecutively contracts all
  $\strictred$-redexes:
  \[
    t = t_0 \pato[\strictred] t_1 \pato[\strictred] t_2 \pato[\strictred] \dots
  \]
  where each reduction $t_i \pato[\strictred] t_{i+1}$ is a complete
  development (cf.\ Section~\ref{sec:infin-strip-lemma}) of all
  $\strictred$-redexes in $t_i$. Since each contraction of a
  $\strictred$-redex at depth $d$ creates at most one new redex at
  depth $d-1$ and no other redexes, this process will terminate. In
  other words, there is some $n < \omega$ such that $t_n$ is a
  $\strictred$-normal form. Since contraction of $\strictred$-redexes
  creates no $\betared$-redexes and $t_0$ is a $\betared$-normal form,
  we know that $t_n$ is a $\betasred$-normal form.
\end{proof}

\begin{lemma}
  \label{lem:fintoRedex}
  If $s \pato[\betasred] t$ contracts a $\betared$-redex at position
  $p$, then there is a finite reduction $s \fto*[\betared] u$ to a
  term $u$ with a $\betared$-redex occurrence at $p$.
\end{lemma}
\begin{proof}
  Let $S\fcolon s \pato[\betasred] s'$ be the prefix of
  $s \pato[\betasred] t$ that converges to a lambda tree $s'$ that has
  a $\betared$-redex occurrence at position $p$. This means that
  $s'(p\concat \seq 1) = \lambda$. By Lemma~\ref{lem:finApprox}, there
  is a finite reduction $s \fto*[\betared] t$ with
  $t(p\concat \seq 1) = \lambda$. That is, $t$ has a $\betared$-redex
  occurrence at $p$.
\end{proof}

\begin{lemma}
  \label{lem:botRedFiniteRedex}
  Let $S\fcolon t \pato[\betasred] \bot$. Then there is a finite
  reduction $t \fto*[\betasred] u$ such that either $u = \bot$ or $u$
  has a $\betared$-redex occurrence at $\ola$-depth $0$.
\end{lemma}
\begin{proof}
  We proceed by induction on the length of $S$. In case $S$ is finite,
  there is nothing to show. Otherwise, $S$ is of the form
  $t \pato[\betasred] s \fto*[\betasred] \bot$.
  \begin{itemize}
  \item Let $s \fto*[\betasred] \bot$ contain a $\betared$-step at
    $\ola$-depth $0$. By Lemma~\ref{lem:fintoRedex}, there is a finite
    reduction $t \fto*[\betared] u$ where $u$ has a $\betared$-redex
    occurrence at $\ola$-depth $0$.
  \item Let $s \fto*[\betasred] \bot$ be empty, i.e.\
    $S\fcolon t \fto*[\betasred] \bot$ is empty. Then $S$ can be
    turned into an open $\betared$-reduction $t \pato[\betared] \bot$
    according to Lemma~\ref{lem:openOneBetaRed}. Which according to
    Proposition~\ref{prop:destructBot}, contains a $\betared$-step at
    $\ola$-depth $0$. By Lemma~\ref{lem:fintoRedex}, there is a finite
    reduction $t \fto*[\betared] u$ a $\betared$-redex at $\ola$-depth
    $0$.
  \item Let $s \fto*[\betasred] \bot$ be non-empty with no
    $\betared$-step at $\ola$-depth $0$. Then $s$ must contain a
    non-root occurrence of $\bot$ at $\ola$-depth $0$, i.e.\ there is
    some $p \in \domBot s$ with $p\neq \emptyset$ and $\adepth p =
    0$. That means $p = q \concat \seq i$ with $a_i = 0$. By
    Corollary~\ref{cor:pconvOpenPosBot}~\ref{item:pconvOpenPosBot1},
    there is a proper prefix $t \pato[\betasred] u$ of the reduction
    $t \pato[\betasred] s$ such that $p\in \domBot u$, too. Hence,
    $t \pato[\betasred] u\fto*[\strictred] \bot$. Since this reduction
    is strictly shorter than $S$, we may apply the induction
    hypothesis to obtain the desired finite reduction
    $t \fto*[\betasred] v$ such that either $v = \bot$ or $v$ has a
    $\betared$-redex occurrence at $\ola$-depth $0$.
  \end{itemize}
\end{proof}

\begin{thmcpy}{Lemma~\ref{lem:botCompress}}
  If $\ola \in \set{001,101,111}$ and
  $S\fcolon t \pato[\betasred] \bot$, then there is a reduction
  $T\fcolon t \pato[\betasred] \bot$ of length $\le\omega$. If $s$ is
  total, then $T$ is a $\betared$-reduction of length $\omega$.
\end{thmcpy}
\begin{proof}
  By Lemma~\ref{lem:botRedFiniteRedex}, we find a finite reduction
  $t \fto*[\betasred] t_1$ that contracts a redex at $\ola$-depth $0$
  or ends in $\bot$.  By Lemma~\ref{lem:stripLem'} there is also a
  reduction $S'\fcolon t_1 \pato[\betasred] \bot$. Thus we can repeat
  the argument for $S'$ (instead of $S$). By iterating this argument,
  we obtain a reduction
  \[
    T\fcolon t \fto*[\betasred] t_1 \fto*[\betasred] t_2 \fto*[\betasred] \dots
  \]
  that either stops at some $t_n = \bot$ or is of length $\omega$ and
  contracts infinitely many redexes at $\ola$-depth $0$ and thus
  $\prs$-converges to $\bot$ according to
  Proposition~\ref{prop:destructBot}. In either case $T\fcolon t
  \pato[\betared] \bot$.
  
  If $s$ is total then $T$ cannot be finite, as finite
  $\betasred$-reductions preserve totality. Hence, no step in $T$ can
  be an $\strictred$-step.
\end{proof}

\section{Finitary Approximation Lemma and Infinitary Strip Lemma}
\label{sec:infin-strip-lemma}

In order to prove the finitary approximation lemma and the infinitary
strip lemma for $\prs$-converging $\betasred$-reductions, we adopt the
familiar technique of descendants and complete developments.

Throughout this section we consider only $\betasred$-reductions over
$\aptree$. As we develop the theory, we have to make restrictions on
the strictness signatures $\ola$ we consider.

Our definition of complete developments for $\prs$-converging
$\betasred$-reductions is a straightforward adaptation of the concept
from the literature~\cite{kennaway95ic,kennaway03book,ketema11ic}:
\begin{definition}[descendants]
  \label{def:desc}
  Let $S\fcolon t_0 \pato[\betasred] t_\alpha$ of length $\alpha$, and
  $U \subseteq \dom{t_0}$. The \emph{descendants} of $U$ by $S$,
  denoted $\desc{U}{S}$, is a subset of $\dom{t_\alpha}$ inductively
  defined as follows:
  \begin{enumerate}[(a)]
  \item If $S = \emptyseq$, then $\desc{U}{S} = U$.
    \label{item:descA}
  \item If $S = \seq\phi$ with $\phi\fcolon s \to[p] t$, then
    $\desc{U}{S} = \bigcup_{u \in U} R_u$, where:
    \\\noindent
    If $\phi$ is a $\betared$-step:
      \[
        R_u =
        \begin{cases}
          \set u & \text{ if } p \not\le u\\%
          \emptyset &\text{ if } u \in \set{p, p\concat \seq 1}\\%
          \setcom{p \concat q \concat w}{
            \begin{aligned}&s (p\concat\seq{1,0}\concat q) \\& =
              p\concat\seq 1 \end{aligned} }
          & \text{ if } u = p \concat \seq 2 \concat w\\
          \set{p\concat w} &\text{ if } u = p \concat \seq{1,0}\concat
          w
          \text{ and } s(u) \neq p\concat \seq 1\\
          \emptyset &\text{ if } u = p \concat \seq{1,0}\concat w
          \text{ and } s(u) = p\concat \seq 1
        \end{cases}
      \]
      \\\noindent
      If $\phi$ is a $\strictred$-step:
      \[
        R_u =
        \begin{cases}
          \emptyset &\text{ if } p \le u \\
          \set{u}  &\text{ if } p \not\le u
        \end{cases}
      \]

  \item If $S = T \concat \seq \phi$, then $\desc{U}{S} =
    \desc{(\desc{U}{T})}{\seq{\phi}}$
    \label{item:descC}
  \item If $S$ is open, then \quad $\desc{U}{S} = \dom{t_\alpha}
    \cap \liminf_{\iota \limto \alpha} \desc{U}{\prefix{S}{\iota}}$
    \\
    That is, $u \in \desc{U}{S} \quad \text{ iff } \quad u \in
    \dom{t_\alpha} \text{ and } \exists \beta < \alpha \forall
    \beta \le \iota < \alpha\fcolon u \in \desc{U}{\prefix{S}{\iota}}$
    \label{item:descD}
  \end{enumerate}
  If, in particular, $U$ is a set of redex occurrences, then
  $\desc{U}{S}$ is also called the set of \emph{residuals} of $U$ by
  $S$. Moreover, by abuse of notation, we write $\desc{u}{S}$ instead
  of $\desc{\set{u}}{S}$.
\end{definition}

The following lemma provides a more convenient characterisation of
descendants in the case of open reductions.
\begin{lemma}
  \label{lem:descAltChar}
  Let $S = (\phi_\iota\fcolon t_\iota \to[p_\iota] t_{\iota+1})_{\iota<\alpha}$ be an open
  $\betasred$-reduction with $S \fcolon s \pato[\betasred] t$ and $U
  \subseteq \dom s$. Then we have the following:
  \[
  p \in \desc U S \iff \exists \beta < \alpha: p \in \desc{U}{\prefix
    S \beta} \text{ and } \forall \beta\le\iota<\alpha: \acut{p_\iota}
  \not \le p
  \]
\end{lemma}
\begin{proof}
  We first prove the ``$\Longrightarrow$'' direction. To this end, we
  assume some $p \in \desc U S$. Consequently, $p \in \dom t$ and
  there is some $\beta_1 < \alpha$ such that $p \in \desc{U}{\prefix S
    \iota}$ for all $\beta_1 \le \iota < \alpha$. According to
  Lemma~\ref{lem:pconvOpenPos}~\ref{item:pconvOpenPos1}, we thus also
  find some $\beta_2 < \alpha$ such that $\acut{p_\iota} \not\le p$
  for all $\beta_2\le \iota < \alpha$. Consequently, given $\beta =
  \max\set{\beta_1,\beta_2}$, we have that $p \in \desc{U}{\prefix S
    \beta}$ and that $\acut{p_\iota} \not\le p$ for all
  $\beta\le\iota<\alpha$.

  For the ``$\Longleftarrow$'' direction, we show by induction on
  $\gamma$ that $p \in \desc{U}{\prefix S \gamma}$ for all $\beta \le
  \gamma \le \alpha$.

  The case $\gamma = \beta$ is trivial. Let $\gamma = \gamma' + 1 >
  \beta$. That is, $\desc{U}{\prefix S \gamma} =
  \desc{\left(\desc{U}{\prefix S
        {\gamma'}}\right)}{\seq{\phi_{\gamma'}}}$. By the induction
  hypothesis, we know that $p \in \desc{U}{\prefix{S}{\gamma'}}$.
  Moreover, $\acut{p_{\gamma'}} \not\le p$ implies that $p_{\gamma'}
  \not\le p$. Consequently, $p \in \desc{U}{\prefix{S}{\gamma'}}$
  implies $p \in \desc{U}{\prefix{S}{\gamma}}$.

  Let $\gamma$ be a limit ordinal. By the induction hypothesis, we
  know that $p \in \desc{U}{\prefix S \iota}$ for all $\beta \le \iota
  < \gamma$. Hence, it remains to be shown that $p \in
  \dom{t_\gamma}$. Since $p \in \desc{U}{\prefix S \beta}$, we know
  that $p \in \dom{t_\beta}$. The latter, combined with the assumption
  that $\acut{p_\iota} \not \le p$ for all $\beta\le\iota<\gamma$,
  implies by Lemma~\ref{lem:pconvOpenPos}~\ref{item:pconvOpenPos2}
  that $p \in \dom{t_\gamma}$.
\end{proof}

\begin{lemma}[monotonicity]
  \label{lem:descMonotone}
  Let $S\fcolon s \pato[\betasred] t$ and $U,V \subseteq \dom s$. If $U
  \subseteq V$, then $\desc U S \subseteq \desc V S$.
\end{lemma}
\begin{proof}
  We prove this statement by induction on the length of $S$. If $S$ is
  empty, the statement is trivial. If $S = T \concat \seq \phi$, then
  \[
  \desc U S = \desc{\left(\desc U T\right)}{\seq \phi}
  \stackrel{\text{IH}}\subseteq \desc{\left(\desc V T\right)}{\seq
    \phi} = \desc V S
  \]

  Let $S$ be open, $\alpha =\len S$, and $p \in \desc U S$. Then $p
  \in \dom t$ and there is some $\beta < \alpha$ such that $p \in
  \desc{U}{\prefix S \iota}$ for all $\beta \le \iota <
  \alpha$. According to the induction hypothesis, we then have that $p
  \in \desc{V}{\prefix S \iota}$ for all $\beta \le \iota <
  \alpha$. Consequently, $p \in \desc V S$.
\end{proof}

\begin{proposition}
  \label{prop:descSingle}
  Let $S \fcolon s \pato[\betasred] t$ and $U \subseteq \dom s$. Then
  $\desc{U}{S} = \bigcup_{u\in U}\desc u S$.
\end{proposition}
\begin{proof}
  We prove this proposition by induction on the length of $S$.  The
  cases $S = \emptyseq$ and $S = \seq \phi$ are trivial.

  If $S = T \concat \seq\phi$, we can reason as follows:
  \begin{align*}
    \desc{U}{S} & =\desc{(\desc{U}{T})}{\seq\phi}
    \stackrel{IH}{=} \desc{\underbrace{(\bigcup_{u \in
          U} \overbrace{\desc{u}{T}}^{V_u})}_V}{\seq\phi}
    \stackrel{IH}= \bigcup_{u\in V} \desc{u}{\seq\phi}
    \\
    &= \bigcup_{u \in U} \bigcup_{v \in V_u}
    \desc{v}{\seq\phi}
    \stackrel{IH}= \bigcup_{u\in U} \desc{V_u}{\seq\phi}
    = \bigcup_{u\in U} \desc{(\desc{u}{T})}{\seq\phi}
    = \bigcup_{u \in U} \desc{u}{S}
  \end{align*}
  
  Let $S$ be open. The ``$\supseteq$'' follows from
  Lemma~\ref{lem:descMonotone}. For the converse direction, we assume
  $S = (t_\iota \to[p_\iota] t_{\iota+1})_{\iota<\alpha}$ and $p \in
  \desc U S$. By Lemma~\ref{lem:descAltChar}, there is some $\beta <
  \alpha$ such that $p \in \desc{U}{\prefix S \beta}$ and
  $\acut{p_\iota} \not\le p$ for all $\beta \le\iota<\alpha$. Hence,
  by induction hypothesis, $p \in \bigcup_{u\in U} \desc{u}{\prefix S
    \beta}$. That is, there is some $u^*\in U$ with $p \in
  \desc{u^*}{\prefix S \beta}$. By Lemma~\ref{lem:descAltChar}, we may
  thus conclude that $p \in \desc{u^*}{S}$ and thus $p \in
  \bigcup_{u\in U}\desc{u}{S}$.
\end{proof}

\begin{proposition}
  \label{prop:descUnique}
  Let $S \fcolon s \pato[\betasred] t$ and $U,V \subseteq \dom s$. If
  $U \cap V = \emptyset$, then $\desc U S \cap \desc V S = \emptyset$.
\end{proposition}
\begin{proof}
  We show the contraposition of the above statement. To this end,
  assume that there is some $w \in \desc U S \cap \desc V S$. By
  Proposition~\ref{prop:descSingle}, we thus find some $u \in U, v \in
  V$ with $w \in \desc u S \cap \desc v S$. We show by induction on
  the length of $S$ that $u = v$, which then implies that $U \cap V
  \neq \emptyset$.

  The case $S=\emptyseq$ is trivial and the case that $S = T \concat
  \seq \phi$ follows immediately from the induction hypothesis.

  Let $S$ be open with length $\alpha$. Since $w \in \desc u S \cap
  \desc v S$, we find $\beta_1,\beta_2$ such that $w \in
  \desc{u}{\prefix S \iota}$ for all $\beta_1 \le \iota < \alpha$ and
  $w \in \desc{v}{\prefix S \iota}$ for all $\beta_2 \le \iota <
  \alpha$. Given $\beta = \max\set{\beta_1,\beta_2}$, we thus have
  that $w \in \desc{u}{\prefix S \beta} \cap \desc{v}{\prefix S
    \beta}$. Hence, by induction hypothesis, we have that $u = v$.
\end{proof}

By combining Proposition~\ref{prop:descSingle} and
Proposition~\ref{prop:descUnique}, we know that for each $p \in \desc
U S$ there is a unique $q \in U$ such that $p\in \desc q S$. This
unique position $q$ is also called an \emph{ancestor}. Moreover, we
can show that every position in a lambda tree has an ancestor:
\begin{lemma}
  \label{lem:ancestor}
  Let $S \fcolon s \pato[\betasred] t$ and $p \in \dom t$. Then there
  is a unique $q \in \dom s$ with $p \in \desc q S$.
\end{lemma}
\begin{proof}
  By Proposition~\ref{prop:descSingle} and
  Proposition~\ref{prop:descUnique}, it suffices to show that $\dom t
  \subseteq \desc{\dom s}{S}$. If we have that, then we find,
  according to Proposition~\ref{prop:descSingle}, for each $p \in \dom
  t$ some $q \in \dom s$ with $p \in \desc q S$. By
  Proposition~\ref{prop:descUnique}, this $q$ is unique.

  Let $S = (t_\iota \to[p_\iota] t_{\iota+1})_{\iota<\alpha}$. We
  prove that $\dom t \subseteq \desc{\dom s}{S}$, by induction on
  $\alpha$.

  The case $\alpha = 0$ is trivial. If $\alpha$ is a successor
  ordinal, the inclusion follows immediately from the induction
  hypothesis.

  Let $\alpha$ be a limit ordinal and let $p \in \dom t$. By
  Lemma~\ref{lem:pconvOpenPos}~\ref{item:pconvOpenPos1}, this implies
  that there is some $\beta < \alpha$ such that $p \in \dom{t_\beta}$
  and $\acut{p_\iota} \not\le p$ for all $\beta \le \iota <
  \alpha$. Hence, by the induction hypothesis, $p \in \desc{\dom
    s}{\prefix S \beta}$. Since $\acut{p_\iota} \not\le p$ for all
  $\beta \le \iota < \alpha$, we know by Lemma~\ref{lem:descAltChar}
  that $p \in \desc{\dom s}{\prefix S \beta}$ implies that $p \in
  \desc{\dom s}{S}$.
\end{proof}

\begin{lemma}
  \label{lem:descPreserveLabel}
  Let $S \fcolon s \pato[\betasred] t$ and $p \in \dom s$. Then we have
  that $s(p) = t(q)$ for all $q \in \desc p S$.
\end{lemma}
\begin{proof}
  Let $S = (t_\iota \to[p_\iota] t_{\iota+1})_{\iota<\alpha}$. We
  prove this lemma by induction on $\alpha$.

  The case $\alpha = 0$ is trivial. If $\alpha$ is a successor
  ordinal, the inclusion follows immediately from the induction
  hypothesis.

  Let $\alpha$ be a limit ordinal and let $q \in \desc p S$. By
  Lemma~\ref{lem:descAltChar}, this implies that there is some $\beta
  < \alpha$ such that $q \in \desc{p}{\prefix S \beta}$ and
  $\acut{p_\iota} \not\le q$ for all $\beta \le \iota < \alpha$. By
  the induction hypothesis, we thus have that $s(p) =
  t_\beta(q)$. Lemma~\ref{lem:pconvOpenPos}~\ref{item:pconvOpenPos2}
  then yields that $t_\beta(q) = t(q)$, which means that we have the
  desired equality $s(p) = t(q)$.
\end{proof}

\begin{lemma}[Finitary Approximation Lemma]
  \label{lem:finApprox}%
  Let $s \pato[\betasred] t$, and $P$ a finite subset of $\dom t$. Then
  there is a reduction $s \fto*[\betared] t'$ with $t(p) = t'(p)$ for
  all $p \in P$.  
\end{lemma}
\begin{proof}
  We prove this by induction on the length of $S$. The case $S =
  \emptyseq$ is trivial.

  Let $S = T \concat \seq \phi$, where $\phi\fcolon s' \to[q]
  t$ is a $\betared$-step. Define
  \[
  P' = \setcom{p' \in \dom{s'}}{\exists p \in P\fcolon p \in
    \desc{p'}{\phi}} \cup \set{ q \concat\seq 1}.
  \]
  By Lemma~\ref{lem:ancestor}, $P'$ is finite, too. Thus, by induction
  hypothesis, there is a finite reduction $S'\fcolon s \fto*[\betared] s''$ such
  that $s'(p) = s''(p)$ for all $p \in P'$. In particular, that means
  that $q$ is still $\betared$-redex occurrence in $s''$. Thus there
  is a $\betared$-reduction step $\phi'\fcolon s'' \to[q] t'$. Let $p
  \in P$. According to Lemma~\ref{lem:ancestor}, there is a unique $p'
  \in \dom{s'}$ with $p \in \desc{p'}{\phi}$. By the construction of
  $P'$, we know that $p' \in P'$. Hence,
  \[
  t'(p) \stackrel{\text{Lemma~\ref{lem:descPreserveLabel}}}= s''(p') \stackrel{\text{IH}}= s'(p')
  \stackrel{\text{Lemma~\ref{lem:descPreserveLabel}}}= t(p)
  \]

  Let $S = T \concat \seq \phi$, where $\phi\fcolon s' \to[q] t$ is a
  $\strictred$-step. Then $s'(p) = t'(p)$ for all $p \in \dom
  t$. Moreover, we also have that $\dom t \subseteq \dom{s'}$, which
  implies that $P \subseteq \dom{s'}$. Hence, we may apply the
  induction hypothesis to obtain a finite reduction
  $S' \fcolon s \fto*[\betared] s''$ such that $s'(p) = s''(p)$ for all
  $p \in P$. Consequently, $s''(p) = t(p)$ for all $p \in P$.
  
  Let $S$ be open. By
  Lemma~\ref{lem:pconvOpenPos}~\ref{item:pconvOpenPos1}, there is a
  prefix $T < S$ with $T\fcolon s \pato t'$ and $t(p) = t'(p)$ for all
  $p\in P$. By applying the induction hypothesis, we then obtain a
  finite reduction $s \fto* t''$ with $t'(p) = t''(p)$ for all $p\in
  P$. Consequently, we have that $t(p) = t''(p)$ for all $p\in P$.
\end{proof}

\begin{proposition}
  \label{prop:descSeq}
  Let $S\fcolon t_0 \pato[\betasred] t_1$, $T \fcolon t_1
  \pato[\betasred] t_2$, and $U \subseteq \dom{t_0}$. Then
  $\desc{U}{S\concat T} = \desc{\left(\desc U S\right)}{T}$.
\end{proposition}
\begin{proof}
  We prove this by induction on the length of $T$. The case $T =
  \empty$ is trivial.
  
  If $T = T'\concat \seq\phi$, then we can reason as follows:
  \[
  \desc{U}{S\concat T'\concat \seq\phi} = \desc{\left(\desc{U}{S\concat
        T'}\right)}{\seq\phi} \stackrel{\text{IH}}=
  \desc{\left(\desc{\left(\desc{U}{S}\right)}{T'}\right)}{\seq\phi} =
  \desc{\left(\desc{U}{S}\right)}{T'\concat\seq\phi}
  \]

  Let $T$ be open. That means, also $S\concat T$ is open. Hence, we
  can reason as follows:
  \begin{align*}
    p \in \desc{U}{S \concat T} &\iff p\in\dom{t_2}, \exists \beta <
    \len{S\concat T}\forall \beta \le \iota \le \len{S\concat T}: p
    \in
    \desc{U}{\prefix{\left(S\concat T\right)}\iota}\\
    &\iff p\in\dom{t_2}, \exists \beta < \len{S\concat T}\forall \beta
    \le \iota \le
    \len{S\concat T}: p \in \desc{U}{S\concat \left(\prefix{T}\iota\right)}\\
    &\stackrel{\text{IH}}\iff p\in\dom{t_2}, \exists \beta <
    \len{S\concat T}\forall \beta \le \iota \le \len{S\concat T}: p
    \in \desc{\left(\desc{U}{S}\right)}{\prefix{T}\iota}\\
    &\iff p \in \desc{\left(\desc U S\right)}{T}
  \end{align*}
\end{proof}

For the next proposition we have to exclude certain strictness
signatures, in particular strictness signatures of the form $01*$. The
problem is that for strictness signatures of this form a descendant of
a redex occurrence may not be a redex occurrence as the following
example demonstrates:
\begin{example}
  Let $\Omega = (\abstree x \vartree x\,\vartree x) (\abstree x
  \vartree x\,\vartree x)$ and $t = (\abstree x \Omega)\, \vartree
  y$. We consider the $\betared$-reduction $S \fcolon t \to[\seq{1,0}]
  t \to[\seq{1,0}] \dots$ that repeatedly contracts the redex $\Omega$
  at $\seq{1,0}$. This reduction $\prs$-converges to $\bot\,\vartree
  y$ if $a_0=0$ and $a_1 = 1$. The descendent of the redex occurrence
  $\emptyseq$ in $t$ by $S$ is not a redex occurrence in
  $\bot\,\vartree y$, which is a $\betared$-normal form.

  If, in addition, $a_2=1$, this phenomenon may also occur for
  developments. Let $I^\omega$ be a lambda tree with $I^\omega =
  (\abstree x \vartree x)I^\omega$, and let $t = (\abstree x
  I^\omega)\, \vartree y$. Both $I^\omega$ and $t$ are lambda trees in
  $\aptree[011]$. Let $U$ be the set of all occurrences of $I^\omega$
  in $t$. A complete development of $U$ in $t$, e.g.\ $S \fcolon t
  \to[\seq{1,0}] t \to[\seq{1,0}] \dots$, $\prs$-converges to
  $\bot\,\vartree y$. Again, $\emptyseq$ is a redex occurrence in $t$,
  but its descendant by $S$ is not a redex occurrence in
  $\bot\,\vartree y$.
\end{example}

\begin{proposition}
  \label{prop:residuals}
  Let $S \fcolon s \pato[\betasred] t$, and $a_0 = 1$ or $a_1 = 0$. If
  $U$ is a set of redex occurrences in $s$, then $\desc{U}{S}$ is a
  set of redex occurrences in $t$.
\end{proposition}
\begin{proof}
  Let $S = (t_\iota \to[p_\iota] t_{\iota+1})_{\iota<\alpha}$. We
  proceed by induction on $\alpha$. The case $\alpha = 0$ is trivial,
  and for $\alpha$ a successor ordinal the statement follows
  immediately from the induction hypothesis.

  Let $\alpha$ be a limit ordinal. To prove the statement, we assume
  some $p \in \desc U S$ and show that $\subtree t p$ is a
  $\betasred$-redex. According to Lemma~\ref{lem:descAltChar}, $p \in
  \desc U S$ implies that there is some $\beta < \alpha$ such that
  \begin{equation}
    p \in \desc U {\prefix S \beta}\text{ and } \forall \beta \le\iota
    < \alpha:\acut{p_\iota}\not\le p%
    \tag{1}\label{eq:residuals1}
  \end{equation}
  By the induction hypothesis, we know that $\subtree{t_\beta}{p}$ is
  a $\betasred$-redex.

  \begin{itemize}
  \item We first consider the case that $\subtree{t_\beta}{p}$ is a $\betared$-redex,
    i.e.\ $t_\beta(p\concat\seq 1) = \lambda$.

    We proceed by showing the following two statements for all
    $\beta \le \gamma \le \alpha$, where $t_\alpha = t$:
    \begin{align}
      t_\gamma(p\concat\seq 1) &= \lambda%
                                 \tag{2}\label{eq:residuals2}\\%
      c_\iota(p\concat\seq 1) &= \lambda\qquad \text{for all } \beta
                                \le \iota < \gamma%
                                \tag{3}\label{eq:residuals3}
    \end{align}
    For the case $\gamma = \alpha$, we then obtain that
    $t(p\concat\seq 1) = \lambda$, i.e. $\subtree t p$ is a
    $\betared$-redex.

    For the case $\gamma = \beta$, we have already shown
    \eqref{eq:residuals2} above, and \eqref{eq:residuals3} is
    vacuously true.

    Let $\gamma = \gamma'+1 > \beta$. According to the induction
    hypothesis, \eqref{eq:residuals3} holds for $\gamma'$, which means
    it remains to be shown that
    $t_\gamma(p\concat\seq 1)= c_{\gamma'}(p\concat\seq 1) =
    \lambda$. We consider $c_{\gamma'}$ first. If
    $\acut{p_{\gamma'}} \not\le p\concat\seq 1$, then
    \[
      c_{\gamma'}(p\concat\seq 1) = t_{\gamma'}(p\concat\seq 1)
      \stackrel{\text{IH}}= \lambda
    \]
    Otherwise, by \eqref{eq:residuals1}, $\acut{p_{\gamma'}}$ must be
    equal to $p\concat\seq 1$. This can only happen if $a_1 =
    1$. Hence, according to the assumption about $\ola$, we know that
    $a_0 = 1$. Moreover, $\subtree{t_{\gamma'}}{p\concat\seq 1}$
    cannot be a $\betared$-redex because, by the induction hypothesis,
    $t_{\gamma'}(p\concat\seq 1) = \lambda$. Consequently,
    $p_{\gamma'}\ge p \concat\seq{1,0}$, which is not possible since
    $a_0 = 1$ and $\acut{p_{\gamma'}} = p\concat\seq 1$.

    Next we consider $t_\gamma$: since
    $c_{\gamma'} \talebot t_\gamma$, we have
    $t_\gamma(p\concat\seq 1) = c_{\gamma'}(p\concat\seq 1) =
    \lambda$.

    Finally, let $\gamma$ be a limit ordinal. Then
    \eqref{eq:residuals3} follows immediately from the induction
    hypothesis. We will show that $p\concat\seq 1 \in \dom{s}$ for
    $s = \Glb_{\beta\le\iota<\gamma} c_\iota$. Since
    $s \talebot c_\beta, t_\gamma$, we then have:
    \[
      t_\gamma(p\concat\seq 1) \stackrel{s \talebot t_\gamma}=
      s(p\concat\seq 1) \stackrel{s \talebot c_\beta}=
      c_\beta(p\concat\seq 1) \stackrel{\text{IH}}= \lambda
    \]

    We know that $\acut{p_\iota} \not\le p\concat\seq 1 \concat w$ for
    all $\beta\le\iota<\gamma$ and $w$ with $\adepth w =
    0$. Otherwise, we would have $\acut{p_\iota} \le p\concat\seq 1$
    for some $\beta\le\iota<\gamma$, which would contradict the
    induction hypothesis for \eqref{eq:residuals2} if
    $\acut{p_\iota} = p\concat\seq 1$, and \eqref{eq:residuals1} if
    $\acut{p_\iota} \le p$. Since we know from induction hypothesis
    for \eqref{eq:residuals3}, that, for each
    $\beta\le \iota < \gamma$, we have that
    $p\concat \seq 1 \in \dom{c_\iota}$, we can apply
    Lemma~\ref{lem:reductionContextCut}, conclude that also
    $p\concat \seq 1\concat w \in \dom{c_\iota}$ for all $w$ with
    $\adepth w = 0$. Consequently, according to
    Proposition~\ref{prop:alebotGlb}, we have that
    $p \concat \seq 1 \in \dom s$.
  \item Finally we consider the case that $\subtree{t_\beta}{p}$ is an
    $\strictred$-redex, i.e.\
    $p \concat \seq{i} \in \domBot {t_\beta}$ for some $i$ with
    $a_i =0$.

    We can then show by a simple induction proof that the following
    holds for all $\beta \le \gamma \le \alpha$, where $t_\alpha = t$:
    \begin{align}
      p\concat\seq i \in \domBot{t_\gamma}%
                                 \tag{4}\label{eq:residuals4}
    \end{align}
    For the case $\gamma = \alpha$, we then obtain that
    $p\concat\seq i \in \domBot{t}$, i.e. $\subtree t p$ is a
    $\strictred$-redex.

    The case $\gamma = \beta$ is trivial.

    If $\gamma = \gamma' + 1$, then
    $p\concat\seq i \in \domBot{t_\gamma}$ follows from the induction
    hypothesis ($p\concat\seq i \in \domBot{t_{\gamma'}}$) and the fact
    that by \eqref{eq:residuals1}, $p_\gamma \not\le p$.

    Let $\gamma$ be a limit ordinal. Then
    $p\concat\seq i \in \domBot{t_\gamma}$ follows from
    Corollary~\ref{cor:pconvOpenPosBot}~\ref{item:pconvOpenPosBot2} using \eqref{eq:residuals1}.
  \end{itemize}
\end{proof}

For the remainder of this section we (tacitly) restrict ourselves to
strictness signatures $\ola$ with $a_0 = 1$ or $a_1 = 0$. This
restriction is necessary in order for the following definition of
developments to make sense, since it depends on
Proposition~\ref{prop:residuals}.

For technical reasons, we have to generalise the notion of
developments. A development of a set of redex occurrences $U$ in $s$
is typically only allowed to contract redexes occurrences that are
descendants of redex occurrences in $U$. In addition to that we also
allow developments to contract \emph{any} $\strictred$-redex.

\begin{definition}[developments]
  Let $s \in \aptree$ and $U$ a set of redex occurrences in $s$.
  \begin{enumerate}[(i)]
  \item A reduction $S \fcolon s \pato[\betasred] t$ is a
    \emph{development of $U$} if each $\iota$-th step
    $\phi_\iota\fcolon t_\iota \to[p_\iota] t_{\iota+1}$ of $S$
    contracts a redex at a position
    $p_\iota \in \desc{U}{\prefix S \iota}$ or an $\strictred$-redex.
  \item A development $S \fcolon s \pato[\betasred] t$ of $U$ is
    called \emph{almost complete} if $\desc S U = \emptyset$. If, in
    addition, $t$ is an $\strictred$-normal form, then $S$ is called
    \emph{complete}, denoted $S \fcolon s \pato[U] t$.a
  \end{enumerate}
\end{definition}

\begin{definition}
  \label{def:doms}
  Given $t \in \aptree$, the set $\doms{t}$ is the smallest set
  satisfying the following:
  \begin{enumerate}[(a)]
  \item $\domBot t \subseteq \doms t$;
    \label{item:doms1}
  \item If $p\concat \seq i \in \doms t$ and $a_i = 0$, then
    $p \in \doms t$; and
    \label{item:doms2}
  \item If $p \in \doms t$ and $p \concat \seq i \in \doms t$, then
    $p \concat \seq i \in \doms t$.
    \label{item:doms3}
  \end{enumerate}
\end{definition}

\begin{proposition}
  \label{prop:strictRedNormalConfl}
  $\strictred$ is infinitarily normalising and confluent. The unique
  $\strictred$-normal form $\strictnf t$ of $t$ can be characterised
  as follows:
  \[
    \dom{\strictnf t} = \dom t \setminus \doms t\qquad \strictnf t(p)
    = t(p) \text{ for all } p \in \dom{\strictnf t}
  \]
\end{proposition}
\begin{proof}
  Let $t \in \aptree$ and define $\strictnf t$ as the restriction of
  $t$ to the domain $\dom t \setminus \doms t$. By \ref{item:doms3} of
  Definition~\ref{def:doms} this definition yields a well-defined
  lambda tree. Moreover, $\strictnf t$ is clearly an
  $\strictred$-normal form.

  To prove the proposition, we assume a reduction
  $S \fcolon t \pato[\strictred] u$ and show that then
  $u \pato[\strictred] \strictnf t$. It is easy to see that any
  reduction step in $S$ contracts a $\strictred$-redex at a position
  in $\doms t$, and thus $\doms u \subseteq \doms t$. Then a reduction
  $u \pato[\strictred] \strictnf t$ can be obtained by contracting all
  $\strictred$-redexes by an innermost parallel reduction strategy.
\end{proof}

According to Proposition~\ref{prop:strictRedNormalConfl}, each lambda
tree $t \in \aptree$ has a unique $\strictred$-normal form. We write
$\strictnf t$ to denote this unique normal form. Moreover, this
$\strictred$-normal form 

\begin{proposition}
  \label{prop:compDev}
  Every set of redex occurrences has a complete development.
\end{proposition}
\begin{proof}
  Below, we construct an almost complete development
  $t_0 \pato[\betasred] s$. This almost complete development can be
  extended to a complete development by a reduction
  $s \pato[\strictred] \strictnf s$ to the $\strictred$-normal form
  $\strictnf s$, which exists according to
  Proposition~\ref{prop:strictRedNormalConfl}.

  Let $t_0 \in \aptree$, $U_0$ a set of redex occurrences in $t_0$ and
  $V_0$ the set of outermost redex occurrences in $U_0$. Furthermore,
  let $S_0\fcolon t_0 \pato[V_0] t_1$ be some complete development of
  $V_0$ in $t_0$. $S_0$ can be constructed by contracting the redex
  occurrences in $V_0$ in a left-to-right order. This step can be
  continued for each $i < \omega$ by taking $U_{i+1} =
  \desc{U_i}{S_i}$, where $S_{i}\fcolon t_{i} \pato[V_{i}] t_{i+1}$ is
  some complete development of $V_{i}$ in $t_{i}$ with $V_{i}$ the set
  of outermost redex occurrences in $U_{i}$.

  Note that then, by iterating Proposition~\ref{prop:descSeq}, we have
  that
  \begin{gather}
    \desc{U}{S_0\concat \dots \concat S_{n-1}} = U_n \quad
    \text{ for all } n < \omega \tag{1}
    \label{eq:compDev1}
  \end{gather}
  If there is some $n < \omega$ for which $U_n = \emptyset$, then $S_0
  \concat \dots \concat S_{n-1}$ is a complete development of $U$
  according to \eqref{eq:compDev1}.

  If this is not the case, consider the reduction $S = \Concat_{i <
    \omega} S_i$, i.e.\ the concatenation of all '$S_i$'s. We claim
  that $S$ is a complete development of $U$. Suppose that this is not
  the case, i.e.\ $\desc{U}{S} \neq \emptyset$. Hence, there is some
  $u \in \desc{U}{S}$. Since all '$U_i$'s are non-empty, so are the
  '$V_i$'s. Consequently, all '$S_i$'s are non-empty reductions which
  implies that $S$ is an open reduction. Therefore, we can apply
  Lemma~\ref{lem:descAltChar} to infer from $u \in \desc{U}{S}$ that
  there is some $\alpha < \len S$ such that $u \in
  \desc{U}{\prefix{S}{\alpha}}$ and all reduction steps beyond
  $\alpha$ do not take place at $u$ or above. This is not possible due
  to the parallel-outermost reduction strategy that $S$ follows.
\end{proof}

Next we want to show that the final lambda tree of complete
developments is uniquely determined by the start lambda tree and the
set of redexes. To this end we restrict ourselves to strictness
signatures $111$, $101$ and $001$, since this uniqueness of complete
developments fails for all other strictness signatures (except for the
trivial $000$).
\begin{example}
  At first let $a_2 = 0$, and let $a_1 = 1$ or $a_0 = 1$. Hence, the
  lambda tree $s$ with $s = (\abstree x s)\,\vartree x$ is in
  $\aptree$. Given $t = (\abstree x \vartree y)\, s$, we find two
  complete developments of the set of all redex occurrences in $t$: $t
  \pato[\betared] \bot$ (by contracting $s$ to itself repeatedly) and
  $t \to[\betared] \vartree y$.

  If we had chosen a more strict notion of complete developments, that
  does not contract arbitrary $\strictred$-redexes, then the
  uniqueness of complete developments would also fail for $101$ and
  $001$:  Let $a_2 = 1$. Then the lambda tree $I^\omega$ with
  $I^\omega = (\abstree x \vartree x)\,I^\omega$ is in $\aptree$. At
  first consider $a_1 = 0$ and
  $t = (\abstree x \vartree x \,\vartree y)\,I^\omega$. Then we have
  two complete developments of the set of all redex occurrences in
  $t$:
  $t \pato[\betared] (\abstree x \vartree x \,\vartree y) \,
  \bot\to[\betared] \bot\,\vartree y$ and
  $t \to[\betared] I^\omega\, \vartree y \pato[\betared]
  \bot$. Similarly, given $a_0 = 0$, and
  $t = (\abstree x \abstree y \vartree x)\,I^\omega$, we have
  $t \pato[\betared] (\abstree x \abstree y \vartree x) \,
  \bot\to[\betared] \abstree y \bot$ and
  $t \to[\betared] \abstree y I^\omega \pato[\betared] \bot$.
\end{example}

\begin{definition}[paths]
  \label{def:path}
  Given a lambda tree $t \in \aptree$ and $U$ a set of redex
  occurrences in $t$, a \emph{$U$-path in $t$} (or simply \emph{path})
  is a finite sequence of length over the set $\dom t\cup\domBot t
  \uplus\set{0,1,2,\lambda,\calV}$ of the form
  $\seq{n_0,e_0,n_1,e_1,\dots,e_{l-1},n_l}$, subject to a number of
  restrictions.  We write paths using the notation
  $n_0\stackrel{e_0}\to n_1 \stackrel{e_1}\to \dots
  \stackrel{e_{l-1}}\to n_l$, and call $n_i$
  nodes and $e_i$ edges. Nodes range over $\dom t \cup \domBot t$ and
  edges over $\set{0,1,2,\lambda,\calV}$. If a path contains
  $n_i\stackrel{e_i}\to n_{i+1}$, we say that $n_i$ has an outgoing
  $e_i$-edge to $n_{i+1}$.

  The set of well-formed $U$-paths in $t$, denoted $\prepaths{t}{U}$,
  is defined as follows. Each path starts with the node $\emptyseq$
  and must end in a node. For each node $n$ with an outgoing $e$-edge
  to $n'$, we require that $n \in \pos t$ and that one of the
  following conditions holds:
  \begin{enumerate}[(a)]
  \item If $n \nin U$, then $n' = n \concat \seq i$ and $e = i$.
    \label{item:path1}
  \item If $n \in U$, then $n' = n \concat \seq{1,0}$ and $e =
    \lambda$.
    \label{item:path2}
  \item If $t(n) = p\concat\seq 1 \in \dom t$ and $p \in U$, then $n'
    = p \concat \seq{2}$ and $e = \calV$.
    \label{item:path3}
  \end{enumerate}
  Note that for (a), we implicitly require that $n \concat \seq i \in
  \dom t\cup\domBot t$.
\end{definition}
Note that the path consisting only of a single node $\emptyseq$ is a
path in any lambda tree.

\begin{definition}[diverging paths]
  \label{def:pathDiv}
  Let $t \in \aptree$ and $U$ a set of redex occurrences in $t$.  The
  set of \emph{diverging} $U$-paths in $t$, denoted $\dpaths{t}{U}$,
  is the subset of $\prepaths t U$ inductively defined as follows:
    \begin{enumerate}[(a)]
    \item Let $n_k \in \dom t\cup\domBot t$ and
      $e_k \in \set{\lambda,\calV}\cup \setcom{i}{a_i = 0}$ for all
      $k \ge 0$. If
      $P \stackrel{e_0}\to n_0 \stackrel{e_1}\to n_1 \stackrel{e_2}\to
      \dots \stackrel{e_m}\to n_m$ is a path in $\prepaths t U$ for
      each $m \ge 0$, then $P \in \dpaths{t}{U}$.
      \label{item:pathDiv1}
    \item If $P \in \prepaths t U$ ends in a node $n \in \domBot t$, then
      $P \in \dpaths{t}{U}$.
      \label{item:pathDiv2}      
    \item If $P \stackrel{i}{\to} n \in \dpaths{t}{U}$ with $a_i = 0$,
      then $P \in \dpaths{t}{U}$.
      \label{item:pathDiv3}
    \item If $P \in \dpaths{t}{U}$ and
      $P \stackrel{i}{\to} n \in \prepaths t U$, then
      $P \stackrel{i}{\to} n \in \dpaths{t}{U}$.
      \label{item:pathDiv4}
    \end{enumerate}
\end{definition}

\begin{definition}[terminated paths]
  Let $t \in \aptree$, $U$ a set of redex occurrences in $t$, and $P$
  a $U$-path in $t$.
  \begin{enumerate}[(i)]
  \item The \emph{position of $P$}, denoted $\pathpos P$, is the
    subsequence of $P$ containing only (and all) $i$-edges (with $i
    \in \set{0,1,2}$), i.e.\
    \[
      \pathpos{n} = \emptyseq \qquad
      \pathpos{P\stackrel{e}\to n} =
      \begin{cases}
        \pathpos P \concat \seq e&\text{ if } e \in \set{0,1,2}\\
        \pathpos P &\text{ if } e \in \set{\lambda,\calV}
      \end{cases}
      \]
 
  \item $P$ is called \emph{terminated} if it is not diverging, does
    not end in a node $n \in U$, and cannot be extended with a
    $\calV$-edge, i.e.\ there is no $U$-path in $t$ of the form
    $P\stackrel\calV\to n'$.  The set of all terminated $U$-paths in
    $t$ is denoted $\paths{t}{U}$.
  \item If $P$ is terminated, we define the labelling of $P$, denoted
    $\pathlab{P}$, as follows:
    \[
    \pathlab{P} =
    \begin{cases}
      t(n)&\text{ if $P$ terminates in a node $n$ with $t(n) \in
        \lamsig\setminus \allpos$}\\
      \pathpos{Q}&\text{ if $P$ terminates in a node $n$ with $t(n)
        \in \allpos$ and}\\
      &\text{\quad$Q$ is the longest prefix of $P$ that ends in $t(n)$ }
    \end{cases}
    \]
  \end{enumerate}
\end{definition}

\begin{lemma}
  \label{lem:pathPreserve}
  Let $S \fcolon s \pato t$ be a development of a set $U$ of redex
  occurrences in $s$. Then there is a surjective mapping
  $\theta_S\fcolon \paths s U \funto \paths{t}{\desc U S}$ that
  preserves $\pathpos\cdot$ and $\pathlab\cdot$, i.e.\
  \[
  \pathpos{\theta_S(P)} = \pathpos P \quad \text{and} \quad
  \pathlab{\theta_S(P)} = \pathlab P\qquad\text{for all } P \in \paths s U.
  \]
\end{lemma}
\begin{proof}
  Let $S = (t_\iota \to[p_\iota] t_{\iota+1})_{\iota<\alpha}$. We
  proceed by induction on $\alpha$. The case $\alpha = 0$ is
  trivial. If $\alpha$ is a successor ordinal, the statement follows
  from the induction hypothesis by careful case analysis (similar to
  \cite{kennaway03book}).

  Let $\alpha$ be a limit ordinal. Furthermore, let
  $P \in \paths{t_0}{U}$ and let
  $P_\iota = \theta_{\prefix S \iota}(P)$ for all $\iota <
  \alpha$. The latter is well-defined by the induction
  hypothesis. Since each $P_\iota$ is terminated, no node in $P_\iota$
  is a volatile position in $S$. Hence, there is some $\beta < \alpha$
  such that $\acut{p_\iota}$ is not a node in $P_\iota$ for all
  $\beta \le \iota < \alpha$. Consequently, $P_\iota = P_\beta$ for
  all $\beta \le \iota < \alpha$. From the characterisation of lubs
  and glbs from Theorem~\ref{thr:alebotCpo} and
  Proposition~\ref{prop:alebotGlb} we can then derive that $P_\beta$
  is also a $\desc U S$-path in $t_\alpha$. Additionally, $P_\beta$
  must also be terminated in $t_\alpha$, and we thus have that
  $P_\beta \in \paths{t_\alpha}{\desc U S}$. Define
  $\theta_S(P) = P_\beta$. Preservation of $\pathpos\cdot$ and
  $\pathlab\cdot$ follows from the induction hypothesis.

  To show that the thus defined function $\theta_S$ is indeed
  surjective, we assume some $P \in \paths{t_\alpha}{\desc U S}$ and
  show that there is some $Q \in \paths{t_0}{U}$ with $\theta_S(Q) = P$.

  Let $V$ be the set of nodes in $P$ (which are positions in
  $t_\alpha$). Since $V$ is finite, we may apply
  Lemma~\ref{lem:pconvOpenPos}~\ref{item:pconvOpenPos1}, to obtain
  some $\beta < \alpha$ such that $t_\iota(p) = t_\alpha(p)$ and
  $\acut{p_\iota} \not\le p$ for all $\beta \le \iota <\alpha$ and $p \in
  V$. Consequently, $P$ is a terminated $\desc{U}{\prefix S
    \beta}$-path in $t_\beta$, i.e.\ $P \in
  \paths{t_\beta}{\desc{U}{\prefix S \beta}}$. Since, by induction
  hypothesis $\theta_{\prefix S \beta}$ is surjective, there is some
  $Q \in \paths{t_0}{U}$ with $\theta_{\prefix S \beta}(Q) =
  P$. Hence, according to the definition of $\theta_S$, we have that
  $\theta_S(Q) = P$.
\end{proof}

We can use the above lemma to directly define the uniquely determined
final lambda tree of an arbitrary complete development of a given set
of redex occurrences:
\begin{definition}
  Let $t \in \aptree$ and $U$ a set of redex occurrences in $t$. Then
  define the set of path labellings of $t$ w.r.t.\ $U$, denoted
  $\pathLabs{t}{U}$, as follows:
  \[
  \pathLabs{t}{U} = \setcom{(\pathpos{P},\pathlab{P})}{P \in \paths{t}{U}}
  \]
\end{definition}


\begin{lemma}
  \label{lem:nodePosSame}
  Let $P$ be a $U$-path in $t$ with $U = \emptyset$. Then

  \begin{enumerate}[(i)]
  \item $P$ ends in the node $\pathpos{P}$, and
    \label{item:nodePosSame1}
  \item $P \in \paths{t}{\emptyset}$ iff $P \nin
    \dpaths{t}{\emptyset}$.
    \label{item:nodePosSame2}
  \end{enumerate}

\end{lemma}
\begin{proof}
  \begin{enumerate}[(i)]
  \item We proceed by induction on $P$. If $P$ consists of a single
    node, which thus has to be $\emptyseq$, then
    $\pathpos P = \emptyseq$. Otherwise, $P = Q \stackrel i \to n$ and
    by induction hypothesis we have that $Q$ ends in $\pathpos
    {Q}$. Since the set $U$ of redex occurrences is empty, only
    \ref{item:path1} of Definition~\ref{def:path} is
    applicable. Hence, we then have that
    $n = \pathpos{Q} \concat \seq i$, which is also the position of
    $P$.
  \item Since $U = \emptyset$, $P$ cannot end in a node in
    $U$. Moreover, $P$ cannot be extended by a $\calV$ edge, since
    $U = \emptyset$ implies that all edges are labelled with numbers
    from $\set{0,1,2}$. Hence, by definition,
    $P \nin \dpaths{t}{\emptyset}$ necessary and sufficient for
    $P \in \paths{t}{\emptyset}$.
  \end{enumerate}
\end{proof}
\begin{lemma}
  \label{lem:pathDiv}
  If $t \in \aptree$ and $U = \emptyset$, then $\dpaths{t}{U}$ is the least subset of
  $\prepaths t U$ satisfying conditions \ref{item:pathDiv2} -
  \ref{item:pathDiv4} of Definition~\ref{def:pathDiv}.
\end{lemma}
\begin{proof}
  Let $\calP$ be the least subset satisfying conditions
  \ref{item:pathDiv2} - \ref{item:pathDiv4} of
  Definition~\ref{def:pathDiv}. Hence,
  $\calP \subseteq \dpaths{t}{U}$. To show that
  $\calP \supseteq \dpaths{t}{U}$, we need to show that $\calP$
  satisfies \ref{item:pathDiv1} as well. To this end, we show that the
  precondition of \ref{item:pathDiv1} can never be satisfied.

  Let $n_k \in \dom t\cup\domBot t$ and
  $e_k \in \set{\lambda,\calV}\cup \setcom{i}{a_i = 0}$ for all
  $k \ge 0$. Moreover, let
  $P \stackrel{e_0}\to n_0 \stackrel{e_1}\to n_1 \stackrel{e_2}\to
  \dots \stackrel{e_m}\to n_m$ is a path in $\prepaths t U$ for each
  $m \ge 0$. We show that this assumption leads to a contradiction.

  Since $U = \emptyset$, we know that $e_k \nin \set{\lambda,\calV}$,
  and thus $e_k \in \setcom{i}{a_i = 0}$. Define the infinite sequence
  $S = \pathpos{P}\concat\seq{e_0,e_1,e_2,\dots}$. By
  Lemma~\ref{lem:nodePosSame}, $S$ is an infinite branch in
  $t$. Moreover, since $e_k \in \setcom{i}{a_i = 0}$ for all
  $k \ge 0$, we know that $S$ $\ola$-bounded. This contradicts the
  assumption that $t \in \aptree$.
\end{proof}

\begin{lemma}
  \label{lem:pathposBijection}
  The mapping
  $\theta\fcolon \paths t \emptyset \to \dom{\strictnf t}$ with
  $\theta(P) = \pathpos P$ is a bijection.
\end{lemma}
\begin{proof}
  It is easy to show by induction that
  $\pathpos{\cdot}\fcolon \prepaths t \emptyset \to \dom t\cup\domBot
  t$ is a bijection. Moreover,
  $\dom{\strictnf t} = \dom t \setminus \doms t$ by
  Proposition~\ref{prop:strictRedNormalConfl}, and
  $\paths t \emptyset = \prepaths t \emptyset \setminus \dpaths t
  \emptyset$, by
  Lemma~\ref{lem:nodePosSame}~\ref{item:nodePosSame2}. Hence, it
  suffices to show that $P \in \dpaths t \emptyset$ iff
  $\pathpos{P} \in \doms t$ for all $P \in \prepaths t \emptyset$.

  We first prove the ``only if'' direction by induction on
  $P \in \dpaths t \emptyset$. By Lemma~\ref{lem:pathDiv}, we only
  have to consider the cases \ref{item:pathDiv2}-\ref{item:pathDiv4}
  of Definition~\ref{def:pathDiv}.
  \begin{enumerate}
  \item[\ref{item:pathDiv2}] Let $P \in \prepaths t \emptyset$ such
    that $P$ ends in a node $n \in \domBot t$. Hence, $n \in \doms
    t$. Since $\pathpos P = n$ by Lemma~\ref{lem:nodePosSame}, we have
    that $\pathpos P \in \doms t$.
  \item[\ref{item:pathDiv3}] Let $P \stackrel i \to n \in \dpaths t \emptyset$ and $a_i = 0$. By induction
    hypothesis, $\pathpos{P \stackrel i} \in \doms t$. Since $\pathpos {P \stackrel i} = \pathpos
    P\concat\seq i$, we then have $\pathpos P \in \doms t$.
  \item[\ref{item:pathDiv4}] Let $P = Q \stackrel i \to n$ and
    $Q \in \dpaths t \emptyset$. By induction hypothesis,
    $\pathpos Q \in \doms t$. Since
    $\pathpos P = \pathpos Q\concat\seq i$, we then have
    $\pathpos P \in \doms t$.
  \end{enumerate}
  The converse direction follows by a similar proof by induction on
  $P \in \prepaths t \emptyset$.
\end{proof}

From the definition of paths, we can derive that the set of path
labellings of a lambda tree w.r.t.\ the empty set is the graph of the
lambda tree itself:
\begin{lemma}
  \label{lem:pathLabsEmpty}
  For all $t \in \aptree$, we have that $\pathLabs{t}{\emptyset} =
  \strictnf t$.
\end{lemma}
\begin{proof}
  By Lemma~\ref{lem:pathposBijection} and
  Proposition~\ref{prop:strictRedNormalConfl}, it suffices to show
  that $\pathlab P = t(\pathpos{P})$ for all
  $P \in \paths t \emptyset$. Let $P \in \paths t \emptyset$:
    \begin{itemize}
    \item If $t(\pathpos{P}) \in \lamsig\setminus \allpos$, then
      $\pathlab{P} = t(n)$, where $n$ is the node $P$
      ends in. By Lemma~\ref{lem:nodePosSame}, $n = \pathpos P$. Thus,
      $\pathlab{P} = t(\pathpos{P})$.
    \item If $t(\pathpos{P}) \in \allpos$, then $\pathlab{P} =
      \pathpos{Q}$, where $Q$ is the unique path that ends in $t(\pathpos{P})$.
      Thus, $\pathlab{P} = t(\pathpos{P})$
    \end{itemize}
\end{proof}

Moreover, from Lemma~\ref{lem:pathPreserve} and
Lemma~\ref{lem:pathLabsEmpty}, we can immediately derive the following
corollary:
\begin{corollary}
  \label{cor:compDevTerm}
  For each complete development $s \pato[U] t$, we have
  $\pathLabs{s}{U} = t$.
\end{corollary}
\begin{proof}
  \qquad
  $\pathLabs{s}{U} \stackrel{\text{Lemma~\ref{lem:pathPreserve}}}=
  \pathLabs{t}{\emptyset}
  \stackrel{\text{Lemma~\ref{lem:pathLabsEmpty}}}= \strictnf t = t$.\\
  The equality $\strictnf t = t$ follows from the fact that $t$ is a
  $\strictred$-normal form by the definition of complete developments.
\end{proof}

From this corollary we may derive the following two corollaries.

\begin{corollary}
  \label{cor:compDevUniqueTerm}
  Given two complete developments $s \pato[U] t_1$ and $s \pato[U]
  t_2$, we have that $t_1 = t_2$.
\end{corollary}
\begin{proof}
  Immediate from Corollary~\ref{cor:compDevTerm}
\end{proof}

\begin{corollary}
  \label{cor:compDevDiamond}
  For every pair of complete developments $S\fcolon s \pato[U] t_1$
  and $T \fcolon s \pato[V] t_2$, we find two complete developments
  $S'\fcolon t_1 \pato[\desc V S] t$ and
  $T'\fcolon t_2 \pato[\desc U T] t$.
\end{corollary}
\begin{proof}
  By Propositions~\ref{prop:descSingle} and \ref{prop:descSeq},
  $S\concat S'$ and $T\concat T'$ are complete developments of $U \cup
  V$. Hence, according to Corollary~\ref{cor:compDevUniqueTerm}, they
  $\prs$-converge to the same lambda tree.
\end{proof}

Corollary~\ref{cor:compDevUniqueTerm} shows that the final term of a
complete development is uniquely determined, no matter in which order
redexes are contracted. One can also show that descendants are
uniquely determined as well, i.e.\ given two complete developments
$S\fcolon s \pato[U] t_1$ and $T\fcolon s \pato[U] t_2$ and a set of
positions $V \subseteq \dom s$, we have that $\desc{V}{S} = \desc V
T$. This suggest the notation $\desc{V}{U}$ for the descendants of $V$
in $s$ by a complete development of $U$ in $s$.

However, there is no need to prove that this is the case since none of
our proofs depend on it. We are only interested in the final term of a
complete development: we know that a complete development of $U$
followed by a complete development of $\desc{V}{U}$ is a complete
development of $U\cup V$, and according to
Corollary~\ref{cor:compDevUniqueTerm}, the final term of such a
complete development is uniquely determined (independent of whether
$\desc{V}{U}$ is uniquely determined).

\begin{figure}
  \centering
  \begin{tikzpicture}[on grid,dots/.style={shorten
      >=.5cm,shorten <=.5cm,loosely dotted,thick}]
    \begin{scope}[node distance=2.5cm]
      \node (t0) {$t_0$}%
      node[right=of t0] (t1) {$t_1$}%
      node[right=of t1] (tb) {$t_\beta$}%
      node[right=of tb] (tb1) {$t_{\beta+1}$}%
      node[right=of tb1] (ta) {$t_\alpha$};
    \end{scope}

    \begin{scope}[node distance=1.5cm]
      \node[below=of t0] (s0) {$s_0$}%
      node[below=of t1] (s1) {$s_1$}%
      node[below=of tb] (sb) {$s_\beta$}%
      node[below=of tb1] (sb1) {$s_{\beta+1}$}%
      node[below=of ta] (sa) {$s_\alpha$};
    \end{scope}

    \draw%
    (t0)%
    edge[single step] node[midway,above] {$p_0$} (t1)%
    edge[strongly] node[midway,left] {$U_0$} (s0)%
    (t1)%
    edge[dots] (tb)%
    edge[strongly] node[midway,left] {$U_1$} (s1)%
    (tb)%
    edge[single step] node[midway,above] {$p_\beta$} (tb1)%
    edge[strongly] node[midway,left] {$U_\beta$} (sb)%
    (tb1)%
    edge[dots] (ta)%
    edge[strongly] node[midway,left] {$U_{\beta + 1}$} (sb1)%
    (ta) edge[strongly] node[midway,left] {$U_\alpha$} (sa)%

    (s0) edge[strongly] node[midway,below] {$\desc{p_0}{U_0}$} (s1)%
    (s1) edge[dots] (sb)%
    (sb) edge[strongly] node[midway,below] {$\desc{p_\beta}{U_\beta}$}
    (sb1)%
    (sb1) edge[dots] (sa)%
    ;

    \draw ($(t0)+(0,.6)$) edge[decorate,decoration=brace]
    node[midway,above=5pt] {$S$}
      ($(ta)+(0,.6)$);%
  \end{tikzpicture}
\caption{The Infinitary Strip Lemma.}
\label{fig:stripLem}
\end{figure}
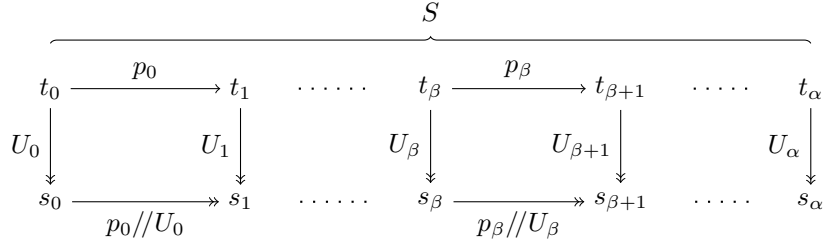

We conclude with the proof of the strip lemma:
\begin{lemma}
  \label{lem:stripLem}
  Let $S\fcolon t_0 \pato t_\alpha$ a $\prs$-convergent reduction, and
  $T\fcolon t_0 \pato[U] s_0$ a complete development of a set $U$ of
  redex occurrences in $t_0$. Then $t_\alpha$ and $s_0$ are joinable
  by a reduction $\proj{S}{T}\fcolon s_0 \pato s_\alpha$ and a
  complete development $\proj{T}{S}\fcolon t_\alpha \pato[\desc{U}{S}]
  s_\alpha$.
\end{lemma}
\begin{proof}
  Let $S = (t_\iota \to[p_\iota] t_{\iota+1})_{\iota<\alpha}$.  To
  prove this proposition, we construct the diagram shown in
  Figure~\ref{fig:stripLem}. The '$U_\iota$'s in the diagram are sets
  of redex occurrences: $U_\iota = \desc{U}{\prefix{S}{\iota}}$ for
  all $0 \le \iota \le \alpha$. In particular, $U_0 = U$. That each
  $U_\iota$ is indeed a set of redex occurrences in $t_\iota$ follows
  from Proposition~\ref{prop:residuals}. All arrows in the diagram
  represent complete developments of the indicated sets of redex
  occurrences. Particularly, in each $\iota$-th step of $S$ the redex
  at $p_\iota$ is contracted. We will construct the diagram by an
  induction on $\alpha$.

  If $\alpha=0$, the diagram is trivial. If $\alpha$ is a successor
  ordinal $\beta + 1$, then we can take the diagram for the prefix
  $\prefix{S}{\beta}$, which exists by the induction hypothesis, and
  extend it to a diagram for $S$. The necessary square to extend the
  diagram follows from Corollary~\ref{cor:compDevDiamond}.

  Let $\alpha$ be a limit ordinal. By induction hypothesis, the
  diagram exists for each proper prefix of $S$. Let $T_\iota\fcolon
  s_0 \pato s_\iota$ denote the reduction at the bottom of the diagram
  for the reduction $\prefix{S}{\iota}$ for each $\iota <
  \alpha$. Since $(T_\iota)_{\iota < \alpha}$ is a monotone sequence,
  $\proj{S}{T} = \Lub_{\iota < \alpha} T_\iota$ is a well-defined
  $\betared$-reduction. Moreover since each $T_\iota$ is
  $\prs$-convergent, $\proj{S}{T}$ is $\prs$-continuous. Hence, it is
  also $\prs$-convergent, i.e.\ there is some $s_\alpha$ such that
  $\proj{S}{T} \fcolon s_0 \pato s_\alpha$.

  Let $T^0\fcolon t_\alpha \pato[U_\alpha] s$ be a complete
  development of $U_\alpha$ in $t_\alpha$. It remains to be shown that
  $s = s_\alpha$. To this end, assume that $T^0 = (\phi_\iota\fcolon
  r_\iota \to[q_\iota] r_{\iota+1})_{\iota<\oh\alpha}$. Moreover, we
  assume the following factorisations of $T^0$: for each
  $\iota<\oh\alpha$, let $T^1_\iota, T^2_\iota$ be such that $T^0 =
  T^1_\iota \concat \seq{\phi_\iota} \concat T^2_\iota$.

  Let $p \in \dom s$. According to Lemma~\ref{lem:ancestor}, we find
  for each $\iota < \oh\alpha$ some $u_\iota \in \dom{r_\iota}$ such
  that $p \in \desc{u_\iota}{\seq{\phi_\iota}\concat T^2_\iota}$. By
  Lemma~\ref{lem:descPreserveLabel}, we have that $s(p) =
  t_\alpha(u_0)$. Hence, by
  Lemma~\ref{lem:pconvOpenPos}~\ref{item:pconvOpenPos1}, we find some
  $\beta < \alpha$ such that $t_\iota(u_0) = t_\alpha(u_0)$ and
  $p_\iota \not\le u_0$ for all $\beta \le \iota<\alpha$.

  By Lemma~\ref{lem:descAltChar}, we know that there must be some
  $\oh\beta <\oh\alpha$ such that $q_\iota \not\le u_\iota$ for all
  $\oh\beta \le \iota < \oh\alpha$. Consequently, we may assume
  w.l.o.g.\ that $\beta$ is chosen large enough such that
  $s_\iota(u_0) = t_\iota(u_0)$ for all $\beta \le \iota <
  \alpha$. Moreover, there must be an upper bound $\beta \le \gamma <
  \alpha$ such that, $s_\iota(p) = t_\iota(u_0)$ and if $v_\iota \in
  \desc{p_\iota}{U_\iota}$, then $v_\iota \not\le p$ for all
  $\gamma \le \iota < \alpha$. Then, if $\proj{S}{T}$ is closed, we
  trivially have that $s_\alpha(p) = t_\gamma(u_0)$. Otherwise, if
  $\proj{S}{T}$ is open, we can apply
  Lemma~\ref{lem:pconvOpenPos}~\ref{item:pconvOpenPos2} to obtain that
  $s_\alpha(p) = t_\gamma(u_0)$. Combining the equalities we have
  found, we obtain that $s_\alpha(p) = t_\gamma(u_0) = t_\alpha(u_0) =
  s(p)$.

  By a similar argument we can show that $\domBot{s} \subseteq
  \domBot{s_\alpha}$. Consequently, we have that $s = s_\alpha$.
\end{proof}

\begin{thmcpy}{Lemma~\ref{lem:stripLem'} (Infinitary Strip Lemma)}
  Given reductions $S\fcolon s \pato[\betasred] t_1$ and $T \fcolon s
  \fto*[\betasred] t_2$, there are reductions $S'\fcolon t_1
  \pato[\betasred] t$ and $T'\fcolon t_2 \pato[\betasred] t$, provided
  $\ola \in \set{001,101,111}$.
\end{thmcpy}
\begin{proof}
  This follows by iterating Lemma~\ref{lem:stripLem} for the special
  case that $T$ a complete development of a single redex occurrence.
\end{proof}

\section{Weak Convergence}
\label{sec:weak-convergence}

We briefly give the definition and some of the properties of weak
convergence. To distinguish this variant of convergence from the one
in the main text of this paper, we refer to the latter explicitly as
strong ($\mrs$-/$\prs$-) convergence.

\begin{definition}[weak convergence]
  An $R$-reduction $S = (t_\iota \to_{R} t_{\iota+1})_{\iota<\alpha}$
  is called \emph{weakly $\mrs$-continuous} resp.\
  \emph{$\prs$-continuous} if, for all limit ordinals $\gamma <
  \alpha$, we have that $\lim_{\iota\limto\gamma} t_\iota = t_\gamma$
  resp.\ $\liminf_{\iota\limto\gamma} t_\iota = t_\gamma$; $S$ is said
  to \emph{weakly $\mrs$-converge} resp.\ \emph{$\prs$-converge} to
  $t$, denoted $S\colon t_0 \mwato_R t$ resp.\ $S\colon t_0 \pwato_R
  t$, if it is weakly $\mrs$-continuous and $\lim_{\iota\limto\alpha}
  t_\iota = t$ resp.\ weakly $\prs$-continuous and
  $\liminf_{\iota\limto\alpha} t_\iota = t$.
\end{definition}

Intuitively, a reduction is continuous if it is well-defined at limit
ordinal indices, and a reduction is convergent if it additionally has
a final result. Since the partially ordered set $(\aptree, \talebot)$
forms a complete semilattice, every weakly $\prs$-continuous reduction
also weakly $\prs$-converges. In contrast, however, a weakly
$\mrs$-continuous reduction is not necessarily weakly
$\mrs$-convergent:

\begin{example}
  Given $I = \abstree x  \vartree x$ and $t = \abstree x I\, \vartree
  x\, \vartree x$, consider the $\betared$-reduction $t\, t
  \to[\betared] I\,t\,t \to[\betared] t\,t \to[\betared] \dots$, which
  is trivially $\mrs$- and $\prs$-continuous. Since the subtree at
  position $\seq 1$ alternates between $t$ and $I\, t$, the reduction
  does not weakly $\mrs$-converge (for any $\ola$); but it does weakly
  $\prs$-converge to $\bot\, t$ if $a_1 = 1$ (i.e.\ position $\seq 1$
  is non-strict) and to $\bot$ if $a_1 = 0$ (i.e.\ $\seq 1$ is
  strict).
\end{example}

Transferring the results from Section~\ref{sec:ideal-completion} to
weak convergence is trivial as these notions of convergence are
directly based on the modes of convergence of the underlying
structures:
\begin{theorem}
  \label{thr:weakConvEq}
  For each $R$-reduction $S$, we have the following:
  \begin{enumerate}[(i)]
  \item $S\colon s \mwato_R t  \implies  S\colon s \pwato_R
    t$.
    \label{item:weakConvEq1}
  \item $S\colon s \pwato_R t  \implies  S\colon s
    \mwato_R t$, provided $S$ and $t$ are total.
    \label{item:weakConvEq2}
  \end{enumerate}
\end{theorem}
\begin{proof}
  \ref{item:weakConvEq1} follows from Theorem~\ref{thr:limLiminf}~\ref{item:limLiminf1},
  \ref{item:weakConvEq2} follows from Theorem~\ref{thr:limLiminf}~\ref{item:limLiminf1}.
\end{proof}
Note that for \ref{item:weakConvEq2} it is not enough to require that
the reduction $S$ is total, since $t$ is not necessarily a part of $S$
but may only arise as a limit inferior of the lambda trees in $S$.

\begin{corollary}
  $S\colon s \mwato t$ iff $S\colon s \pwato t$ whenever $S$ and $t$ are
  total.
\end{corollary}

Another observation is that the strong notions of convergence indeed
imply their weak counterpart -- however, with a small caveat for
$\prs$-convergence. To prove this, we need to prove the following
observation:
\begin{lemma}
  \label{lem:lamMax}
  Each $t\in \itree$ is maximal in $(\iptree,\talebot)$.
\end{lemma}
\begin{proof}
  Let $t\in \itree$ and $s\in \iptree$ with $t \talebot s$. We prove
  that $s = t$ by showing that $\dom{s} \subseteq \dom{t}$ by
  induction on the length of positions.  If $\emptyseq \in \dom{s}$,
  then $\emptyseq \in \dom{t}$ as $t\in \itree$. If $p\concat\seq 0
  \in \dom{s}$, then $s(p) = \lambda$. Since $t \talebot s$ and, by
  induction hypothesis, $p \in \dom{t}$, we know that $t(p) =
  \lambda$. As $t\in \itree$, we can conclude that $p \concat\seq 0
  \in \dom{t}$. The cases $p\concat\seq 1$ and $p \concat\seq 2$
  follow analogously.
\end{proof}

\begin{lemma}
  \label{lem:strongWeakConv}
  For each $R$-reduction $S$, we have the following:
  \begin{enumerate}[(i)]
  \item $S\colon s \mato t$ implies $S\colon s \mwato t$.
    \label{item:strongWeakConv1}
  \item $S\colon s \pato t$ implies $S\colon s \pwato t$, provided $S$
    and $t$ are total.
    \label{item:strongWeakConv2}
  \end{enumerate}
\end{lemma}
\begin{proof}
  \ref{item:strongWeakConv1} follows immediately from the definition
  $\mrs$-convergence. For \ref{item:strongWeakConv2}, we use the fact
  that weak/strong $\prs$-convergence on lambda trees is an
  instantiation of the abstract notion of weak/strong
  $\prs$-convergence from \cite{bahr10rta}. Proposition 6.5 from
  \cite{bahr10rta} states that (strong) $\prs$-convergence implies
  weak $\prs$-convergence if the lambda trees in $S$ and the lambda
  tree $t$ is maximal w.r.t.\ $\talebot$, which follows from
  Lemma~\ref{lem:lamMax}.
\end{proof}

\section{Direct Proof of Correspondence}
\label{sec:direct-approach}

In this section we prove directly that the there is a one-to-one
correspondence between the ideal completion $(\iplam,\subseteq)$ of
$(\plam,\alebot)$ and the metric completion $(\mplam,\ddac)$ of
$(\plam,\dda)$. To this end we use the meta theory of
Majster-Cederbaum and Baier~\cite{majster-cederbaum96tcs}.

The first step is to show that the metric $\dda$ may be canonically
derived from the partial order $\alebot$ by what Majster-Cederbaum and
Baier call a weight, which in our case will be the height of lambda
terms:
\begin{definition}[$\ola$-height]
  The \emph{$\ola$-height} $\ahgt M$ of a term $M\in\plam$, is
  \[
  \ahgt{M} = \min\setcom{d<\omega}{\forall p \in \pos{M}. \adepth{p} <
    d}
  \]
  Instead of $111$-height and $\ahgt[111]\cdot$, we also use height
  and $\hgt\cdot$, respectively. For each $d < \omega$, define
  \[
  \andown{d}{M} = \setcom{N \alebot M}{\ahgt{N} \le d}
  \]
\end{definition}
Alternatively, we may characterise the $\ola$-height of a term as
follows.
\begin{lemma}
  \label{lem:heightRec}
  For each $M_1,M_2\in\plam$, we have the following:
  \begin{align*}
    \ahgt{\bot} &= 0 \qquad \ahgt{x} = 1\\
    \ahgt{M_1 M_2} &= \max\set{1,\ahgt{M_1} + a_1,\ahgt{M_2} + a_2}\\
    \ahgt{\lambda x.M_1} &= \max\set{1,\ahgt{M_1} + a_0}
  \end{align*}
\end{lemma}
\begin{proof}
  Follows straightforwardly from the definition.
\end{proof}

\begin{lemma}
  \label{lem:alebotPos}
  For all $M,N \in \plam$ with $M\alebot N$, we have that $\pos M
  \subseteq \pos N$.
\end{lemma}
\begin{proof}
  We define the relation $\preceq$ by $M \preceq N$ iff $\pos M
  \subseteq \pos N$, and show that $\preceq$ satisfies the condition
  in Definition~\ref{def:alebot}. Since $\alebot$ is the least such
  relation, the lemma follows.

  Since $\subseteq$ is a preorder, so is $\preceq$. $\bot \preceq M$
  holds for all $M$ since $\pos{\bot} = \emptyset$. Let $M \preceq
  N$. If $p \in \pos{\lambda x . M}$ then $p = \seq 0 \concat q$ with
  $q \in \pos{M}$. By $M \preceq N$, we have that $q \in \pos N$ and
  thus $p = \seq 0 \concat q \in \pos {\lambda x . N}$. Hence,
  $\lambda x.M \preceq \lambda . N$. The remaining two closure
  properties follow similarly.
\end{proof}

Moreover, we have that the $\ola$-height satisfies the condition of a
weight according to Majster-Cederbaum and
Baier~\cite{majster-cederbaum96tcs}:
\begin{lemma}
  For each $M \in \plam$, we have that
  \begin{enumerate}[(i)]
  \item\label{item:heightWeight1} $\ahgt{M} = 0$ \quad iff \quad $M = \bot$.
  \item\label{item:heightWeight2} $M \alebot N$ \quad implies \quad $\ahgt{M} \le \ahgt{N}$.
  \item\label{item:heightWeight3} For each $d < \omega$,
    $\andown{d}{M}$ has a greatest element.
  \end{enumerate}
\end{lemma}
\begin{proof}
  \ref{item:heightWeight1} follows immediately from the definition;
  \ref{item:heightWeight2} follows from Lemma~\ref{lem:alebotPos}.

  For \ref{item:heightWeight3}, we construct by induction for each
  $M\in\plam$ a term $M^{d}$ that is the greatest element of
  $\andown{d}{M}$.  If $d = 0$, then $M^d$ is obviously $\bot$. In the
  following we assume that $d > 0$. The cases $M =x$ and $M=\bot$ are
  trivial.

  If $M = M_1 M_2$, then we may assume, by induction hypothesis, terms
  $M_i^{d_i}$ for $d_i = d - a_i$. If $M_1^{d_1}M_2^{d_2} \alebot M_1
  M_2$ then define $M^d = M_1^{d_1}M_2^{d_2}$; otherwise $M^d =
  \bot$. In either case, $M^d \alebot M$ and $\ahgt{M^d} \le d$, i.e.\
  $M^d \in \andown{d}{M}$.

  In order to show that $M^d$ is the greatest element in
  $\andown{d}{M}$, we assume some $N \in \andown{d}{M}$ and show that
  then $N \alebot M^d$. We have that $N \alebot M$ and thus either $N
  = \bot$, in which case $N \alebot M^d$ follows immediately, or $N =
  N_1 N_2$ with $N_i \alebot M_i$. In the latter case we then have,
  according to \ref{item:heightWeight2}, that $\ahgt{N_i} \le
  \ahgt{M_i} \le d_i$, i.e.\ $N_i \in \andown{d_i}{M_i}$. By induction
  hypothesis, we thus have that $N_i \alebot M_i^{d_i}$. Note that
  this means that if $M_i^{d_i} = \bot$, then $N_i = \bot$. Since
  $N\alebot M$, this then implies that $M_i = \bot$ or $a_i = 1$. In
  sum, we thus have that $M_1^{d_1}M_2^{d_2} \alebot M$ and therefore
  $M^d = M_1^{d_1}M_2^{d_2}$. It thus remains to be shown that $N_1
  N_2 \alebot M^{d_1}_1 M^{d_1}_1$. To this end, we show that $N_1 N_2
  \alebot M^{d_1}_1 N_2$; $M^{d_1}_1 N_2 \alebot M^{d_1}_1M^{d_2}_2$
  follows analogously.

  If $a_1 = 1$ or $N \neq \bot$, then $N_1 N_2 \alebot M^{d_1}_1 N_2$
  follows immediately from $N_1 \alebot M^{d_1}_1$. If $a_1=0$ and
  $N_1 = \bot$, then $M_1 = \bot$ follows from $N \alebot
  M$. Consequently, $M_1 = \bot$ and $N_1 N_2 \alebot M^{d_1}_1 N_2$
  follows by reflexivity.

  The case $M = \lambda x. M'$ follows analogously.
\end{proof}
In fact, $\ola$-height is a finite weight since, by definition,
$\ahgt{M} < \omega$ for all lambda terms $M$.

According to Majster-Cederbaum and
Baier~\cite{majster-cederbaum96tcs}, the measure provided by
$\ahgt\cdot$ can thus be used to define an ultrametric $\dd$ on
$\plam$ as follows:
\[
\dd(M,N) = \Glb \setcom{2^{-d}}{\andown{d}{M} =
  \andown{d}{N}}.
\]
The following two lemmas, will help us show that $\dd$ and $\dda$
coincide.
\begin{lemma}
  \label{lem:conflictAlebot}
  If all conflicts of $M,N \in \plam$ have an $\ola$-depth of at least
  $d$, then $M' \alebot N$ for all $M' \in \andown d M$.
\end{lemma}
\begin{proof}
  We proceed by induction on $M'$. If $d = 0$, then $\ahgt{M'} \le d$
  implies that $M' = \bot$ and thus $M' \alebot N$ follows. In the
  following we thus assume $d > 0$. The case $M= \bot$ is trivial. If
  $M' = x$, then also $M = x$. Since $d > 0$, $\emptyseq$ is not a
  conflict of $M,N$, which means that $N = x$. Hence, $M' \alebot N$
  by reflexivity.

  If $M' = M'_1 M'_2$, then $M =M_1 M_2$ with $M'_i \alebot
  M_i$. Moreover, since $d > 0$, $\emptyseq$ is not a conflict of
  $M,N$ and thus $N$ is of the form $N= N_1N_2$ and all conflicts of
  $M_i,N_i$ have $\ola$-depth $\ge d - a_i$. Moreover, because, by
  Lemma~\ref{lem:heightRec}, $\ahgt{M'_i} \le \ahgt{M'} - a_i \le d -
  a_i$, we may apply the induction hypothesis to obtain that $M'_i
  \alebot N_i$. In order to show that $M' \alebot N$, we show that
  $M'_1M'_2 \alebot N_1 M'_2$; $N_1 M'_2 \alebot N_1 N_2$ follows
  likewise. If $a_1 = 1$ or $M'_1 \neq \bot$, then $M'_1M'_2 \alebot
  N_1 M'_2$ follows immediately from $M'_1\alebot N_1$. Otherwise, if
  $a_1=0$ and $M'_1 = \bot$, then $M_1 = \bot$ since $M' \alebot
  M$. Hence, also $N_1 = \bot$, since otherwise, $\seq{1}$ is a
  conflict of $M,N$ of $\ola$-depth $0$. Consequently, $M'_1M'_2
  \alebot N_1 M'_2$ follows by reflexivity.

  The case $M' = \lambda x. M_1'$ follows analogously.
\end{proof}

\begin{lemma}
  \label{lem:alebotHeightOne}
  For all $M \in \plam \setminus\set\bot$ there is some $N \alebot M$
  with $\ahgt N = 1$.
\end{lemma}
\begin{proof}
  We proceed by induction on $M$. If $M = x$, set $N = x$.

  Let $M = M_1 M_2$. If $M_i = \bot$ or $a_i = 1$, set $N_i =
  \bot$. Otherwise, there are, by induction hypothesis, $N_i \alebot
  M_i$ with $\ahgt{N_i} = 1$ and $N_i \neq \bot$. In either case set
  $N = N_1N_2$. Moreover, we have $N = N_1N_2 \alebot M_1N_2 \alebot
  M_1M_2 = M$ and 
  \[
  \ahgt{N} = \max\set{1,\ahgt{N_1} + a_1,\ahgt{N_2} +
    a_2} = 1.
  \]

  The case $M = \lambda x .M'$ follows analogously.
\end{proof}

\begin{lemma}
  \label{lem:conflictAlebot2}
  If $p$ is a conflict of $M,N$ with $\ola$-depth $d$, then there is
  some $M' \alebot M$ with $\ahgt{M'} = d+1$ and $M' \nalebot N$,
  or vice versa there is some $N' \alebot N$ with $\ahgt{N'} = d+1$
  and $N' \nalebot M$
\end{lemma}
\begin{proof}
  We proceed by induction on $p$.

  Let $p = \emptyseq$. Since $\emptyseq$ a conflict, $M,N$ cannot be
  both $\bot$. W.l.o.g.\ assume that $M\neq\bot$. By
  Lemma~\ref{lem:alebotHeightOne}, there is some $M' \alebot M$ of
  $\ola$-height $1$. Hence, $M'$ and $M$ are either both the same
  variable, both applications or both abstractions. Hence, $p$ is also
  a conflict of $M', N$, which implies that $M' \nalebot N$.

  Let $p = \seq i \concat q$. We assume that $i = 1$; the cases for $i
  \in \set{0,2}$ follow analogously. Since $p$ is a conflict of $M,N$,
  we know that $M = M_1M_2$, $N = N_1N_2$ and $q$ is a conflict of
  $M_1,N_1$. By induction hypothesis, we can assume w.l.o.g.\ that
  there is some $M_1' \alebot M_1$ of $\ola$-height $d + 1 - a_1$ and
  with $M'_1 \nalebot N_1$. From the latter we deduce that $M'_1 \neq
  \bot$ and thus $M_1' M_2 \alebot M$. If $a_2 = 1$ or $M_2 = \bot$,
  then set $M'_2 = \bot$. We get that $M_1'M_2' \alebot
  M_1'M_2$. Otherwise, if $a_2 = 0$ and $M_2 \neq\bot$, then there is,
  according to Lemma~\ref{lem:alebotHeightOne}, some $M_2' \neq \bot$
  of $\ola$-height $1$ with $M_2' \alebot M_2$. In either case, we
  have that $M'_1M'_2 \alebot M'_1M_2 \alebot M$ and that $\ahgt{M_2'}
  + a_2 \le 1$. The latter implies, by Lemma~\ref{lem:heightRec}, that
  $\ahgt{M_1'M_2'} = \max\set{\ahgt{M_1'} + a_1, 1} = \max\set{d + 1,
    1} = d + 1$

\end{proof}

Finally, we can prove that two metrics $\dd$ and $\dda$ coincide:
\begin{proposition}
  For all $M,N \in \plam$, we have that
  \[
  \dda(M,N) = \Glb\setcom{2^{-d}}{\andown{d}{M} = \andown{d}{M}}
  \]
\end{proposition}
\begin{proof}
  We write $\dd(M,N)$ to denote the right-hand side of the above
  equation. If $\dda(M,N) = 0$, then $M,N$ have no conflicts. By
  Lemma~\ref{lem:conflictAlebot}, we then have that $\andown{d}{M} =
  \andown{d}{N}$ for all $d < \omega$. Hence, $\dd(M,N) =0$, too.

  Otherwise, $\dda(M,N) = 2^{-d}$ such that there is a conflict $p$ of
  $M,N$ with $\ola$-depth $d$ and each conflict of $M,N$ has
  $\ola$-depth at least $d$. The former implies, by
  Lemma~\ref{lem:conflictAlebot2}, that $\andown{e}{M} \neq
  \andown{e}{N}$ for all $e > d$ and the latter implies, by
  Lemma~\ref{lem:conflictAlebot}, that $\andown{d}{M} =
  \andown{d}{N}$. Hence, $\dd(M,N) = 2^{-d}$ as well.
\end{proof}

The ultrametric induced by $\ahgt\cdot$ can be canonically extended to
the ideal completion $\iplam$ of $(\plam,\alebot)$:
\begin{definition}
  For each set $S \subseteq \plam$ and $d < \omega$, define
  \[\anidown{d}{S} = \setcom{M \in S}{\ahgt{M}\le d}.\]
  Define the distance measure $\ddi$ on $\iplam$ as follows:
  \[
  \ddi(I,J) = \Glb\setcom{2^{-d}}{\anidown{d}{I} = \anidown{d}{J}}
  \]
\end{definition}

According to Majster-Cederbaum and
Baier~\cite{majster-cederbaum96tcs}, $\ddi$ is an
ultrametric. Moreover, they also show that the metric completion
$(\ddac,\mplam)$ is isometric to ideal completion $\iplam$ endowed
with the metric $\ddi$ whenever, each $\anidown{d}{I}$ of an ideal $I$
is finite.
\begin{lemma}
  \label{lem:posIncl}
  For each $M,N \in \plam$ with $M \alebot N$ and each $p \in \pos{M}$
  and $p\concat\seq i\in \pos{N}$ with $a_i = 0$, we have that $p
  \concat \seq i \in \pos M$.
\end{lemma}
\begin{proof}
  We proceed by induction on $p$. Note that $M \neq\bot$ since
  $\pos{M} \neq \emptyset$. Let $p = \emptyseq$. If $i = 0$, then $N =
  \lambda x . N_1$ and thus $M = \lambda x . M_1$. Since $\seq 0 \in
  \pos N$, we know that $N_1 \neq \bot$. Because $a_0 = 0$, this
  implies that $M_1 \neq \bot$. Hence, $\seq 0 \in \pos M$. The cases
  for $i \in \set{1,2}$ follow analogously.

  Let $p = \seq j\concat q$ and assume $j = 0$. Then $M = \lambda x
  . M_1$, $N = \lambda x . N_1$, and $M_1 \alebot N_1$. Since $q \in
  \pos{M_1}$ and $q\concat\seq i \in \pos{N_1}$, we may apply the
  induction hypothesis to obtain that $q \concat\seq i \in
  \pos{M_1}$. Consequently, $\seq j\concat q \concat\seq i \in
  \pos{M_1}$. The cases for $j\in\set{1,2}$ follow likewise.
\end{proof}

\begin{definition}
  For each $I \subseteq \plam$ define the set of positions $\pos{I}$
  of $I$ as $\bigcup_{M\in I} \pos{M}$.
\end{definition}

A set $I$ in $(\plam,\alebot)$ is said to be \emph{finitely bounded}
if for each $M,N \in I$ there is some $\oh M \in \plam$ with $M,N
\alebot \oh M$.

\begin{proposition}
\label{prop:idealFinPos}
Given a finitely bounded set $I$ in $(\plam,\alebot)$, $I$ is finite
iff $\pos{I}$ is finite.
\end{proposition}
\begin{proof}
  The ``only if'' direction is trivial. For the converse direction,
  assume that $I$ is infinite. We show that $I_d = \setcom{M \in
    I}{\hgt M \le d}$ is finite for each $d < \omega$. From this we
  can then deduce that for each $d < \omega$, the set $I\setminus I_d$
  is non-empty, i.e.\ there is an $M \in I$ with a height greater than
  $d$. Consequently, for each $d$, there is a position $p$ with height
  $d$ in $\pos I$, i.e.\ $\pos I$ is infinite.

  We show the above claim that $I_d$ is finite for any finitely
  bounded set $I$ by induction on $d$. The case $d = 0$ is trivial,
  since only $\bot$ has depth $0$. Let $d > 0$. We show that $I_d$
  contains only finitely many different variables, applications and
  abstractions. For each two variables $x, y\in I_d$, we find some $M$
  with $x,y \alebot M$ since $I$ finitely bounded. Hence $x = y$. Let
  $J = \setcom{M}{\exists N\in\plam\fcolon M\,N \in I}$ and $K =
  \setcom{N}{\exists M\in\plam\fcolon M\,N \in I}$. Then also $J$ and
  $K$ are finitely bounded. For each abstraction $M N \in I_d$, we
  know that $M \in J_{d-1}$ and $N \in K_{d-1}$. Since, by induction
  hypothesis, both $J_{d-1}$ and $K_{d-1}$ are finite, there are also
  only finitely many applications in $I_d$. The same argument applies
  to abstractions.
\end{proof}

\begin{lemma}
  \label{lem:anidownFin}
  For each $I \in \iplam$ and $d < \omega$, $\anidown{d}{I}$ is a
  finite set.
\end{lemma}
\begin{proof}
  Assume that the lemma is not true, i.e.\ there is some $I \in
  \iplam$ and $d < \omega$ such that $\anidown{d}{I}$ is an infinite
  set. According to Proposition~\ref{prop:idealFinPos}, the set
  $\pos{\anidown{d}{I}}$ is infinite, too.
  
  Since there are only finitely many different sequences over
  $\set{0,1,2}$ of a given finite length, there must be an infinite
  sequence $q\colon \omega \to \set{0,1,2}$ such that $P_{i} =
  \setcom{p \in \pos{\anidown{d}{I}}}{p_i \le p}$ is finite for all $i
  < \omega$, where $p_i$ is the prefix of $q$ of length $i$.

  Note that since each $p_i \in \pos{\anidown{d}{I}}$, we know that
  $\adepth{p_i} \le d$. Hence, there must be a $k < \omega$ such that
  $a_{q(i)} = 0$ for all $i \ge k$.

  Let $M \in \anidown{d}{I}$ with $p_k \in \pos{M}$. We show by
  induction on $i$ that $p_i \in \pos{M}$ for all $i \ge k$, which is
  impossible as $\pos{M}$ is finite. Thus our assumption that the
  lemma is not true is false.

  The case $i = k$ is trivial. Let $i \ge k$ and $p_i \in
  \pos{M}$. Moreover, let $N$ be a term with $p_{i+1} \in \pos
  N$. Since $I$ is directed, we find a term $M' \in I$ with $M,N
  \alebot M'$. By Lemma~\ref{lem:alebotPos}, we thus also have that
  $p_{i+1} \in \pos{M'}$. Since $p_{i+1} = p_i\concat\seq{a_{q(i+1)}}$
  with $a_{q(i+1)} = 0$, we can derive from $p_i \in \pos M$ that
  $p_{i + 1} \in \pos M$ according to Lemma~\ref{lem:posIncl}.
\end{proof}

\begin{proposition}
  The pair $(\ddac,\mplam)$ is isometric $(\ddi,\iplam)$.
\end{proposition}
\begin{proof}
  This proposition follows from Theorem~3.16 of Majster-Cederbaum and
  Baier~\cite{majster-cederbaum96tcs} with Lemma~\ref{lem:anidownFin}
  as $\ahgt\cdot$ is a finite weight.
\end{proof}
